\documentclass[11pt]{article}
\usepackage{fullpage}
\usepackage[utf8]{inputenc}
\usepackage{booktabs} % For formal tables
\usepackage[linesnumbered,noend,ruled,noline]{algorithm2e} % For algorithms
\usepackage{newpxtext}
\usepackage{amsmath,amsthm,graphicx,url}
\usepackage{amssymb}
\usepackage{threeparttable}
\usepackage{tikz}
\usetikzlibrary{shapes.geometric, arrows.meta, positioning}
\usetikzlibrary{calc}
\usepackage{subcaption}
\usepackage{mathtools}
\usepackage[normalem]{ulem}
\usepackage{soul}
\usepackage{graphicx}
\usepackage{booktabs}
\usepackage{tablefootnote}
\usepackage[numbers,sort&compress]{natbib}
\usepackage{enumitem}
\usepackage{comment}
\usepackage[dvipsnames]{xcolor}
\DeclareMathOperator*{\argmax}{arg\,max}

\usepackage{bm}
\usepackage{footnote}
\usepackage{hyperref}
\usepackage[nameinlink, capitalise]{cleveref-usedon}
\usepackage{color-edits}
\DeclareUnicodeCharacter{0302}{}
\definecolor{darkgreen}{rgb}{0.0,0.4,0.0}
\addauthor{jack}{red}

\addauthor{xt}{darkgreen}

\addauthor{aviad}{blue}

\newcommand{\are}[1]{\aviadedit{#1}}

 \usepackage[ruled, noline]{algorithm2e} % For algorithms

 \SetKwBlock{Initialize}{Initialize:}{} \SetKwFor{From}{from}{do}{}
 \SetKwFor{ForAllInParallel}{for all}{in parallel do}{}

\SetCommentSty{mycommfont}
\SetKw{Return}{return}

\renewcommand\citep[1]{\citealp{#1}}
\renewcommand\citet[1]{\citeauthor{#1} [\citealp{#1}]}

\newtheorem{theorem}{Theorem}[section]
\newtheorem{lemma}[theorem]{Lemma}

\newtheorem{corollary}[theorem]{Corollary}
\newtheorem{proposition}[theorem]{Proposition}

\newtheorem{observation}[theorem]{Observation}
\newtheorem*{theorem*}{Theorem}
\newtheorem*{corollary*}{Corollary}

\theoremstyle{definition}
\newtheorem{definition}[theorem]{Definition}

\newtheorem{remark}[theorem]{Remark}

\newtheorem{assumption}[theorem]{Assumption}

\AddToHook{env/lemma/begin}{\crefalias{theorem}{lemma}}
\AddToHook{env/conjecture/begin}{\crefalias{theorem}{conjecture}}
\AddToHook{env/corollary/begin}{\crefalias{theorem}{corollary}}
\AddToHook{env/proposition/begin}{\crefalias{theorem}{proposition}}
\AddToHook{env/claim/begin}{\crefalias{theorem}{claim}}
\AddToHook{env/observation/begin}{\crefalias{theorem}{observation}}
\AddToHook{env/definition/begin}{\crefalias{theorem}{definition}}
\AddToHook{env/example/begin}{\crefalias{theorem}{example}}
\AddToHook{env/remark/begin}{\crefalias{theorem}{remark}}
\AddToHook{env/question/begin}{\crefalias{theorem}{question}}
\AddToHook{env/condition/begin}{\crefalias{theorem}{condition}}
\AddToHook{env/assumption/begin}{\crefalias{theorem}{assumption}}

\crefname{theorem}{Theorem}{Theorems}
\crefname{lemma}{Lemma}{Lemmas}
\crefname{conjecture}{Conjecture}{Conjectures}
\crefname{corollary}{Corollary}{Corollaries}
\crefname{proposition}{Proposition}{Propositions}
\crefname{claim}{Claim}{Claims}
\crefname{observation}{Observation}{Observations}
\crefname{definition}{Definition}{Definitions}
\crefname{example}{Example}{Examples}
\crefname{remark}{Remark}{Remarks}
\crefname{question}{Question}{Questions}
\crefname{condition}{Condition}{Conditions}
\crefname{assumption}{Assumption}{Assumptions}

\DeclareMathOperator*{\E}{\mathbb{E}}
\newcommand{\R}{\mathbb{R}}
\newcommand{\gft}{\mathtt{GFT}}
\newcommand{\GFT}{\mathtt{GFT}}

\newcommand{\cM}{\mathcal M}

\numberwithin{equation}{section}

\newcommand{\btprice}{\ifmmode\mathrm{BT}\else\textrm{BT}\fi}

\newcommand{\sdma}{\ifmmode\mathrm{MA}\else\textrm{MA}\fi}

\SetAlFnt{\small}
\SetAlCapFnt{\small}
\SetAlCapNameFnt{\small}
\SetAlCapHSkip{0pt}
\IncMargin{-\parindent}

\definecolor{commentc}{rgb}{0.1, 0.3, 0.9}

\SetArgSty{textnormal}
\crefname{algorithm}{Auction}{Auctions}

\definecolor{linkc}{rgb}{0.7, 0.2, 0.3}
\definecolor{citec}{rgb}{0.2, 0.3, 0.7}
\definecolor{urlc}{rgb}{0.2, 0.4, 0.3}
\hypersetup{
    colorlinks=true,
    linkcolor=linkc,
    citecolor=citec,
    urlcolor=urlc
}

\newcommand{\Var}{\text{Var}}

\newcommand{\mainresultBT}{
 Consider any Bayesian multidimensional bilateral trade with an XOS buyer and a additive seller over independent items. The mechanism that randomizes between the seller-optimal mechanism and buyer-optimal mechanism with equal probability satisfies Bayesian Incentive Compatible (BIC), Interim Individually Rational (IIR), and ex-ante Weakly Budget Balanced (WBB), and its expected gains-from-trade is at least a $1/44$  fraction of the first-best expected gains-from-trade. 
}

\newcommand{\mainresultBTsub}{
Let $\cM_E$ be the cost-price two-part tariff with entry fee \cref{def:warmup_2pt}. Let $\cM_S^{\mathrm{tail}}$ and $\cM_B^{\mathrm{copy}}$ be the simultaneous one-item posted-price menus derived in \cref{thm:warmup_tail_seller_menu,thm:warmup_tail_buyer_menu}. Consider the simple randomized protocol $\cM_{\mathrm{rand}}^{\mathrm{simple}}$: with probability $1/2$, delegate pricing power to the seller, who \emph{can} randomize uniformly between $\cM_E$ and $\cM_S^{\mathrm{tail}}$. With probability $1/2$, delegate pricing power to the buyer, who \emph{can} run $\cM_B^{\mathrm{copy}}$. This protocol achieves expected GFT is at least $1/88$ fraction of the first-best expected gains-from-trade.
}

\newcommand{\mainresult}{
Under Assumption~\ref{ass:standing}, consider any Bayesian multidimensional instance with $n$ XOS buyers and one additive
seller over independent items. The mechanism $\mathcal \cM^{\mathrm{final}}$ is BIC, interim IR, and ex-ante weakly budget balanced. Moreover, it satisfies:$$\gft(\mathcal \cM^{\mathrm{final}}) \ge \alpha\,\gft^*,$$where $\alpha \approx 1/1033$.
}
\newcommand{\ASPE}{\text{ASPE}}

\usepackage{amsmath,amssymb,amsthm,mathtools,bbm,mathrsfs}
\usepackage{mathpazo}
\title{Approximately Efficient Multidimensional Bilateral Trade}
\author{%
\begin{tabular}{ccc}
\begin{minipage}[t]{0.3\textwidth}\centering
Aviad Rubinstein\\
\normalsize Stanford University\\
\normalsize\href{mailto:aviad@stanford.edu}{aviad@stanford.edu}
\end{minipage}
&
\begin{minipage}[t]{0.3\textwidth}\centering
Xizhi Tan\\
\normalsize Stanford University\\
\normalsize\href{mailto:xizhi@stanford.edu}{xizhi@stanford.edu}
\end{minipage}
&
\begin{minipage}[t]{0.3\textwidth}\centering
Zixin Zhou\\
\normalsize Stanford University\\
\normalsize\href{mailto:jackzhou@stanford.edu}{jackzhou@stanford.edu}
\end{minipage}
\end{tabular}%
}
\date{}

\begin{document}

\maketitle

\thispagestyle{empty}
\setcounter{page}{0}

\begin{abstract}
A central challenge in mechanism design is to develop truthful trade mechanisms that maximize the expected gains-from-trade (GFT) in two-sided markets. Because achieving the full GFT is generally impossible, the literature has focused on constant-factor approximations—a notoriously difficult problem even in simple settings. It was only recently that a breakthrough result by \cite{dmsw22} achieved a constant-factor approximation for single-item bilateral trade. The same guarantee was later extended to single-dimensional matching markets with general downward-closed constraints \cite{BRTW26}. 

Most existing results, however, are limited to single-dimensional agents. A notable multi-dimensional exception is \cite{CaiGMZ21}. They considered a market with one constrained-additive buyer and $n$ single-dimensional sellers and provided a mechanism that achieves a $\log^2(n)$ approximation to the second-best GFT, i.e., the maximum expected GFT theoretically achievable by any mechanism satisfying Bayesian Incentive Compatibility (BIC), Interim Individual Rationality (IIR), and ex-ante Weak Budget Balance (WBB).

In this paper, we study multi-dimensional bilateral trade problem where both sides of the market are multi-dimensional. We start with one buyer with XOS valuation and one seller with an additive cost function. We then generalize to a market with $n$ XOS buyers and one additive seller. Assuming independent items' values and costs, in both settings we propose simple mechanisms that are BIC, IIR, and ex-ante WBB, while achieving a constant fraction of the optimal (first-best) expected GFT.

\end{abstract}

\newpage
\tableofcontents

\newpage

\setcounter{page}{1}

\begin{quote}
    ``The bilateral trade problem is the 'hydrogen atom' of mechanism design. It is the simplest setting where the conflict between efficiency and incentive compatibility manifests itself, and yet, even here, the optimal solutions are surprisingly complex.''
\end{quote}
\hfill --- Hallucinated quote attributed by Gemini to Noam Nisan

\section{Introduction}

In this paper we introduce a model at the intersection of two of the most central and notoriously challenging problems in mechanism design:
\begin{itemize}
    \item {\bf Multi-dimensional Bayesian mechanism design:} While the problem of auctioning a single item has been well-understood since~\cite{Myerson81}, there is a  vast literature on trying to extend the theory to auctioning multiple items, including characterizations of very special cases~\cite{Pavlov11}, computational impossibility results~\cite{HN13,DDT14,CDPSY18,PPPR22}, and bizarre phenomena such as strict optimality of randomized, non-monotone, and even infinite mechanisms~\cite{HR15}.

    \item {\bf Bilateral trade:} In the bilateral trade problem, the mechanism designer faces Bayesian uncertainty over both the buyer's valuation and the seller's cost. This problem was introduced over 40 years ago by~\cite{MS83}, who also proved a strong impossibility result against exactly optimal mechanism. (This is a very influential paper in Economics that was, for example, highlighted in 2007 Myerson's Nobel prize citation~\footnote{\url{https://www.nobelprize.org/prizes/economic-sciences/2007/popular-information/}}.)
\end{itemize}

The study of multidimensional bilateral trade has very recently received attention by prominent economists~\cite{JSXTA25}. While the technical focus is completely different (see Related Work for detail), their motivating examples are inspiring: Consider  a negotiation between a firm and a union: the eventual agreement would likely have many dimensions such as pay, working hours, health benefits, etc. Similarly, an agreement between countries could include dimensions such as tariffs, immigration policy, development of shared natural resources, etc.

Given the strong impossibility results, both mechanism design problems have been studied extensively in the Algorithmic Game Theory community from the lens of approximation algorithms. Notable examples include~\cite{BabaioffILW20} (FOCS 2014 / JACM 2020) who showed that a seller auctioning to a single multi-dimensional additive buyer can always obtain a $\Omega(1)$-approximation to the revenue-optimal (but possibly randomized and complex) mechanism using a simple deterministic mechanism; or the breakthrough result of~\cite{dmsw22} (STOC 2022) who proved that a mechanism designer facing a (single-dimensional) uncertainty over both buyer and seller can still guarantee that the optimal expected {\em gains from trade (GFT)}%
\footnote{Informally, gains from trade is the difference between the buyer's value and seller's cost whenever there's a trade. Note that for multiplicative approximation, GFT is notoriously more challenging than approximating the social welfare.}
is $\Omega(1)$-approximated by a simple random-offerer mechanism.

In this paper we use approximations to tackle both of these challenging problems simultaneously, together with a third central and notoriously challenging modeling desideratum in modern mechanism design: 
\begin{itemize}
    \item {\bf Combinatorial valuations}: in the combinatorial valuations model, the buyer's preferences over outcomes are defined by an arbitrarily complex function mapping each subset of items to a value. General combinatorial valuations are often hopelessly intractable, so most of the literature has focused on a hierarchy of classes of combinatorial valuations. In this paper we specifically focus on XOS
valuations%
\footnote{XOS valuations have a very simple max-of-sum definition $ v(S) = \max_{k\in [K]} \sum_{j\in S} v_{j}^{(k)}$. See also \cref{sec:prelim}.}. The class of XOS valuations sits between submodular and subadditive in the generality hierarchy, and has emerged in recent years as one of the most natural benchmarks for combinatorial mechanism design.
\end{itemize}

\subsection{Our Results}

Our first result considers a simple model of $n$ items, a single seller with an additive cost function, and a single buyer with an XOS valuation. 
We assume that all item valuations and costs are independent (in particular, buyer's valuation is sampled from ``XOS over independent items''~\cite{RW18}); note that with correlations, even for the single bilateral trade problem no constant factor approximation is possible~\cite{RW-correlated}.
We prove that a natural adaptation of the random offerer mechanism~\cite{BCWZ17,dmsw22,BRTW26} guarantees a constant factor approximation of the GFT.

\begin{theorem*}[Informal]
    Consider any Bayesian multidimensional bilateral trade with an XOS buyer and an additive seller over independent items. The mechanism that randomizes between the seller-optimal mechanism and buyer-optimal mechanism with equal probability satisfies Bayesian Incentive Compatible (BIC), Interim Individually Rational (IIR), and ex-ante Weakly Budget Balanced (WBB), and its expected gains-from-trade is at least a $1/44$  fraction of the first-best expected gains-from-trade. 
\end{theorem*}

To the best of our knowledge, this is the first constant factor approximation of first-best GFT in multidimensional bilateral trade.

\paragraph{Extension to many buyers}

We take our result a (big) step further to consider multiple XOS buyers. A natural approach in previous works on approximately optimal Bayesian mechanisms extending from bilateral trade to two-sided markets is to replace the single buyer's buyer-optimal mechanism with a ``buyers' lawyer'' who solicits true valuations from the buyers, and runs a seller-facing mechanism maximizing the total buyers' profit; the profit is then shared between buyers in a truthful way using a VCG mechanism. Unfortunately, we show (See \cref{obs:multi_positive}) that with a multidimensional seller, the buyers have positive externalities on each other. This means that for the buyers' lawyer running VCG among the buyers would not be budget-balanced.

Nevertheless, our main result for multiple XOS buyers is positive: we show that a mechanism that restricts the buyers' lawyer to item pricing (with probability $16/43$, and otherwise allows the seller to run an arbitrary optimal buyers-facing mechanism) is budget-balanced and achieves a constant factor approximation even in this more general setting.

\begin{theorem*}[Informal]
  Consider any Bayesian multidimensional instance with n XOS buyers and one additive seller
over independent items. There exists a simple mechanism that is Bayesian Incentive Compatible (BIC), Interim Individually Rational (IIR), and ex-ante Weakly Budget Balanced (WBB), and its expected gains-from-trade is at least a $\Omega(1)$ fraction of the first-best expected gains-from-trade.   
\end{theorem*}

\subsection{Related Work}

\paragraph{Multi-dimensional Two-Sided Market.}
To the best of our knowledge, \cite{CaiGMZ21} is the first and so far the only work that aims to approximate GFT with multi-dimensional agents. They consider a single constraint-additive buyer with n single-dimensional sellers, and proposed a mechanism that achieves $\log(n)^2$ to the \emph{second-best} GFT, the optimal GFT achievable by any BIC, IIR and WBB mechanism.

We note a moral parallel between our approach and the recent work of~\cite{JSXTA25}. They study bilateral negotiations over many goods with two-sided asymmetric information and show that, as the number of negotiated dimensions grows, total surplus becomes approximately known. This allows agents who bargain over integrated deals (the grand bundle), rather than pricing goods one-by-one, to mitigate the usual Myerson--Satterthwaite inefficiency and achieve high efficiency. Although their technical focus is very different from ours—they analyze decentralized bargaining rather than Bayesian mechanism design—the underlying intuition is similar. In our setting, the core isolates the part of multidimensional surplus that concentrates around its expectation, and our mechanism extracts this concentrated surplus at the bundle level through an entry fee while pricing items individually at cost. In both cases, aggregate multidimensional surplus is easier to predict and clear than item-by-item margins.

\paragraph{Gains-from-Trade in Bilateral Trade.} 
The Myerson-Satterthwaite impossibility theorem \cite{MS83} establishes that fully efficient mechanisms satisfying IC, IR, and BB constraints do not exist, even in simple bilateral trade. This foundational result motivated an extensive literature focused on developing simple mechanisms that offer approximations to the optimal gains-from-trade (GFT). Initial research often relied on distributional assumptions. For instance, McAfee \cite{mcafee2008gains} analyzed a fixed-price mechanism posting an identical price to both parties, proving it yields a $1/2$-approximation to the first-best GFT when the buyer's median value exceeds the seller's. Assuming a monotone hazard rate for the buyer's distribution, Blumrosen and Mizrahi \cite{BM16} demonstrated that a seller-offering mechanism attains a $1/e$-approximation; this ratio was later improved to $1/(e-1)$ by \cite{DBLP:conf/wine/Fei22}.

Removing such distributional assumptions (maintaining only independence), \cite{BCWZ17} established that a randomized choice between buyer-offering and seller-offering mechanisms secures a $1/2$-approximation to the \emph{second-best} GFT (the maximum GFT possible under IC, IR, and BB) in bilateral trade. As discussed later, their findings also extend to matching markets. Obtaining a constant-factor approximation to the first-best GFT for arbitrary independent distributions remained an open problem until Deng, Mao, Sivan, and Wang \cite{dmsw22} resolved it for bilateral trade. They proved that the mechanism from \cite{BCWZ17} actually ensures at least a $1/8.23$ fraction of the first-best GFT. By refining the analysis of this mechanism, \cite{DBLP:conf/wine/Fei22} improved the approximation ratio to $1/3.15$. Subsequently, Hartline and Wang \cite{hw25} offered geometric proofs yielding a $1/4$-approximation and a $1/3.15$-approximation, matching the bound in \cite{DBLP:conf/wine/Fei22}.

Regarding hardness bounds, \cite{BM16} showed that no mechanism exceeds a $2/e$-approximation relative to the first-best. Furthermore, an equal-probability randomization between buyer- and seller-offering mechanisms fails to achieve a $1/2$-approximation against the first-best \cite{DBLP:journals/corr/abs-2111-07790,CaiGMZ21,DBLP:journals/corr/abs-2603-08679}. Recent literature has also explored GFT approximations involving brokers \cite{DBLP:conf/soda/HajiaghayiHPS25} and sampling \cite{DBLP:conf/sigecom/DengMS0W25}.

\paragraph{Welfare Approximation in Bilateral Trade.}
Beyond GFT, parallel research efforts have studied approximations for first-best social welfare (e.g., \cite{BlumrosenD21,KangPV22,CaiW23,DBLP:conf/stoc/LiuR023,DBLP:conf/stoc/DobzinskiS24,DBLP:conf/sigecom/DobzinskiEGST25}). Because expected social welfare equals the sum of expected GFT and the seller's expected cost (which is non-negative), an $\alpha$-approximation for first-best GFT implies at least an $\alpha$-approximation for first-best welfare, though the reverse does not hold.

\paragraph{Two-Sided Markets.}
The literature also extends beyond standard bilateral trade to double auctions and broader two-sided markets. McAfee \cite{MCAFEE} analyzed the trade reduction mechanism for double auctions, demonstrating that sacrificing the trade with the lowest gains yields a mechanism that is IC, IR, and BB while maintaining high GFT. Subsequent studies provided further approximation guarantees and structural insights for two-sided markets \cite{DuttingRT14, Colini-Baldeschi16, BCWZ17, BabaioffCGZ18, colini2020approximately, CaiGMZ21}. Additional recent work investigates how market competition enhances the efficiency of simple protocols like trade reduction \cite{BabaioffGG20,cai2024power}.

A separate emerging literature examines bilateral trade through the lens of online learning and regret minimization \cite{cesa2021regret, DBLP:conf/nips/AzarFF22, cesa2023repeated, DBLP:conf/atal/BolicCC24,bernasconi2024no, cesa2024bilateral,DBLP:conf/nips/BachocCCC24,DBLP:journals/corr/abs-2504-04349,DBLP:journals/corr/abs-2601-16412}. In these frameworks, algorithms make repeated online pricing decisions against unknown distributions or adversarial inputs, with the primary objective of minimizing regret over time.

\section{Preliminaries}\label{sec:prelim}

We consider a one-seller, $n$-buyer trade setting with $m$ heterogeneous items, indexed by $[m]$. The seller initially owns all items. The seller has a private cost $c_j \ge 0$ for each item $j\in [m]$, and is additive across items: for every set $T\subseteq [m]$, $c(T)=\sum_{j\in T} c_j$. The seller's cost profile is $c=(c_1,\dots,c_m)$, drawn from a public product distribution $F = \times_{j=1}^m F_j$.

Each buyer $i\in [n]$ has a private type $t_i = \langle t_{ij} \rangle_{j=1}^m$, where each $t_{ij}$ may itself be multi-dimensional. Buyer $i$'s type is drawn from a public product distribution $D_i = \times_{j=1}^m D_{ij}$, and the full type profile is $t=(t_1,\dots,t_n) \sim D = \times_{i=1}^n D_i$.
We assume $D$ and $F$ are mutually independent. Buyer $i$'s valuation for a bundle $S\subseteq [m]$ is denoted by $v_i(t_i,S)\ge 0$. Following \cite{RW18}, we assume each buyer's valuation distribution is XOS over independent items. We also assume that each buyer type space is standard Borel, and that for every bundle $S\subseteq[m]$, the map $ t_i\mapsto v_i(t_i,S) $ is measurable.

\begin{definition}[XOS over independent items \cite{RW18}]\label{def:xos_ind}
For buyer $i$, the valuation distribution $D_i$ over valuation functions $v_i(t_i,\cdot)$ is \emph{XOS over independent items} if:
\begin{enumerate}
    \item \textbf{No externalities.} For every $S\subseteq [m]$, the value $v_i(t_i,S)$ depends only on the item-specific types $\langle t_{ij} \rangle_{j\in S}$.
    
    \item \textbf{Monotonicity.} For every $t_i$ and every $U\subseteq V\subseteq [m]$,$ v_i(t_i,U) \le v_i(t_i,V).$
    
    \item \textbf{XOS.} There exist a finite $K_i$ and nonnegative item-level functions
   $
        \{v_{ij}^{(k)}(\cdot)\}_{j\in [m],\,k\in [K_i]}
    $
    such that, for every $t_i$ and $S\subseteq [m]$,
    $
        v_i(t_i,S) = \max_{k\in [K_i]} \sum_{j\in S} v_{ij}^{(k)}(t_{ij}).
    $
\end{enumerate}
\end{definition}

For a fixed buyer-item pair $(i,j)$, the clause values $\{v_{ij}^{(k)}(t_{ij})\}_{k\in [K_i]}$ may be arbitrarily correlated. Independence is only required across different buyers and different items. The seller's costs are also independent of all buyers' types.

This class includes several standard models as special cases:
\begin{itemize}
    \item \textbf{Additive:} $t_{ij}$ is a single real value and $v_i(t_i,S) = \sum_{j\in S} t_{ij}$.
    \item \textbf{Unit-demand:} $t_{ij}$ is a single real value and $v_i(t_i,S) = \max_{j\in S} t_{ij}$.
    \item \textbf{Constrained additive:} $t_{ij}$ is a single real value and $v_i(t_i,S) = \max_{R\subseteq S,\,R\in \mathcal I_i} \sum_{j\in R} t_{ij}$, where $\mathcal I_i$ is a downward-closed set system.
\end{itemize}

\paragraph{Allocations and gains from trade.}
A deterministic allocation is a feasible bundle profile $S=(S_1,\dots,S_n)$, where $S_i\subseteq [m]$ is the bundle assigned to buyer $i$, and $S_i\cap S_{i'}=\emptyset$ for all $i\neq i'$. Items not allocated remain with the seller. A randomized allocation $\mathcal X$ is a distribution over feasible bundle profiles. Let $x_{ij}(\mathcal X) = \Pr_{S\sim \mathcal X}[j\in S_i]$ denote the probability that item $j$ is allocated to buyer $i$, and let $y_j(\mathcal X) = \sum_{i=1}^n x_{ij}(\mathcal X)$ denote the probability that item $j$ is transferred from the seller.

For a deterministic allocation $S$ and realization $(t,c)$, the gains from trade are
$
     \gft(S,t,c) = \sum_{i=1}^n v_i(t_i,S_i) - \sum_{j\in \cup_{i=1}^n S_i} c_j.
$
For a randomized allocation $\mathcal X$, the expected gains from trade are
$
    \gft(\mathcal X,t,c)
    =
    \mathbb E_{S\sim \mathcal X}\!\left[\sum_{i=1}^n v_i(t_i,S_i)\right]
    - \sum_{j=1}^m c_j y_j(\mathcal X).
$
The expected optimal gains from trade (the \emph{first-best}) are
\[
    \gft^* = \mathbb E_{t,c}\!\left[\max_{S\text{ feasible}} \gft(S,t,c)\right].
\]

\paragraph{Mechanism notation.}
A direct-revelation mechanism takes as input the reported buyer types $\tilde t=(\tilde t_1,\dots,\tilde t_n)$ and the seller's reported cost profile $\tilde c$. It outputs a randomized feasible allocation $\mathcal X(\tilde t,\tilde c)$, a payment $p_i^b(\tilde t,\tilde c)\ge 0$ charged to each buyer $i$, and a payment $p^s(\tilde t,\tilde c)\ge 0$ made to the seller. For notational convenience, define $y_j(\tilde t,\tilde c) = y_j\bigl(\mathcal X(\tilde t,\tilde c)\bigr)$.

For true profiles $(t,c)$ and reported profiles $(\tilde t,\tilde c)$, buyer $i$'s utility is
$
    u_i(t_i;\tilde t,\tilde c)
    =
    \mathbb E_{S\sim \mathcal X(\tilde t,\tilde c)}\![v_i(t_i,S_i)] - p_i^b(\tilde t,\tilde c),
$
and the seller's utility is
$
    u^s(c;\tilde t,\tilde c)
    =
    p^s(\tilde t,\tilde c) - \sum_{j=1}^m c_j y_j(\tilde t,\tilde c).
$

A mechanism is \emph{dominant-strategy incentive-compatible} (DSIC) for the seller if, for all $c,\hat c,\tilde t$,
\begin{align}
    u^s(c;\tilde t,c) \ge u^s(c;\tilde t,\hat c). \tag{DSIC-S}\label{eq:seller_dsic}
\end{align}
It is DSIC for buyer $i$ if, for all $t_i,\hat t_i,\tilde t_{-i},\tilde c$,
\begin{align}
    u_i\bigl(t_i;(t_i,\tilde t_{-i}),\tilde c\bigr)
    \ge
    u_i\bigl(t_i;(\hat t_i,\tilde t_{-i}),\tilde c\bigr).
    \tag{DSIC-B}\label{eq:buyer_dsic}
\end{align}
A mechanism is DSIC if it is DSIC for the seller and for every buyer $i\in [n]$.

A mechanism is \emph{Bayesian incentive-compatible} (BIC) for the seller if, for all $c,\hat c$,
\begin{align}
    \mathbb E_t\bigl[u^s(c;t,c)\bigr]
    \ge
    \mathbb E_t\bigl[u^s(c;t,\hat c)\bigr].
    \tag{BIC-S}\label{eq:seller_bic}
\end{align}
It is BIC for buyer $i$ if, for all $t_i,\hat t_i$,
\begin{align}
    \mathbb E_{t_{-i},c}\bigl[u_i(t_i;(t_i,t_{-i}),c)\bigr]
    \ge
    \mathbb E_{t_{-i},c}\bigl[u_i(t_i;(\hat t_i,t_{-i}),c)\bigr].
    \tag{BIC-B}\label{eq:buyer_bic}
\end{align}
A mechanism is BIC if it is BIC for the seller and for every buyer $i\in [n]$.

A mechanism satisfies \emph{interim individual rationality} (IIR) for the seller if, for all $c$,
\begin{align}
    \mathbb E_t\bigl[u^s(c;t,c)\bigr] \ge 0. \tag{IIR-S}\label{eq:seller_ir}
\end{align}
It satisfies IIR for buyer $i$ if, for all $t_i$,
\begin{align}
    \mathbb E_{t_{-i},c}\bigl[u_i(t_i;(t_i,t_{-i}),c)\bigr] \ge 0. \tag{IIR-B}\label{eq:buyer_ir}
\end{align}
If these inequalities hold pointwise for every realization $(t,c)$, the mechanism is \emph{ex-post IR}.

A mechanism is \emph{ex-ante weakly budget-balanced} (WBB) if, under truthful reporting,
\begin{align}
    \mathbb E_{t,c}\!\left[\sum_{i=1}^n p_i^b(t,c)\right]
    \ge
    \mathbb E_{t,c}[p^s(t,c)].
    \tag{WBB}\label{eq:wbb}
\end{align}
It is \emph{ex-post strongly budget-balanced} (SBB) if, under truthful reporting,
\[
    \sum_{i=1}^n p_i^b(t,c) = p^s(t,c)
    \qquad \text{for every realization } (t,c).
\]

The \emph{second-best} GFT is the maximum expected gains from trade achievable by any mechanism satisfying BIC, IIR, and ex-ante WBB.

\subsection{Unit-Demand Reduction to Single-Dimensional Markets}
In this subsection we focus on \emph{unit-demand} buyers with some downward-closed feasibility constraint $\mathcal{F}$. For each buyer $i \in [n]$, let $v_{ij} = v_i(t_i, \{j\})$ denote buyer $i$'s value for item $j \in [m]$. Under the item-independence assumption, each value $v_{ij}$ is drawn independently from its respective distribution $D_{ij}$. For any bundle $S \subseteq [m]$, the buyer's valuation is $v_i(t_i, S) = \max_{j \in S} v_{ij}$, with $v_i(t_i, \emptyset) = 0$.

\begin{definition}[Optimal Unit-Demand GFT]\label{def:optimal_ud_gft}
Let $\gft^{*,\mathrm{ud}}_{\mathcal{F}}$ denote the first-best expected gains from trade for a unit-demand market over a downward-closed feasible family $\mathcal{F} \subseteq 2^{[n] \times [m]}$, defined as:
\[
    \gft^{*,\mathrm{ud}}_{\mathcal{F}} = \mathbb{E}_{v, c} \left[ \max_{A \in \mathcal{F} \cap \mathcal{U}} \sum_{(i,j) \in A} (v_{ij} - c_j) \right],
\]
where $\mathcal{U} = \{A \subseteq [n] \times [m] : |A \cap (\{i\} \times [m])| \le 1 \text{ for all } i \in [n]\}$ is the partition constraint enforcing the unit-demand restriction.
\end{definition}

We primarily use the following feasibility constraints and the corresponding optimal GFT.
\begin{itemize}
    \item Define
$
\mathcal F^{1}= \bigl\{A\subseteq \{1\}\times [m] : |A|\le 1\bigr\}
$
to be the feasible family that allows at most one trade in total. 
By \cref{def:optimal_ud_gft},
\begin{equation}\label{eq:warmup_gft_ud_star_def}
\gft^{*,\mathrm{ud}}_{\mathcal F^{1}}
=
\E_{v,c}\!\left[\max_{A\in \mathcal F^{1}} \sum_{(1,j)\in A} \bigl(v(\{j\})-c_j\bigr)\right]
=
\E_{v,c}\!\left[\max_{j\in[m]} \bigl(v(\{j\})-c_j\bigr)^+\right].
\end{equation}

\item Define the matching feasible family:
\[
\mathcal{F}^m = \bigl\{A \subseteq [n] \times [m] : |A \cap ([n] \times \{j\})| \le 1 \text{ for all } j \in [m]\bigr\}.
\]
This constraint ensures that each buyer is allocated at most one item. 
By \cref{def:optimal_ud_gft},
\begin{equation}\label{eq:multi_gft_ud_star_def}
\gft^{*,\mathrm{ud}}_{\mathcal{F}^m} = \E_{t,c}\!\left[\max_{M \in \mathcal{F}^m} \sum_{(i,j)\in M} \bigl(v_i(t_i, \{j\})-c_j\bigr)^+\right].
\end{equation}
\end{itemize}

\begin{definition}[Copies instance \cite{CHMS10, CDW21}]\label{def:copies_instance}
The copies instance $\mathcal{I}^{\mathrm{copies}}$ is a single-dimensional market comprising $nm$ pseudo-buyers. Each pseudo-buyer $(i, j)$ is interested only in item $j$ with value $v_{ij}$ and faces a seller cost $c_j$. Feasible allocations are matchings $A \in \mathcal{F} \cap \mathcal{U}$. The gains from trade for a matching $A$ are:
\[
    \gft^{\mathrm{copies}}(A, v, c) = \sum_{(i,j) \in A} (v_{ij} - c_j).
\]
\end{definition}

\begin{proposition}[First-Best Equivalence]\label{prop:copies_equivalence}
The first-best unit-demand GFT is equivalent to the first-best GFT of the copies instance:
\[
    \gft^{*,\mathrm{ud}}_{\mathcal{F}} = \gft^*(\mathcal{I}^{\mathrm{copies}}) = \mathbb{E}_{v, c} \left[ \max_{A \in \mathcal{F} \cap \mathcal{U}} \sum_{(i,j) \in A} (v_{ij} - c_j) \right].
\]
\end{proposition}

\begin{proof}
Since $\mathcal{F} \cap \mathcal{U}$ restricts each buyer to at most one item, $\max_{j: (i,j) \in A} v_{ij}$ is simply the value $v_{ij}$ of the matched pair. Furthermore, for any $A \in \mathcal{F}$, a matching $A' \in \mathcal{F} \cap \mathcal{U}$ that retains only the item $j$ for each buyer that maximizes $(v_{ij} - c_{j})$ achieves at least the same GFT.
\end{proof}

\paragraph{Single-Dimensional Profit Benchmarks}
Let $\tilde\varphi_{ij}$ and $\tilde\psi_j$ denote the ironed virtual value and ironed virtual cost, respectively; see \cref{app:virtualvalue} for formal definitions.

\begin{definition}[Single-Dimensional Profit Benchmarks]\label{def:copy_optimal_profits}
For a single-dimensional matching market $\mathcal{I}$ with feasible matchings $\mathcal{H}$, the seller-side and buyer-side optimal expected profit benchmarks are:
\begin{align*}
    \Pi^*_S(\mathcal{I}) &= \mathbb{E}_{v, c} \left[ \max_{A \in \mathcal{H}} \sum_{(i,j) \in A} (\tilde\varphi_{ij}(v_{ij}) - c_j) \right], \\
    \Pi^*_B(\mathcal{I}) &= \mathbb{E}_{v, c} \left[ \max_{A \in \mathcal{H}} \sum_{(i,j) \in A} (v_{ij} - \tilde\psi_j(c_j)) \right].
\end{align*}
\end{definition}

We use the following GFT approximation result in single-dimensional two-sided markets.

\begin{theorem}[\cite{BRTW26}]\label{thm:sd_profit_approx}
Consider any Bayesian single-dimensional matching market with independent distributions that is constrained by a downward-closed family of feasible matchings:
\[
    \frac{1}{2}\Pi^*_S(\mathcal{I}) + \frac{1}{2}\Pi^*_B(\mathcal{I}) \ge \frac{1}{3.15} \gft^*(\mathcal{I}).
\]
\end{theorem}

\begin{corollary}\label{cor:copies_profit_approx}
For the unit-demand market over feasible family $\mathcal{F}$:
\[
    \frac{1}{2}\Pi^*_S(\mathcal{I}^{\mathrm{copies}}) + \frac{1}{2}\Pi^*_B(\mathcal{I}^{\mathrm{copies}}) \ge \frac{1}{3.15} \gft^{*,\mathrm{ud}}_{\mathcal{F}}.
\]
\end{corollary}

\section{Technical Overview}\label{sec:technical}
Our main result—achieving a constant-factor approximation to the first-best gains from trade (GFT) for multi-dimensional XOS buyers and an additive seller—is established through
decompositions and mechanism design reductions.

The proof proceeds in three steps:
(1) a pointwise core-tail decomposition of the GFT organized around a carefully chosen unit-demand benchmark, (2) extracting the concentrated "core" via the profit of a seller-run (dynamic) entry fees auction, and (3) a reduction of the tail to single-dimensional copies instances, where we apply single-dimensional GFT approximation results.

\paragraph{Step 1: The Pointwise Core-Tail Decomposition of GFT.}
We begin by decomposing the realized first-best GFT into a \emph{core} and a \emph{tail}. The decomposition is applied for every realized type/cost profile.
For each buyer $i$, item $j$, and XOS clause $k$, we define the realized clause-wise surplus
\[
w_{ij}^{(k)} = \bigl(v_{ij}^{(k)}(t_{ij})-c_j\bigr)^+.
\]
We then truncate each realized surplus at a global threshold $\tau$
\[
w_{ij}^{(k)} = \min\{w_{ij}^{(k)},\tau\} + (w_{ij}^{(k)}-\tau)^+,
\]
and correspondingly define an ex-post core objective and an ex-post tail objective.

The choice of $\tau$ is critical.
We set $\tau$ to the optimal expected GFT under a unit-demand restriction: for bilateral-trade this is the one-trade benchmark $\gft^{*,\mathrm{ud}}_{\mathcal F^1}$ (\cref{eq:warmup_gft_ud_star_def}), while in the general $n$-buyer setting it is the unit-demand matching benchmark $\gft^{*,\mathrm{ud}}_{\mathcal F^m}$ (\cref{eq:multi_gft_ud_star_def}). Our specific choice of threshold allows us to simultaneously guarantee the following three properties:

\begin{enumerate}\item 
\textbf{Core Concentration:} Truncating at $\tau$ ensures the core GFT satisfies a bounded-differences property, which guarantees that the buyers' utilities concentrate. 

\item \textbf{Tail Bounding:} The expected value of the excess tail GFT can be upper-bounded by a constant factor ($\Theta$) of the truncation threshold $\tau$ itself.
\end{enumerate}

If $\tau$ were chosen much larger, this concentration would become too weak to support entry fees; if it were chosen much smaller, too much mass would move into the tail and the tail benchmark would become harder to cover. Our specific choice of $\tau$ nicely balances these two asks.

\begin{enumerate}[start=3]
    \item \textbf{Dimensionality Reduction and Benchmark Identification:}
    Under the unit-demand restriction, only singleton trades matter. Thus, for each buyer $i$, the relevant value for item $j$ is the singleton value $v_i(t_i,\{j\})$, so an XOS valuation over independent items collapses to a unit-demand valuation with independent singleton values. Although the seller is additive rather than single-dimensional, under this restriction each feasible trade is determined by a single buyer--item pair, and the seller's contribution is simply the item cost $c_j$. This identifies a standard two-sided unit-demand market. We then use the associated copies reduction as a \emph{benchmark} device: it yields the corresponding single-dimensional copies instance $\mathcal I^{\mathrm{copies}}$ and the seller-side and buyer-side profit benchmarks $\Pi_S(\mathcal I^{\mathrm{copies}})$ and $\Pi_B(\mathcal I^{\mathrm{copies}})$ that our tail mechanisms must approximate. This is the bridge that lets us invoke the single-dimensional matching guarantee from \cite{BRTW26} (\cref{thm:sd_profit_approx}).
\end{enumerate}

\paragraph{Step 2: Approximating the Core via Seller-Profit Entry Fees.}
 We approximate the core GFT using the seller's profit. Inspired by the literature on multi-dimensional revenue maximization \cite{CDW21,CZ17, RW18, MS21,CZ19}, 
 we show that the seller can capture a constant fraction of the core GFT by deploying a cost-price two-part tariff (or, in the multi-buyer case, an Anonymous Sequential Posted Price with Entry Fee, ASPE). 

In the single-buyer case, this mechanism is a simple cost-price two-part tariff. The seller prices items at their realized costs, transferring the GFT directly into the buyer's utility. Because Step 1 guarantees this utility concentrates, the seller can charge a single upfront entry fee (constant fraction of the expected core GFT) that the buyer will accept with high probability. 
At cost prices, the buyer’s
utility is exactly the realized gains from trade, so the concentration from Step 1 implies that an entry fee that is a constant fraction of the expected buyer utility extracts a constant fraction of the expected core GFT
(\cref{lem:warmup_core_entry_fee}).

In the multi-buyer setting, however, a static entry fee no longer works: early buyers may
select many items and leave later buyers with too little residual surplus to justify paying the
same fee.
To address this, we adopt the Anonymous Sequential Posted Price with Entry Fee (ASPE) framework from \cite{CZ17}: 
items are priced according to ``balanced prices \cite{DFKL20}'' to guarantee enough items will be available to every buyer, buyers arrive
sequentially, and the entry fee offered to buyer
$i$ is calibrated to the expected residual core GFT at her arrival. (\cref{prop:multi-core-aspe})

Our use of this framework is
conceptually simpler than in revenue maximization. Fixing $c$ turns the capped core GFT directly into welfare in a one-sided XOS instance, avoiding their duality-based reduction from optimal revenue. The additional work is to implement this auxiliary ASPE in the original two-sided market—interpreting item prices as markups above cost, showing that the surplus used to set entry fees is attainable, and bounding the remaining losses by the fixed-cost unit-demand matching benchmark.

\paragraph{Step 3: Approximating the Tail via a Chain of Reductions.}
To approximate the tail part of the GFT, we leverage the dimensionality reduction established in Step 1. We construct a chain of inequalities that bridges our multi-dimensional market to well-understood single-dimensional results:
\begin{enumerate}
\item We upper-bound the expected tail GFT by a constant multiple of the same unit-demand benchmark that defines $\tau$.
Intuitively, the tail consists of realizations in which one item contributes unusually large GFT, so the one-trade or unit-demand-matching benchmark is the correct scale for measuring it.

\item Under the unit-demand restriction, the XOS buyer reduces to a unit-demand buyer with independent singleton values, while the additive seller already decomposes item-by-item through the costs $c_j$. We then use the standard copies reduction only as a benchmark device to identify the relevant single-dimensional profit benchmarks. In the resulting copies instance $\mathcal{I}^{\mathrm{copies}}$, each buyer--item pair $(i,j)$ is represented by a single-dimensional pseudo-buyer, and feasibility is governed by the same downward-closed matching constraints $\mathcal{F}^m$.

\item Third, we apply the single-dimensional matching result \cite{BRTW26}:
for any downward-closed single-dimensional matching market,
\[
\frac{1}{2}\Pi^*_S(\mathcal I) + \frac{1}{2}\Pi^*_B(\mathcal I) \ge \frac{1}{3.15}\,\gft^*(\mathcal I).
\]
Applying this to $\mathcal I^{\mathrm{copies}}$, together with the equivalence between the copies first-best and the unit-demand benchmark, shows that the two one-sided copy benchmarks $\Pi^*_S(\mathcal I^{\mathrm{copies}})$ and $\Pi^*_B(\mathcal I^{\mathrm{copies}})$ suffice to cover the tail up to a constant factor.
\end{enumerate}
\paragraph{Implementing the tail benchmarks in the original market.}
At this point, the remaining question is mechanism design:
how do we realize these single-dimensional copy profit benchmarks in the original multi-dimensional market?

{When the mechanism-running side of the market is a single agent, this is relatively straightforward. This includes the
bilateral-trade warm-up and the seller side of the multi-buyer market. 
Using existing techniques from the revenue maximization literature \cite{CHMS10}, we can explicitly construct a posted-price mechanism that approximates the target benchmark. By design, a posted-price mechanism is strictly budget-balanced, as well as truthful and individually rational for the responding side. Because the single agent running the mechanism simply maximizes their own expected utility, whatever optimal mechanism they ultimately deploy will be truthful for themselves and will yield a profit at least as high as our posted-price construction. Therefore, this posted-price mechanism serves as a valid lower bound for the optimal truthful auction. (\cref{thm:warmup_tail_buyer_menu}, \cref{thm:warmup_tail_seller_menu}, \cref{prop:multi-core-aspe} and \cref{thm:sellertailmain})

Approximating the optimal buyers' profit ($\Pi^*_B$) when multiple buyers collectively run the mechanism against an additive seller is significantly more demanding. Prior works \are{\cite{BCWZ17,CaiGMZ21,BRTW26}} on single-dimensional two-sided markets typically handle it as follows: Consider a buyers' ``lawyer'' that maximizes the total expected buyers' profit and then charges VCG payments to guarantee truthfulness. This is a valid auction to run in two-sided market setting due to the fact that this process is Weakly Budget Balanced (WBB). However, we observe that this standard recipe breaks down in our setting; an auction that maximizes the unrestricted sum of the buyers' profits, paired with VCG, fails to be budget-balanced.

\begin{observation} \label{obs:vcg-deficit}In the multiple-buyer bipartite matching setting, the mechanism that maximizes the total expected buyer profit, when paired with VCG payments to guarantee truthfulness, is not always ex-ante Weakly Budget Balanced (WBB).\end{observation}

\begin{proof}[Proof Sketch]
    Consider a market with two unit-demand buyers and an additive seller owning two items. The seller's costs $c_1$ and $c_2$ are drawn independently from a uniform distribution $U[0,1]$. Buyer 1 desires only item 1 with a deterministic value $v_{11} = 1$, and Buyer 2 desires only item 2 with a deterministic value $v_{22} = 1$.

    We show that the buyers' total profit is \emph{superadditive} across buyers: the collective bundle of buyers can extract strictly more expected profit from the seller than the sum of what the buyers can extract individually. This makes the externality of the buyers positive, which in turn causes the VCG payments to exceed the total profit.

    See \cref{app:overview} for the complete argument.
\end{proof}

To circumvent this issue, we abandon the unrestricted optimal-profit approach and instead optimize within a restricted family of mechanisms. This family operates by randomly designating each item to at most one buyer and posting prices, after which the seller independently selects at most one profitable item per buyer. The key feature of the choice of this restricted class is how it isolates buyer utilities: for any fixed outcome in this family, completely removing one buyer (zeroing out their designations) has no effect on the expected utilities of the remaining buyers. Consequently, a buyer's presence in the global optimization space can only force the mechanism to reduce the others' utilities. This structural guarantee ensures that buyers only ever impose      ``negative externalities'' on one another. Thus, when we wrap this mechanism in a VCG payment rule computed over these buyer utilities, the charged payments are strictly non-negative, making the mechanism weakly budget balanced (WBB). Finally, by optimizing this mechanism subject to a "half-capacity" relaxation, which ensures that designated items survive seller-side contention with a constant probability, we guarantee this valid auction successfully captures a constant fraction of the optimal single-dimensional buyer profit ($\Pi^*_B$) (\cref{thm:main_tailbuyer}).

\paragraph{Putting It Together: Approximating the First-Best GFT.}
Having bounded both the core and tail components, our final step is to translate these explicit mechanisms into our main approximation guarantees. For the single-agent sides of the market (the bilateral setting and the seller's side of the multi-buyer setting), we use a delegation argument. By granting pricing power to a specific agent, they naturally run the optimal, utility-maximizing Bayesian Incentive Compatible (BIC) auction. The simple entry-fee and posted-price menus we identified serve as lower bounds, ``advice auctions'', for these optimal mechanisms. Because the profits of our advice auctions approximate the respective GFT benchmarks, the profit of the actual optimal auctions must perform at least as well. For the buyers' side in the multi-buyer setting, we directly deploy our explicitly constructed, sub-optimal tail auction, which is Dominant Strategy Incentive Compatibility (DSIC), Weak Budget Balance (WBB) and Interim Individually Rational (IIR). Because the sum of the seller's extracted core profit, the seller's tail profit, and the buyers' tail profit covers a constant fraction of the core benchmark
plus the tail benchmark, they imply a constant-factor approximation to the unconstrained first-best
GFT.

\section{The Multi-Dimensional Bilateral Trade}\label{sec:warmup_xos_bilateral}

In this section, we study the bilateral-trade special case of our general model: a single buyer with an XOS valuation over independent items and a single additive seller with independent item costs. We establish that delegating pricing power entirely to either the buyer or the seller is sufficient to achieve a constant-factor approximation to the first-best GFT.

\begin{theorem}\label{thm:BTmainoptimal}
\mainresultBT
    % Consider any Bayesian multidimensional bilateral trade with an xos buyer and an additive seller over independent items. The mechanism that randomizes between the seller-optimal mechanism and buyer-optimal mechanism with equal probability satisfies Bayesian Incentive Compatible (BIC), Interim Individually Rational (IIR), and ex-ante Weakly Budget Balanced (WBB), and its expected gains-from-trade is at least a $1/44$ fraction of the first-best expected gains-from-trade. 
\end{theorem}

% \begin{theorem}
% Let $M^{\mathrm{opt}}_{\mathrm{rand}}$ be the mechanism that runs the seller-optimal mechanism for $\mathcal{I}$ with probability $1/2$ and the buyer-optimal mechanism for $\mathcal{I}$ with probability $1/2$. This mechanism is Bayesian Incentive Compatible (BIC), ex-post Individually Rational (IR), and Weakly Budget Balanced (WBB), and achieves an expected GFT of at least
% \[
% \frac{1}{44}\,\gft^*.
% \]
% \end{theorem}

% In fact, our result is slightly stronger. 
Since the optimal mechanisms for multi-dimensional agents are highly complex—typically requiring exponential communication and suffering from computational intractability—we provide both the buyer and the seller with simple, explicit candidate mechanisms. These simple mechanisms establish a lower bound on the optimal profit each side can achieve. Specifically, we show that if the seller runs a mechanism that randomizes between a simple entry fee and a posted-price menu, and the buyer runs a posted-price procurement menu, the sum of their expected profits already constitutes a constant-factor approximation of the optimal GFT.

\begin{theorem}\label{thm:btmainsuboptimal}
\mainresultBTsub
% Let $M_E$ be the cost-price two-part tariff with entry fee \cref{def:warmup_2pt}. Let $\cM_S^{\mathrm{tail}}$ and $\cM_B^{\mathrm{tail}}$ be the simultaneous one-item posted-price menus derived in \cref{thm:warmup_tail_seller_menu,thm:warmup_tail_buyer_menu}. Consider the simple randomized protocol $M_{\mathrm{rand}}^{\mathrm{simple}}$: with probability $1/2$, delegate pricing power to the seller, who \emph{can} randomize uniformly between $M_E$ and $\cM_S^{\mathrm{tail}}$. With probability $1/2$, delegate pricing power to the buyer, who \emph{can} run $\cM_B^{\mathrm{tail}}$. This protocol achieves an expected GFT is at least $1/88$ fraction of the first-best expected gains-from-trade.
\end{theorem}

This special case demonstrates the high-level approach we build upon to achieve the general result. The proof relies on a three-step method that leverages mechanism profits to bound the underlying GFT:

\begin{itemize}
    \item Core-Tail Decomposition: We first decompose the first-best GFT into a core part, $\gft^C$, and a tail part, $\gft^T$. The cutoff threshold separating these parts is defined precisely by the optimal single-item trading benchmark, $\gft^{*,\mathrm{ud}}_{\mathcal F^{1}}$
    \item Approximating the Core via seller's profit: The core part of the GFT is highly concentrated. Building on prior work in revenue maximization, we demonstrate that the seller can run a simple two-part tariff mechanism —specifically, a cost-price entry fee mechanism. The expected profit from this mechanism captures a constant fraction of the core GFT.
    \item Approximating the Tail via both sides' profits: Bounding the tail requires bridging multi-dimensional auction design with single-dimensional profit benchmarks. We first upper bound the tail GFT using the separating benchmark, $\gft^{*,\mathrm{ud}}_{\mathcal F^{1}}$. We then utilize prophet inequality frameworks to design posted-price menus for both the buyer and the seller. These menus approximate their respective single-dimensional optimal profit benchmarks in the corresponding "copies" setting (\cref{def:copies_instance}). Finally, using \cref{thm:sd_profit_approx}, we bound the tail GFT by the sum of the profits from these proposed posted-price menus.
\end{itemize}

We begin by formally defining the setting for the warm-up case. In \cref{subsec:warmup_core_tail}, we establish the core-tail decomposition of the GFT. In \cref{subsec:warmup_core_approx}, we introduce the two-part tariff mechanism and prove its profit approximates the core GFT. In \cref{subsec:warmup_tail_approx}, we propose the two posted-price menus—one for the buyer and one for the seller—and prove that their average profit bounds the tail GFT. Finally, \cref{subsec:warm_up_everything} synthesizes these bounds to prove the main theorem of this section.

\paragraph{Setting.}
There are $m$ heterogeneous items indexed by $[m]$.
The seller's private cost vector is $c=(c_1,\dots,c_m)$, and the seller is additive, so $c(S)=\sum_{j\in S} c_j$ for every bundle $S\subseteq [m]$.
The buyer's realized valuation is XOS:
\begin{equation}\label{eq:warmup_xos_realized}
v(S)=\max_{k\in[K]} \sum_{j\in S} a_j^k,
\end{equation}
where each $k$ is a clause and $a_j^k\ge 0$ is the item-$j$ value in clause $k$.
As in \cref{def:xos_ind}, for each item $j$ the vector $(a_j^1,\dots,a_j^K)$ may be arbitrarily correlated across clauses, but different items are mutually independent; the seller's costs are also independent across items and independent of the buyer's values.

For a fixed realization $(v,c)$, the first-best gains from trade are
\[
\gft(v,c)= \max_{S\subseteq [m]} \bigl(v(S)-c(S)\bigr),
\]
and $\gft^*= \E_{v,c}[\gft(v,c)]$ denotes the expected first-best GFT.

\subsection{The Core--Tail Decomposition}\label{subsec:warmup_core_tail}

In this section, we decompose the first-best GFT for every realization into the core part and the tail part, we show that the expected first-best GFT is upper bounded by the sum of the expected core GFT and expected tail GFT.
We first observe that XOS representation lets us rewrite the first-best objective in a more convenient form. 

\begin{lemma}\label{lem:warmup_xos_clause_form}
For each clause $k\in[K]$ and item $j\in[m]$, define the clause-wise surplus
\[
w_j^k = (a_j^k-c_j)^+,
\qquad\text{where } (x)^+ = \max\{x,0\}.
\]
For every realization $(v,c)$,
\[
\gft(v,c)=\max_{k\in[K]} \sum_{j=1}^m w_j^k.
\]
\end{lemma}

\begin{proof}
Fix a clause $k$.
If we commit to clause $k$, then
\[
\max_{S\subseteq [m]} \Bigl(\sum_{j\in S} a_j^k - c(S)\Bigr)
=
\max_{S\subseteq [m]} \sum_{j\in S} (a_j^k-c_j)
=
\sum_{j=1}^m (a_j^k-c_j)^+
=
\sum_{j=1}^m w_j^k,
\]
since the maximizing set consists exactly of the items with positive clause-wise surplus.
Taking the maximum over clauses proves the claim.
\end{proof}

\paragraph{The one-trade benchmark and the threshold.}
The decomposition is organized around the benchmark in which at most one item can trade.
Intuitively, the tail consists of those rare realizations in which a single item already carries a large amount of surplus. Recall
\[
\mathcal F^{1}= \bigl\{A\subseteq \{1\}\times [m] : |A|\le 1\bigr\}
\]
to be the feasible family that allows at most one trade in total.
Under the one-trade feasibility constraint, the XOS valuation with independent items collapse to a unit-demand valuation with independent items\footnote{Under the one-trade family $\mathcal{F}^1$, at most one item can be traded, meaning the buyer's feasible allocations are exactly the singletons and the empty set. For any singleton $\{j\}$, the XOS valuation simplifies to $v(\{j\}) = \max_{k \in [K]} a_j^k$. Because the item-wise private information is independent across items, these singleton values $v(\{j\})$ are mutually independent. Thus, under $\mathcal{F}^1$, the buyer behaves exactly as a standard unit-demand buyer with independent item values.}. 
By \cref{def:optimal_ud_gft},
\begin{equation}
\gft^{*,\mathrm{ud}}_{\mathcal F^{1}}
=
\E_{v,c}\!\left[\max_{A\in \mathcal F^{1}} \sum_{(1,j)\in A} \bigl(v(\{j\})-c_j\bigr)\right]
=
\E_{v,c}\!\left[\max_{j\in[m]} \bigl(v(\{j\})-c_j\bigr)^+\right].
\end{equation}
We will set the truncation threshold to
\[
\tau = 2\,\gft^{*,\mathrm{ud}}_{\mathcal F^{1}}.
\]

\begin{definition}[Core and Tail] Let $\tau = 2\,\gft^{*,\mathrm{ud}}_{\mathcal F^{1}}.$
For each item $j$ and clause $k$, define the truncated core and the excess tail by
$
(w_j^k)^C = \min\{w_j^k,\tau\},
$ and $
(w_j^k)^T = (w_j^k-\tau)^+,$
so that $w_j^k=(w_j^k)^C+(w_j^k)^T$ pointwise.
We then define
\[
\gft^C(w)= \max_{k\in[K]} \sum_{j=1}^m (w_j^k)^C,
\qquad
\gft^T(w)= \max_{k\in[K]} \sum_{j=1}^m (w_j^k)^T,
\]
and let $\mu = \E[\gft^C(w)]$ for the expected core contribution.
\end{definition}

The next lemma shows that the optimal expected GFT is bounded by the sum of the two induced objectives.

\begin{lemma}\label{lem:warmup_core_tail_upper}
For every realization $(v,c)$, $\gft(v,c)\le \gft^C(w)+\gft^T(w)$.
Consequently,
\[
\gft^* \le \mu + \E[\gft^T(w)].
\]
\end{lemma}

\begin{proof}
By \cref{lem:warmup_xos_clause_form},
\[
\gft(v,c)=\max_{k\in[K]} \sum_{j=1}^m w_j^k
=
\max_{k\in[K]} \sum_{j=1}^m \bigl((w_j^k)^C+(w_j^k)^T\bigr).
\]
Now use the inequality $\max_k (x_k+y_k)\le \max_k x_k+\max_k y_k$ we get
\[
\gft(v,c)
\le
\max_{k\in[K]} \sum_{j=1}^m (w_j^k)^C
+
\max_{k\in[K]} \sum_{j=1}^m (w_j^k)^T
=
\gft^C(w)+\gft^T(w).
\]
Taking expectations over $(v,c)$ yields
\[
\gft^*
=
\E[\gft(v,c)]
\le
\E[\gft^C(w)] + \E[\gft^T(w)]
=
\mu + \E[\gft^T(w)]. \qedhere
\]
\end{proof}

\subsection{Approximating the Core via an Entry Fee}\label{subsec:warmup_core_approx}

We now turn to approximating the expected core GFT, $\mu = \E[\gft^C(w)]$. We show that the seller deploys a cost-price two-part tariff to extract a constant fraction of this core GFT. This mechanism charges an upfront entry fee and subsequently prices every item exactly at its realized cost. This design aims to align the buyer's utility with the realized gains from trade: because the item payments exactly covers the seller's cost, the buyer's utility upon entering is exactly $\gft(v,c)$ minus the entry fee. Consequently, the seller extracts their profit entirely through the entry fee. This argument closely resembles the analysis in \cite{CDW21}.

\begin{definition}[Cost-price two-part tariff]\label{def:warmup_2pt}
Fix an entry fee $E\ge 0$.
The seller first asks the buyer to pay $E$.
If the buyer declines, no trade occurs.
If the buyer accepts, then each item $j$ is offered at price $p_j=c_j$, and the buyer may purchase any bundle $S\subseteq [m]$.
\end{definition}

Under \cref{def:warmup_2pt}, the buyer's utility from entering is
\[
\max_{S\subseteq [m]} \bigl(v(S)-c(S)\bigr)-E
=
\gft(v,c)-E,
\]
so the buyer enters whenever $\gft(v,c)\ge E$. Furthermore, whenever the buyer enters, the seller's profit is exactly $E$.

To extract surplus via a fixed entry fee, the seller must ensure that the realized core GFT exceeds this fee with high probability. This can be guaranteed if the core objective is tightly concentrated around its expectation. Because the item-wise core surpluses are truncated at $\tau$, the marginal contribution of any single item to the total core GFT is strictly bounded. This bounded-difference property allows us to formally bound the variance of the core objective using the Efron--Stein inequality \cite{ES81}.

\begin{lemma}\label{lem:warmup_core_concentration}Let $Z= \gft^C(w)$ and $\mu=\E[Z]$.Then$$\Var(Z)\le \tau\,\mu.$$\end{lemma}

\begin{proof}
To bound the variance of $Z$, we analyze how much the objective can drop when a single item is removed. 
For each item $j$, let $Z^{(-j)} = \max_{k\in[K]} \sum_{\ell\neq j} (w_\ell^k)^C$ denote the value of the core objective without item $j$. 
If we fix a tie-breaking rule and let $k^\star \in \arg\max_{k\in[K]} \sum_{\ell=1}^m (w_\ell^k)^C$, we have $Z = \sum_{\ell=1}^m (w_\ell^{k^\star})^C$.

We first establish two simple deterministic bounds on the drop $Z - Z^{(-j)}$. First, because clause $k^\star$ remains a valid choice even after deleting item $j$, the objective decreases by at most the removed item's contribution to this clause. We therefore have 
\begin{equation}\label{eq:differenceupperbound}
    0 \le Z - Z^{(-j)} \le (w_j^{k^\star})^C \le \tau.
\end{equation}

Second, summing these individual drops over all items bounds the total leave-one-out drop by $Z$ itself,
\begin{equation}\label{eq:drop_upper_bound}
    \sum_{j=1}^m (Z - Z^{(-j)}) \le \sum_{j=1}^m (w_j^{k^\star})^C = Z.
\end{equation}
$$$$

We now bound the variance using the Efron--Stein inequality \cite{ES81}. Let $Z_j'$ be the value of $Z$ after independently resampling all random variables associated with item $j$ (namely $a_j^1,\dots,a_j^K, c_j$). The standard Efron--Stein inequality states that$$\mathrm{Var}(Z) \le \frac{1}{2} \sum_{j=1}^m \mathbb{E}[(Z - Z_j')^2].$$

By symmetry, because $Z$ and $Z_j'$ are identically distributed, the expected squared positive difference perfectly equals the expected squared negative difference. This allows us to decompose the total expected squared difference as $\mathbb{E}[(Z - Z_j')^2] = 2\mathbb{E}[(Z - Z_j')_+^2]$. Substituting this identity into the standard inequality cleanly absorbs the factor of $1/2$, yielding the one-sided bound:
$$\mathrm{Var}(Z) \le \sum_{j=1}^m \mathbb{E}[(Z - Z_j')_+^2]$$
Because the resampled item can only add non-negative value to the objective without it, we know $Z_j' \ge Z^{(-j)}$, which implies $(Z - Z_j')_+ \le Z - Z^{(-j)}$. Substituting this bound in \cref{eq:differenceupperbound}, we get that $(Z - Z_j')_+^2 \le (Z - Z^{(-j)})^2 \le \tau(Z - Z^{(-j)})$.

Taking the expectation and summing over all items $j$ allows us to apply \cref{eq:drop_upper_bound} we get:
\[\Var(Z) \le \tau \,\E\!\left[\sum_{j=1}^m \bigl(Z - Z^{(-j)}\bigr)\right] \le \tau\,\E[Z] = \tau\,\mu. \qedhere\]
\end{proof}

With the variance of the core GFT bounded by its mean and the truncation threshold, we can now guarantee that the realized surplus drops significantly below its expectation only rarely. By applying the Paley-Zygmund inequality, we show that as long as the expected core surplus $\mu$ is sufficiently large relative to the threshold $\tau$, an entry fee set to $\mu/4$ will be accepted by the buyer with high probability.

\begin{lemma}\label{lem:warmup_core_entry_fee}
Fix any $\beta>0$.
If $\mu\ge \beta\tau$, then the cost-price two-part tariff with entry fee $E=\mu/4$ earns expected seller profit at least
\[
\frac{9\beta}{64(\beta+1)}\,\mu.
\]
\end{lemma}

\begin{proof}
Let $Z=\gft^C(w)$.
Since $w_j^k\ge (w_j^k)^C$ for every $j$ and $k$, we have $\gft(v,c)\ge Z$ pointwise. Therefore the buyer certainly enters the tariff whenever $Z\ge \mu/4$.
By the Paley--Zygmund inequality,
\[
\Pr\!\left[Z\ge \frac{\mu}{4}\right]
\ge
\left(1-\frac14\right)^2 \frac{\mu^2}{\E[Z^2]}
=
\frac{9}{16}\cdot \frac{\mu^2}{\Var(Z)+\mu^2}.
\]
Using \cref{lem:warmup_core_concentration}, we obtain
\[
\Pr\!\left[Z\ge \frac{\mu}{4}\right]
\ge
\frac{9}{16}\cdot \frac{\mu^2}{\mu^2+\tau\mu}
=
\frac{9}{16}\cdot \frac{\mu}{\mu+\tau}.
\]
If $\mu\ge \beta\tau$, then $\mu/(\mu+\tau)\ge \beta/(\beta+1)$, so
\[
\Pr\!\left[Z\ge \frac{\mu}{4}\right]
\ge
\frac{9}{16}\cdot \frac{\beta}{\beta+1}.
\]
Whenever the buyer enters, the seller earns exactly the entry fee.
Therefore the expected seller profit is at least
\[
\frac{\mu}{4}\cdot \frac{9}{16}\cdot \frac{\beta}{\beta+1}
=
\frac{9\beta}{64(\beta+1)}\,\mu. \qedhere
\]
\end{proof}

Since the cost-price two-part tariff is itself a feasible seller-side mechanism, \cref{lem:warmup_core_entry_fee} immediately implies that the seller-profit-optimal mechanism earns at least the same profit whenever $\mu\ge \beta\tau$.

\subsection{Approximating the Tail via the Unit-Demand Reduction}\label{subsec:warmup_tail_approx}

We now turn to bounding the expected tail GFT, $\E[\gft^T(w)]$. As established in Section \ref{subsec:warmup_core_tail}, restricting feasible allocations to the one-trade family $\mathcal{F}^1$ forces the multi-dimensional XOS buyer to behave exactly as a standard unit-demand buyer with independent item values.

Our argument proceeds in two phases. First, we upper bound the expected tail GFT using the optimal single-item trading benchmark, $\gft^{*,\mathrm{ud}}_{\mathcal F^{1}}$. Second, we approximate this unit-demand benchmark using truthful posted-price menus in the original bilateral-trade instance, leveraging the corresponding single-dimensional "copies" instance and prophet inequalities.

Let $\mathcal I^{\mathrm{copies}}_{1}$ denote the copies instance associated with the one-trade family $\mathcal F^{1}$, as defined in \cref{def:copies_instance}.
By \cref{cor:copies_profit_approx},
\[
\frac12 \Pi^*_S(\mathcal I^{\mathrm{copies}}_{1})
+
\frac12 \Pi^*_B(\mathcal I^{\mathrm{copies}}_{1})
\ge
\frac{1}{3.15}\,\gft^{*,\mathrm{ud}}_{\mathcal F^{1}}.
\]

To upper bound the tail GFT, we must compare the sum of all large item surpluses to the single largest one. The following lemma provides a general probabilistic tool for this: it shows that for independent non-negative random variables, the expectation of their sum can be tightly bounded by the expectation of their maximum, provided the variables are rarely positive simultaneously. All the missing proofs can be found in \cref{app:warmup}
\begin{lemma} \label{lem:tail_expectation_bound}
Let $X_1, X_2, \dots, X_n$ be independent, non-negative random variables. The following inequality holds:
$$ \mathbb{E}\left[\sum_{i=1}^n X_i\right] \le \mathbb{E}\left[\max_i X_i\right] + \left( \sum_{i=1}^n \Pr[X_i > 0] \right) \cdot \mathbb{E}\left[\sum_{i=1}^n X_i\right] $$
\end{lemma}

\begin{lemma}\label{lem:warmup_xos_tail_bound}
For the one-trade benchmark $\gft^{*,\mathrm{ud}}_{\mathcal F^{1}}$ and the induced tail objective $\gft^T(w)$, the expected tail gains from trade over the joint distribution of $(v,c)$ satisfies:
\[
\E[\gft^T(w)]
\le
\frac{1}{1-\ln 2}\,\gft^{*,\mathrm{ud}}_{\mathcal F^{1}}.
\]
\end{lemma}

\begin{proof}
For each item $j$, let $Y_j = (v_j-c_j)^+$, let $X_j = (Y_j-\tau)^+$, and let $p_j = \Pr[Y_j>\tau]$.
Since $v_j=\max_k a_j^k$, we have $Y_j=\max_{k\in [K]} w_j^k$. Therefore, for every realization,
\[
\gft^T(w)
=
\max_{k\in [K]} \sum_{j=1}^m (w_j^k)^T
\le
\sum_{j=1}^m X_j.
\]
Hence, $\E[\gft^T(w)] \le \E\!\left[\sum_{j=1}^m X_j\right]$.

Next, we bound $\sum_{j=1}^m p_j$.
Since $\tau = 2\gft^{*,\mathrm{ud}}_{\mathcal F^{1}}$, applying Markov's inequality yields
\[
\Pr\!\left[\max_{j\in [m]} Y_j > \tau\right]
\le
\frac{\E[\max_j Y_j]}{\tau}
=
\frac{\gft^{*,\mathrm{ud}}_{\mathcal F^{1}}}{2\gft^{*,\mathrm{ud}}_{\mathcal F^{1}}}
=
\frac12.
\]
On the other hand, because the variables $Y_j$ are independent, we have
\[
\Pr\!\left[\max_{j\in [m]} Y_j > \tau\right]
=
1-\prod_{j=1}^m (1-p_j)
\ge
1-e^{-\sum_{j=1}^m p_j}.
\]
Combining these bounds gives $1-e^{-\sum_j p_j}\le \frac12$, which simplifies to $\sum_{j=1}^m p_j \le \ln 2$.

We now apply \cref{lem:tail_expectation_bound} directly to the independent non-negative random variables $X_1,\dots,X_m$. Because $\Pr[X_j>0] = \Pr[Y_j>\tau] = p_j$, and we established $\sum_{j=1}^m p_j \le \ln 2$ we have:
\[
\E\!\left[\sum_{j=1}^m X_j\right]
\le
\E[\max_j X_j]
+
(\ln 2)\,\E\!\left[\sum_{j=1}^m X_j\right] \qquad \Rightarrow \qquad \E\!\left[\sum_{j=1}^m X_j\right]
\le
\frac{1}{1-\ln 2}\,\E[\max_j X_j].
\]
% Rearranging terms, we obtain
% \[

% \]
Finally, because $X_j\le Y_j$ pointwise for every $j$, we have $\E[\max_j X_j] \le \E[\max_j Y_j] = \gft^{*,\mathrm{ud}}_{\mathcal F^{1}}$. Combining this with our earlier bounds proves the lemma.
\end{proof}

\paragraph{Implementing the copies benchmarks in the original market.}
The previous lemma reduces the tail bounding problem to the one-trade benchmark $\gft^{*,\mathrm{ud}}_{\mathcal F^{1}}$. By \cref{cor:copies_profit_approx}, this benchmark is bounded by the sum of the optimal single-dimensional profits in the copies instance: $\Pi^*_S(\mathcal I^{\mathrm{copies}}_{1})$ and $\Pi^*_B(\mathcal I^{\mathrm{copies}}_{1})$. 

We now show that both of these virtual surplus benchmarks can be approximated up to a factor of $2$ by truthful, one-item menus in the original multi-dimensional market. The core engine for this reduction is Samuel-Cahn's prophet inequality \cite{Cahn84}, which yields a uniform threshold that translates into an online, sequential posted-price mechanism for the single-dimensional copies instance. Crucially, by the equivalence lemma in \cite{CHMS10}, offering this exact same price vector as a simultaneous, take-it-or-leave-it menu to the original multi-dimensional buyer preserves the approximation guarantee.

\paragraph{Seller one-item menu.}
We first design the seller-side mechanism. Our goal is to extract a constant fraction of the seller's optimal single-dimensional profit benchmark, $\Pi^*_S(\mathcal I^{\mathrm{copies}}_{1})$, by constructing a truthful, simultaneous posted-price menu for the original buyer. We achieve this through a two-step reduction. In the first step, we analyze the single-dimensional copies instance and utilize prophet inequalities to compute an optimal virtual threshold. This threshold defines an order-oblivious sequential posted-price mechanism (OPM) that guarantees a constant fraction of the benchmark regardless of the order in which the pseudo-buyers are evaluated. In the second step, we translate this sequential pricing rule into a simultaneous menu for the original multi-dimensional market. Because the OPM's profit guarantee holds against any adversarial arrival sequence, offering the exact same price vector as a simultaneous, take-it-or-leave-it menu preserves the profit guarantee when the unit-demand buyer deterministically selects their utility-maximizing item.

\begin{lemma}\label{lem:warmup_tail_seller_threshold}
For every realized cost vector $c$, there exists a common threshold $\lambda(c)\ge 0$ and a possibly randomized price vector
$\mathbf p(c)=(p_1(c),\dots,p_m(c))$
such that the induced order-oblivious posted-price mechanism in $\mathcal I^{\mathrm{copies}}_{1}$ earns conditional expected seller profit at least
\[
\frac12 \,\E_v\!\left[\max_{j\in [m]} \bigl(\tilde\varphi_j(v_j)-c_j\bigr)^+ \,\middle|\, c\right].
\]
In the regular case one may take
\[
p_j(c)=\inf\{t : \tilde\varphi_j(t)-c_j\ge \lambda(c)\}.
\]
In the general ironed case, if the threshold intersects a flat ironed interval, one randomizes between the two boundary prices exactly as in prior work.
\end{lemma}

\begin{proof}
Fix the seller's realized cost vector $c$, and for each item $j$ define $X_j = (\tilde\varphi_j(v_j)-c_j)^+$.
Because the singleton values $v_1,\dots,v_m$ are independent, the random variables $X_1,\dots,X_m$ are independent and nonnegative.
Samuel-Cahn's prophet inequality therefore gives a threshold $\lambda(c)$ such that the rule that accepts the first $X_j$ exceeding $\lambda(c)$ obtains expected reward at least
\[
\frac12 \,\E_v\!\left[\max_{j\in [m]} X_j \,\middle|\, c\right].
\]

We can translate this virtual threshold back into posted prices.
For regular distributions, the price
\[
p_j(c)=\inf\{t : \tilde\varphi_j(t)-c_j\ge \lambda(c)\}
\]
makes buyer $j$ accept exactly when $X_j\ge \lambda(c)$.
For ironed distributions, a flat virtual-value interval may appear at the threshold, and then we randomize between the two boundary prices so that the acceptance probability matches the threshold rule exactly.

In the copies instance, each pseudo-buyer is single-parameter.
Hence, by Myerson's payment identity (with the standard boundary randomization on ironed intervals), the expected seller profit of the resulting posted-price mechanism equals its expected ironed virtual surplus.
Therefore the expected seller profit of this order-oblivious posted-price mechanism is at least
\[
\frac12 \,\E_v\!\left[\max_{j\in [m]} \bigl(\tilde\varphi_j(v_j)-c_j\bigr)^+ \,\middle|\, c\right]. \qedhere
\]
\end{proof}

\begin{proposition}\label{thm:warmup_tail_seller_menu}
There exists a DSIC-B and ex-post IR seller-posted-price mechanism $\cM_S^{\mathrm{tail}}$ for the original bilateral-trade instance such that the buyer may purchase at most one item and
\[
\Pi_S(\cM_S^{\mathrm{tail}})
\ge
\frac12\,\Pi^*_S(\mathcal I^{\mathrm{copies}}_{1}).
\]
\end{proposition}

\begin{proof}
For each realized cost vector $c$, \cref{lem:warmup_tail_seller_threshold} gives a price vector $\mathbf p(c)$ that is a $2$-approximate order-oblivious posted-price mechanism in the copies instance.
By the approximation-preserving reduction (Theorem 4 in ~\cite{CHMS10}), the same price vector induces a truthful posted-price mechanism for the original unit-demand agent with the same approximation guarantee.
In the present warm-up instance this mechanism is simply a simultaneous menu of prices from which the buyer may choose at most one item.
Truthfulness and individual rationality are immediate: the buyer either declines all items or chooses the item that maximizes her utility.
Taking expectation over $c$ gives
\[
\Pi_S(\cM_S^{\mathrm{tail}})
\ge
\frac12\,\Pi^*_S(\mathcal I^{\mathrm{copies}}_{1}).
\qedhere
\]
\end{proof}

\paragraph{Buyer one-item menu.} 
We now design the buyer-side mechanism. The argument here is the exact procurement analogue of the seller's mechanism design above. Because the original seller is additive across items, the buyer can treat them as a unit-supply agent who chooses at most one item to trade. We defer the argument to \cref{app:warmup}.

\begin{proposition}\label{thm:warmup_tail_buyer_menu}
There exists a DSIC-S and ex-post IR procurement mechanism $\cM_B^{\mathrm{tail}}$ for the original bilateral-trade instance such that the seller may choose at most one item to trade and
\[
\Pi_B(\cM_B^{\mathrm{tail}})
\ge
\frac12\,\Pi_B^*(\mathcal I^{\mathrm{copies}}_{1}).
\]
\end{proposition}

% \begin{proof}
% The proof is the exact analogue of \cref{thm:warmup_tail_seller_menu}, with values and costs interchanged.
% For each realized singleton-value vector $v$, \cref{lem:warmup_tail_buyer_threshold} gives a procurement-price vector $\mathbf q(v)$ that is a $2$-approximate order-oblivious posted-price mechanism in the seller copies instance.
% The same reduction argument applies verbatim after viewing the original seller as the unit-supply agent who chooses which item, if any, to trade.
% In the present warm-up instance, the resulting mechanism is simply a simultaneous procurement menu, and it is DSIC and IR because the seller either rejects all offers or chooses the item that maximizes her utility $q_j(v)-c_j$.
% Taking expectation over $v$ yields
% \[
% \Pi_B(\cM_B^{\mathrm{tail}})
% \ge
% \frac12\,\Pi_B^*(\mathcal I^{\mathrm{copies}}_{1}).
% \qedhere
% \]
% \end{proof}

\begin{corollary}[Tail approximation via truthful one-item menus]\label{thm:warmup_profit_to_ud_gft}
Let $\cM_S^{\mathrm{tail}}$ and $\cM_B^{\mathrm{tail}}$ be the mechanisms from
\cref{thm:warmup_tail_seller_menu,thm:warmup_tail_buyer_menu}.
Then
\[
\frac12\,\Pi_S(\cM_S^{\mathrm{tail}})
+
\frac12\,\Pi_B(\cM_B^{\mathrm{tail}})
\ge
\frac{1}{6.3}\,\gft^{*,\mathrm{ud}}_{\mathcal F^{1}}.
\]
\end{corollary}

\begin{proof}
By \cref{thm:warmup_tail_seller_menu,thm:warmup_tail_buyer_menu},
\[
\frac12\,\Pi_S(\cM_S^{\mathrm{tail}})
+
\frac12\,\Pi_B(\cM_B^{\mathrm{tail}})
\ge
\frac14\,\Pi^*_S(\mathcal I^{\mathrm{copies}}_{1})
+
\frac14\,\Pi^*_B(\mathcal I^{\mathrm{copies}}_{1}).
\]
Applying \cref{cor:copies_profit_approx} to the copies instance associated with $\mathcal F^{1}$ gives
\[
\frac14\,\Pi^*_S(\mathcal I^{\mathrm{copies}}_{1})
+
\frac14\,\Pi^*_B(\mathcal I^{\mathrm{copies}}_{1})
\ge
\frac{1}{2\cdot 3.15}\,\gft^{*,\mathrm{ud}}_{\mathcal F^{1}}
=
\frac{1}{6.3}\,\gft^{*,\mathrm{ud}}_{\mathcal F^{1}}. \qedhere
\]
\end{proof}

\subsection{Putting Everything Together}\label{subsec:warm_up_everything}

We are now ready to combine the core and tail analyses and prove the main result of this section.
Let $\mathcal{I}$ denote the original multi-dimensional bilateral-trade instance. We formally define our optimal benchmarks as follows: let $\Pi^*_S(\mathcal{I})$ and $\Pi^*_B(\mathcal{I})$ denote the expected seller profit and buyer profit achievable by the optimal BIC-B (resp BIC-S) and IIR mechanisms in $\mathcal{I}$, respectively.

Analogous to the single-dimensional storyline (see \cref{thm:sd_profit_approx}), we first state a bound based on the optimal buyer and seller profits. We note that this does not require the mediator (mechanism designer) to compute these complex optimal mechanisms. Instead, we simply delegate all pricing power to either the buyer or the seller with equal probability, relying on the chosen agent to rationally compute and deploy their own profit-maximizing auction.

\begingroup
\renewcommand{\thetheorem}{\ref{thm:BTmainoptimal}}
\begin{theorem}
\mainresultBT
\end{theorem}
\addtocounter{theorem}{-1}
\endgroup

\begin{proof}
We first verify the properties of $\cM^{\mathrm{opt}}_{\mathrm{rand}}$. By definition, the optimal bilateral-trade mechanisms for $\mathcal{I}$ are BIC, ex-post IR, and WBB. Because $\cM^{\mathrm{opt}}_{\mathrm{rand}}$ simply randomizes over these two valid mechanisms using a coin flip that is independent of the agents' types and reports, it inherits all three properties.

Because every mechanism in question is weakly budget balanced, the expected GFT is always strictly lower bounded by the expected profit. Hence, we can lower bound the GFT by evaluating the expected profit of the randomized mechanism: $\frac{1}{2}\Pi^*_S(\mathcal{I}) + \frac{1}{2}\Pi^*_B(\mathcal{I})$.

Let
\[
\gamma = \frac{1}{1-\ln 2} \approx 3.26.
\]
By \cref{lem:warmup_core_tail_upper,lem:warmup_xos_tail_bound}, the first-best gains from trade is upper bounded by:
\begin{equation}\label{eq:warmup_master_upper_bound}
\gft^*
\le
\mu + \gamma\,\gft^{*,\mathrm{ud}}_{\mathcal F^{1}}.
\end{equation}
This inequality dictates our proof template: either the core term $\mu$ is large, in which case the entry-fee mechanism extracts sufficient profit, or the core is small, in which case the one-trade tail benchmark controls the entire instance. We split the analysis into two cases based on the relative size of the core.

\smallskip
\noindent\emph{Case 1: $\mu \ge 1.75\tau = 3.5\,\gft^{*,\mathrm{ud}}_{\mathcal F^{1}}$.}
Applying \cref{lem:warmup_core_entry_fee} with $\beta=7/4$, the cost-price two-part tariff earns expected seller profit at least:
\[
\frac{9(7/4)}{64(7/4+1)}\mu = \frac{63}{704}\mu.
\]
Since this tariff is a feasible seller-side mechanism in $\mathcal{I}$, we have $\Pi^*_S(\mathcal{I}) \ge \frac{63}{704}\mu$. 
Therefore, the expected profit of our randomized mechanism is at least $\frac{1}{2}\left(\frac{63}{704}\mu\right) = \frac{63}{1408}\mu$.
Combining \cref{eq:warmup_master_upper_bound} with the case assumption yields $\gft^* \le (1+\gamma/3.5)\mu$. Substituting this lower bound for $\mu$ provides the final guarantee:
\[
\frac{1}{2}\Pi^*_S(\mathcal{I}) \ge \frac{63}{1408}\mu \ge \frac{63}{1408} \cdot \frac{\gft^*}{1+\gamma/3.5} \ge \frac{\gft^*}{44}.
\]

\smallskip
\noindent\emph{Case 2: $\mu < 1.75\tau = 3.5\,\gft^{*,\mathrm{ud}}_{\mathcal F^{1}}$.}
Because the expected core is small, \cref{eq:warmup_master_upper_bound} implies the tail benchmark bounds the total GFT:
\[
\gft^* \le (3.5+\gamma)\,\gft^{*,\mathrm{ud}}_{\mathcal F^{1}}.
\]
To lower bound the optimal profits in the original instance $\mathcal{I}$, we observe that the one-item menus $\cM_S^{\mathrm{tail}}$ and $\cM_B^{\mathrm{tail}}$ are valid, feasible mechanisms for the original multi-dimensional agents. Therefore, the optimal mechanisms trivially dominate the profits of these explicit menus: $\Pi^*_S(\mathcal{I}) \ge \Pi_S(\cM_S^{\mathrm{tail}})$ and $\Pi^*_B(\mathcal{I}) \ge \Pi_B(\cM_B^{\mathrm{tail}})$. Applying \cref{thm:warmup_profit_to_ud_gft} and substituting the GFT bound yields:
\[
\frac{1}{2}\Pi^*_S(\mathcal{I}) + \frac{1}{2}\Pi^*_B(\mathcal{I})
\ge
\frac{1}{2}\Pi_S(\cM_S^{\mathrm{tail}}) + \frac{1}{2}\Pi_B(\cM_B^{\mathrm{tail}})
\ge
\frac{1}{6.3}\,\gft^{*,\mathrm{ud}}_{\mathcal F^{1}}
\ge
\frac{\gft^*}{6.3(3.5+\gamma)}
\ge
\frac{\gft^*}{44}.
\]

In both cases, uniformly randomizing between the optimal mechanisms guarantees at least a $1/44$ fraction of the first-best GFT.
\end{proof}

The previous theorem establishes a constant-factor approximation assuming the chosen agent can compute and deploy their optimal bilateral-trade mechanism. However, optimal multi-dimensional mechanisms are complex and often computationally intractable. 

The power of our explicitly constructed mechanisms—the cost-price two-part tariff and the one-item posted-price menus—is that they bypass this complexity entirely. They serve as simple, computationally efficient \emph{advice}. When pricing power is delegated, the chosen agent does not need to solve an intractable mechanism design problem; they can simply deploy our recommended candidate auction. Because this explicitly defined advice is always available to them, whatever utility-maximizing mechanism they ultimately decide to run will yield an expected profit at least as large as the profit guaranteed by our advice. By evaluating these simple candidate mechanisms, we obtain a constructive constant-factor approximation.

\begingroup
\renewcommand{\thetheorem}{\ref{thm:btmainsuboptimal}}
\begin{theorem}
\mainresultBTsub
\end{theorem}
\addtocounter{theorem}{-1}
\endgroup

\begin{proof}
By linearity of expectation, the expected profit of the grand mechanism $\cM_{\mathrm{rand}}^{\mathrm{simple}}$ is:$$\Pi(\cM_{\mathrm{rand}}^{\mathrm{simple}}) = \frac{1}{4}\Pi_S(\cM_E) + \frac{1}{4}\Pi_S(\cM_S^{\mathrm{tail}}) + \frac{1}{2}\Pi_B(\cM_B^{\mathrm{tail}}).$$
We use the same decomposition bounds as above.
\smallskip
\noindent\emph{Case 1: $\mu \ge 1.75\tau = 3.5\,\gft^{*,\mathrm{ud}}_{\mathcal F^{1}}$.}
By \cref{lem:warmup_core_entry_fee} with $\beta=7/4$, the entry-fee mechanism $M_E$ earns expected seller profit at least $\frac{63}{704}\mu$.

Ignoring the non-negative tail profits, the protocol guarantees:
$$\Pi(M_{\mathrm{rand}}^{\mathrm{simple}}) \ge \frac{1}{4} \left(\frac{63}{704}\mu\right) = \frac{63}{2816}\mu.$$
Combining \cref{eq:warmup_master_upper_bound} with the case assumption yields $\gft^* \le (1+\gamma/3.5)\mu$. Substituting this lower bound for $\mu$ provides the final guarantee:$$\Pi(M_{\mathrm{rand}}^{\mathrm{simple}}) \ge \frac{63}{2816}\mu \ge \frac{63}{2816} \cdot \frac{\gft^*}{1+\gamma/3.5} \ge \frac{\gft^*}{88}.$$
\smallskip
\noindent\emph{Case 2: $\mu < 1.75\tau = 3.5\,\gft^{*,\mathrm{ud}}_{\mathcal F^{1}}$.}
Because the expected core is small, \cref{eq:warmup_master_upper_bound} implies the tail benchmark bounds the total GFT:
$$\gft^* \le (3.5+\gamma)\,\gft^{*,\mathrm{ud}}_{\mathcal F^{1}}.$$
We now evaluate the expected profit of the candidate mechanisms. Let $S_{\mathrm{tail}} = \Pi_S(\cM_S^{\mathrm{tail}})$ and $B_{\mathrm{tail}} = \Pi_B(\cM_B^{\mathrm{tail}})$. By ignoring the non-negative core profit from $M_E$, we have:
$$\Pi(M_{\mathrm{rand}}^{\mathrm{simple}}) \ge \frac{1}{4}S_{\mathrm{tail}} + \frac{1}{2}B_{\mathrm{tail}}.$$
Because $B_{\mathrm{tail}} \ge 0$, we can cleanly factor out $1/2$ and apply the bound established in \cref{thm:warmup_profit_to_ud_gft}. Chaining this with our GFT upper bound yields:
$$\Pi(M_{\mathrm{rand}}^{\mathrm{simple}}) \ge \frac{1}{2}\left(\frac{1}{2}S_{\mathrm{tail}} + B_{\mathrm{tail}}\right) \ge \frac{1}{2}\left(\frac{1}{2}S_{\mathrm{tail}} + \frac{1}{2}B_{\mathrm{tail}}\right) \ge \frac{1}{12.6}\,\gft^{*,\mathrm{ud}}_{\mathcal F^{1}} \ge \frac{\gft^*}{12.6(3.5+\gamma)} \ge \frac{\gft^*}{88}.$$

In both cases, the randomized protocol guarantees at least a $1/88$ fraction of the first-best GFT.\end{proof}

\section{Multiple XOS Buyers}

In this section, we extend our results from the bilateral-trade warm-up to our general setting: a market with $n$ multi-dimensional buyers, each with an XOS valuation over independent items, and a single additive seller with independent item costs. We establish that by delegating pricing power to either the seller or the buyers, we can still achieve a constant-factor approximation to the first-best gains from trade (GFT).

In particular, we propose a mechanism, $\mathcal M^{\mathrm{final}}$, as a randomized delegation protocol: with probability $p$, we delegate pricing power to the seller; with probability $1-p$, we deploy the randomized buyer-side mechanism, denoted as $\mathcal M_B^{\mathrm{copy}}$, from \cref{thm:main_tailbuyer}. As in the bilateral trade setting, we first state a bound assuming the seller computes their optimal BIC and IIR mechanism.

\begin{theorem}\label{thm:main}
\mainresult
% Consider any Bayesian multidimensional instance with $n$ XOS buyers and one additive
% seller over independent items. the mechanism $\mathcal M^{\mathrm{final}}$ is BIC, interim IR, and ex-ante weakly budget balanced. Moreover, it satisfies:$$\gft(\mathcal M^{\mathrm{final}}) \ge \alpha\,\gft^*,$$where $\alpha \approx 1/1033$.
\end{theorem}

The core--tail decomposition itself extends cleanly from the bilateral-trade warm-up: once we replace the one-trade benchmark by the unit-demand matching benchmark $\gft^{*,\mathrm{ud}}_{\mathcal F^m}$, the same truncation-based decomposition goes through. The main new work in the multiple-buyer case is therefore not the decomposition, but the mechanism design needed to implement the core and tail bounds under truthfulness and budget balance.

There are two additional difficulties.

\begin{itemize}
    \item \textbf{Tail approximation now requires a new buyer-side mechanism.}
    In the single XoS buyer case, once we reduce the tail to the corresponding unit-demand benchmark, the relevant copy-instance profit benchmarks can be implemented by simple one-agent posted-price menus. With multiple buyers, the analogous buyer-side route becomes substantially harder. A natural idea is to maximize total buyer profit and then charge VCG payments to make the buyers truthful. However, as our counterexample (\cref{obs:vcg-deficit}) in the technical overview shows, this approach is not always weakly budget balanced. Thus, for the tail, we cannot rely on an ``optimal auction + VCG'' template. Instead, we must design from scratch a truthful, interim IR, weakly budget-balanced buyer-side mechanism that still captures a constant fraction of the buyer-side copy benchmark.

    \item \textbf{A static entry fee no longer suffices for the core.}
    In the single-buyer warm-up, pricing items at cost turns the buyer's utility from entering into the realized gains from trade, so a single entry fee is enough to extract a constant fraction of the concentrated core. With multiple buyers, this breaks down because buyers arrive sequentially and the set of available items depends on earlier purchases. The continuation utility of buyer $i$ is therefore determined by the remaining items at her arrival, and a static fee cannot simultaneously extract substantial surplus when many attractive items remain and still preserve participation when the remaining items are sparse. To handle this dependence on the remaining items, we adopt the ASPE framework of Cai and Zhao~\cite{CZ17}: anonymous posted prices control contention across buyers, while the entry fee $\delta_i(S)$ is allowed to depend on the currently available set $S$.
\end{itemize}

With these two issues isolated, the rest of the proof follows the same blueprint as in the bilateral-trade case. In \cref{subsec:multi_core_tail} we extend the decomposition using the matching benchmark $\gft^{*,\mathrm{ud}}_{\mathcal F^m}$. We then approximate the core via the ASPE mechanism, and approximate the tail via the copy-instance benchmarks together with the new buyer-side truthful mechanism. Combining these pieces yields the final constant-factor approximation theorem.

\subsection{The Core--Tail Decomposition}\label{subsec:multi_core_tail}

Similar to the bilateral-trade setting, we first decompose the first-best GFT for every realization into a \emph{core} part and a \emph{tail} part.

For a realized type profile $t$, let $a_{ij}^k$ denote buyer $i$'s realized value for item $j$ under clause $k$ (formally, $a_{ij}^k = v_{ij}^{(k)}(t_{ij})$). For every buyer $i$, clause $k$, and item $j$, define the realized clause-wise surplus:
\[
w_{ij}^k = (a_{ij}^k - c_j)^+, \qquad\text{where } (x)^+ = \max\{x,0\}.
\]

Because each buyer's valuation is XOS, the total gains from trade can be represented by choosing the optimal allocation of items and the optimal clause for each buyer.

\begin{observation}\label{obs:multi_positive}
For every realization $(t,c)$, the first-best GFT can be written as:
\[
\max_{(S_1,\dots,S_n)} \gft(t,c;S_1,\dots,S_n) = \max_{(S_1,\dots,S_n)}\ \max_{k_1,\dots,k_n} \sum_{i=1}^n\sum_{j\in S_i} w_{ij}^{k_i}.
\]
\end{observation}

\paragraph{The matching benchmark and the threshold.}
To define the truncation threshold for multiple buyers, we naturally extend the one-trade benchmark to a bipartite matching benchmark. Recall the matching feasible family:
\[
\mathcal{F}^m = \bigl\{A \subseteq [n] \times [m] : |A \cap ([n] \times \{j\})| \le 1 \text{ for all } j \in [m]\bigr\}.
\]
This constraint ensures that each buyer is allocated at most one item. Under $\mathcal{F}^m$, the multi-dimensional XOS setting collapses: the relevant valuation for buyer $i$ and item $j$ simplifies to the singleton value $v_i(\{j\}) = \max_{k \in [K]} a_{ij}^k$. Because item types are drawn independently, these singleton values act exactly as independent item values in a standard unit-demand market.

By \cref{def:optimal_ud_gft}, we define the optimal unit-demand matching GFT benchmark as:
\begin{equation}
\gft^{*,\mathrm{ud}}_{\mathcal{F}^m} = \E_{t,c}\!\left[\max_{M \in \mathcal{F}^m} \sum_{(i,j)\in M} \bigl(v_i(t_i, \{j\})-c_j\bigr)^+\right].
\end{equation}

We set our global truncation threshold relative to this matching benchmark:
\[
\tau := 2\,\gft^{*,\mathrm{ud}}_{\mathcal{F}^m}.
\]

\begin{definition}[Core and tail benchmarks]\label{def:global-core-tail}
Fix the threshold $\tau = 2\,\gft^{*,\mathrm{ud}}_{\mathcal{F}^m}$. For every buyer $i$, clause $k$, and item $j$, define the truncated core and the excess tail by:
\[
(w_{ij}^k)^C = \min\{w_{ij}^k,\tau\}, \qquad \text{and} \qquad (w_{ij}^k)^T = \bigl(w_{ij}^k-\tau\bigr)^+.
\]
So that $w_{ij}^k = (w_{ij}^k)^C + (w_{ij}^k)^T$ pointwise. For every realization $(t,c)$, we define the core and tail objectives as:
\begin{align*}
\gft^C(t,c) &= \max_{(S_1,\dots,S_n)}\ \max_{k_1,\dots,k_n} \sum_{i=1}^n \sum_{j\in S_i} (w_{ij}^{k_i})^C, \\
\gft^T(t,c) &= \max_{(S_1,\dots,S_n)}\ \max_{k_1,\dots,k_n} \sum_{i=1}^n \sum_{j\in S_i} (w_{ij}^{k_i})^T.
\end{align*}
We write the expected core contribution as $\mu=\E[\gft^C]$. For a fixed cost vector $c$, we also write $\mu(c)=\E_t[\gft^C(t,c)]$.
\end{definition}

As in the single buyer case, splitting the clause-wise surpluses directly bounds the total expected GFT.

\begin{lemma}\label{lem:global_split}
For every realization $(t,c)$,
\[
\max_{(S_1,\dots,S_n)} \gft(t,c;S_1,\dots,S_n) \le \gft^C(t,c)+\gft^T(t,c).
\]
Consequently,
\[
\gft^* \le \mu + \E[\gft^T].
\]
\end{lemma}

\begin{proof}
Let $(S_1^*,\dots,S_n^*)$ and $k_1^*,\dots,k_n^*$ be the partition and clauses that attain the maximum first-best GFT in \cref{obs:multi_positive}. Then,
\[
\max_{(S_1,\dots,S_n)} \gft(t,c;S_1,\dots,S_n) = \sum_{i=1}^n \sum_{j\in S_i^*} \bigl((w_{ij}^{k_i^*})^C + (w_{ij}^{k_i^*})^T\bigr).
\]
Because the allocation $(S_1^*,\dots,S_n^*)$ is feasible, the core sum is strictly upper-bounded by the unconstrained maximum $\gft^C(t,c)$, and the tail sum is strictly upper-bounded by $\gft^T(t,c)$. Taking the expectations of both sides proves the second claim.
\end{proof}

\subsection{Proof Overview}

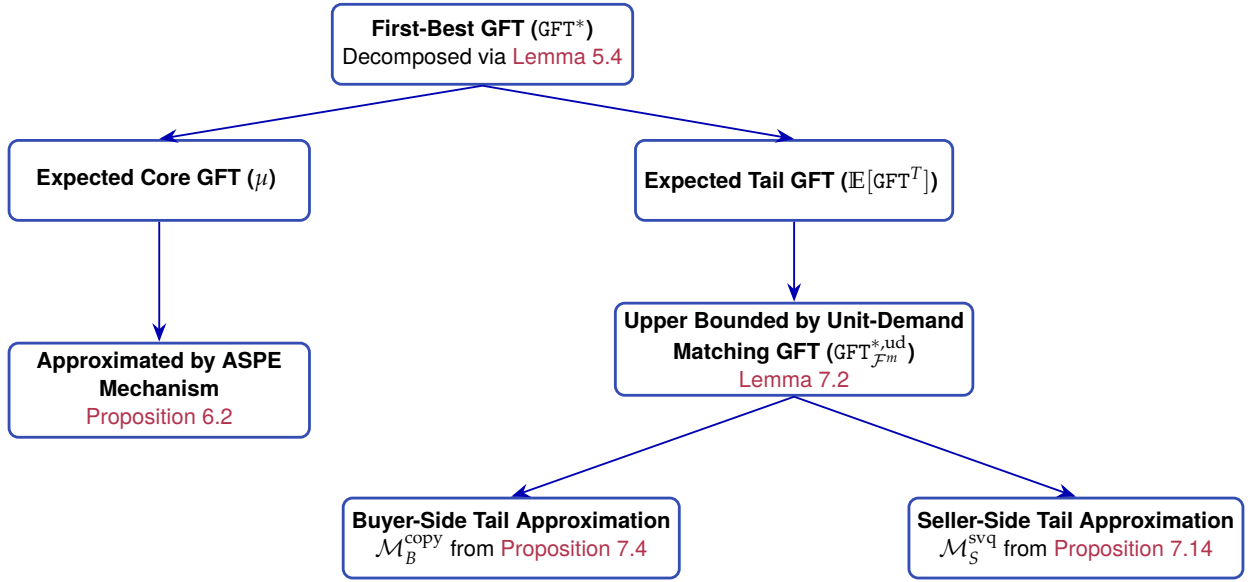
\begin{figure}[h]
    \centering
    \resizebox{\textwidth}{!}{
\begin{tikzpicture}[
    node distance=0.6cm and 0.3cm,
    box/.style={
        rectangle,
        rounded corners,
        % draw=blue!70!black,
        draw=citec,
        % fill=blue!5,
        very thick,
        align=center,
        minimum width=4.5cm,
        minimum height=1.2cm,
        font=\small\sffamily
    },
    arrow/.style={-{Stealth[scale=1.2]}, thick, blue!70!black},
    label_font/.style={font=\scriptsize\itshape, blue!80!black}
]

% Root node
\node (root) [box] {\textbf{First-Best GFT ($\gft^*$)} \\ Decomposed via \cref{lem:global_split}};

% Level 1: Core vs Tail
\node (core) [box, below left=0.8cm and 0.3cm of root] {\textbf{Expected Core GFT ($\mu$)} 
% \\ \cref{lem:global_split
};
\node (tail) [box, below right=0.8cm and 0cm of root] {\textbf{Expected Tail GFT ($\mathbb{E}[\gft^T]$)} 
% \\ \cref{lem:global_split}
};

% Core Path
\node (aspe) [box, below=1.8cm of core] {\textbf{Approximated by ASPE} \\ \textbf{Mechanism} 
\\ \cref{prop:multi-core-aspe}
};

% Tail Path intermediate step
\node (ud_bound) [box, below=1.2cm of tail] {\textbf{Upper Bounded by Unit-Demand} \\ \textbf{Matching GFT ($\gft^{*,\mathrm{ud}}_{\mathcal{F}^m}$)} \\ \cref{lem:global_tail_bound}};

% Final Tail sub-branches
\node (buyer_tail) [box, below left=1.5cm and -1.0cm of ud_bound] {\textbf{Buyer-Side Tail Approximation} \\ $\cM_B^{\mathrm{copy}}$ from \cref{thm:main_tailbuyer}};
\node (seller_tail) [box, below right=1.5cm and -1.0cm of ud_bound] {\textbf{Seller-Side Tail Approximation} \\ $\cM_S^{\mathrm{svq}}$ from \cref{thm:sellertailmain}};

% Drawing arrows
\draw [arrow] (root.south) -- (core.north);
\draw [arrow] (root.south) -- (tail.north);
\draw [arrow] (core.south) -- (aspe.north);
\draw [arrow] (tail.south) -- (ud_bound.north);
\draw [arrow] (ud_bound.south) -- (buyer_tail.north);
\draw [arrow] (ud_bound.south) -- (seller_tail.north);

\end{tikzpicture}
}
 \caption{Proof Overview for Multiple XOS Buyer Case}
    \label{fig:placeholder}
\end{figure}
In \cref{sec:buyercore}, we propose an anonymous sequential posted price with entry fee (ASPE) mechanism, whose profit approximates the core GFT up to constant factor. In \cref{subsec:global_tail}, we approximate tail GFT by mechanisms that approximate the corresponding copies instance profit benchmarks. Finally, in \cref{sec:final}, we aggregate the result and prove our main theorem.
\newcommand{\UD}{\mathrm{UD}}
\newcommand{\FAV}{\mathrm{FAV}}
\newcommand{\NONFAV}{\mathrm{NONFAV}}
\newcommand{\CORE}{\mathrm{CORE}}
\newcommand{\TAIL}{\mathrm{TAIL}}
\newcommand{\POSTREV}{\mathrm{PostRev}}
\newcommand{\supp}{\mathrm{supp}}
\newcommand{\pos}[1]{\left(#1\right)^+}
\newcommand{\cF}{\mathcal F}
\newcommand{\cE}{\mathcal E}
\newcommand{\cA}{\mathcal A}
\newcommand{\1}{\mathbbm 1}
\newcommand{\Prb}{\Pr}
% \section{Multiple XOS buyers}

\section{Approximating the Core via an ASPE Mechanism}\label{sec:buyercore}

In this section, we show how to approximately extract the \emph{core} benchmark in the $n$-buyer setting.
Recall from \cref{def:global-core-tail} that, for a fixed seller cost vector $c$,
\[
\mu(c)=\E_t[\gft^C(t,c)]
\]
denotes the expected contribution of the core part.
For every fixed $c$, we construct a buyer-truthful mechanism whose expected seller profit captures a constant fraction of $\mu(c)$, up to an additive loss controlled by a fixed-cost unit-demand matching benchmark.

At first glance, one might hope to mimic the bilateral-trade case and charge each buyer a single static entry fee before offering items at posted prices.
In the multi-buyer setting, however, such a static fee is too coarse.
The reason is that buyers arrive sequentially, and a buyer's utility depends on the set of remaining items after earlier buyers have chosen their bundles.
A fee that does not depend on the remaining set of items cannot simultaneously extract substantial surplus when many attractive items remain and still allow participation when the remaining inventory is sparse.

To handle this dependence on the remaining inventory, we build on the \emph{Anonymous Sequential Posted-Pricing with Entry Fees} (ASPE) framework of~\cite{CZ17}.
Cai and Zhao use ASPE to extract revenue from a truncated one-sided XOS welfare benchmark.
Once the seller's cost vector is fixed, our core GFT objective can be represented as welfare in an analogous one-sided XOS instance whose item-level values are capped net gains from trade.
This allows us to follow the same high-level mechanism.

The mechanism derives anonymous item markups $Q_j$ from supporting prices for some benchmark allocation and visits the buyers sequentially.
When buyer $i$ is reached and $S$ is the set of remaining items, she is shown the cost-plus-markup prices $c_j+Q_j$ and an entry fee $\delta_i(S)$ calibrated to the median of her residual surplus from $S$.
The markups control competition among buyers and generate profit from the items that sell, while the inventory-dependent entry fees extract part of the surplus that remains after pricing.
The main work is to adopt this one-sided analysis to our original two-sided market.
We show that the fixed-cost core fits the one-sided XOS framework, that the surplus used to set the entry fees is attainable under the buyers' true valuations at the proposed prices, and that the losses in the analysis are controlled by the fixed-cost unit-demand matching benchmark
\[
\gft^{\mathrm{ud}}_{\mathcal F^m}(c)
=
\E_t\!\left[
\max_{\substack{M\subseteq[n]\times[m]\\ M\text{ is a matching}}}
\sum_{(i,j)\in M}\bigl(v_i(t_i,\{j\})-c_j\bigr)^+
\right].
\]
We next define the induced fixed-cost instance, the corresponding ASPE parameters, and the resulting mechanism. The proof of the guarantee is deferred to \cref{app:buyercore}.

\paragraph{The induced fixed-cost core instance.}
Fix a seller cost vector $c$.
For every buyer $i$, define the induced core valuation
\[
v_i^C(t_i,S)=\max_k\sum_{j\in S}(w_{ij}^k)^C.
\]
Because $c$ is fixed and each coefficient $(w_{ij}^k)^C$ depends only on the item-specific type $t_{ij}$, the induced valuation remains XOS over independent items.
By the definition of $\gft^C(t,c)$,
\[
\mu(c)
=
\E_t\!\left[
\max_{(S_1,\dots,S_n)}
\sum_{i=1}^n v_i^C(t_i,S_i)
\right].
\]

Let
\[
\sigma^C(t)=\bigl(S_1^C(t),\dots,S_n^C(t)\bigr)
\]
be an allocation rule maximizing the induced core-welfare objective, and define
\[
q_{ij}=\Pr_t\!\bigl[j\in S_i^C(t)\bigr].
\]
Feasibility of $\sigma^C$ implies that $\sum_i q_{ij}\le 1$ for every item $j$.
For each buyer-item pair $(i,j)$, let
\[
V_{ij}(t_{ij})
=
v_i^C(t_i,\{j\})
=
\min\!\left\{
\bigl(v_i(t_i,\{j\})-c_j\bigr)^+,\tau
\right\}.
\]

\paragraph{From the core benchmark to posted prices.}
The posted prices are derived from a truncated version of the induced core instance.
We first define
\[
\beta_{ij}
=
\inf\Bigl\{
x\ge 0:
\Pr_{t_{ij}}\!\bigl[V_{ij}(t_{ij})\ge x\bigr]
\le \tfrac14 q_{ij}
\Bigr\}.
\]
The cutoff $\beta_{ij}$ ties the probability of an unusually large singleton surplus to the benchmark allocation probability $q_{ij}$.
Together with $\sum_iq_{ij}\le1$, this ensures that competition generated by such large realizations is limited for every item.\footnote{If a distribution has an atom at one of the thresholds used in this section, we adopt the standard threshold-splitting convention, as in~\cite{CZ17}, so that the corresponding tail-probability constraint can be treated as holding with equality.}

For every buyer $i$, define
\[
\chi_i
=
\inf\Bigl\{
x\ge 0:
\sum_{j=1}^m
\Pr_{t_{ij}}\!\bigl[
V_{ij}(t_{ij})\ge \beta_{ij}+x
\bigr]
\le \frac12
\Bigr\},
\]
and let
\[
C_i(t_i)
=
\bigl\{
j\in[m]:
V_{ij}(t_{ij})<\beta_{ij}+\chi_i
\bigr\}.
\]
The threshold $\chi_i$ removes the remaining realizations in which buyer $i$'s induced surplus may be dominated by an exceptional singleton value.
The contribution discarded by this truncation will be charged to the unit-demand matching benchmark.
Define the resulting valuation
\[
v_i^0(t_i,S)
=
v_i^C\bigl(t_i,S\cap C_i(t_i)\bigr).
\]

Since $v_i^0(t_i,\cdot)$ is XOS, it admits exact supporting prices.
For every $t_i$ and every set $S\subseteq[m]$, let
\[
\{\gamma_j^S(t_i)\}_{j\in S}
\]
be exact supporting prices for $v_i^0(t_i,S)$, so that
\[
v_i^0(t_i,S)
=
\sum_{j\in S}\gamma_j^S(t_i)
\]
and
\[
v_i^0(t_i,T)
\ge
\sum_{j\in T}\gamma_j^S(t_i)
\qquad
\text{for every }T\subseteq S.
\]
We use these supporting prices to define the anonymous markups
\[
Q_j
=
\frac12
\sum_{i=1}^n
\E_t\!\left[
\mathbf 1\!\bigl[j\in S_i^C(t)\bigr]\,
\gamma_j^{S_i^C(t)}(t_i)
\right].
\]
Intuitively, $Q_j$ is the amount that the mechanism seeks to extract from item $j$ beyond reimbursing the seller's transfer cost $c_j$.
Accordingly, the actual posted price of item $j$ will be $c_j+Q_j$, and the seller earns the markup $Q_j$ whenever the item sells.

\paragraph{Entry fees.}
Even after the preceding truncation, a buyer's residual surplus at the posted markups may be too sensitive to the realization of a single item.
We therefore apply one final truncation so that this residual surplus concentrates around its median.
For every buyer $i$, define
\[
\eta_i
=
\inf\Bigl\{
x\ge 0:
\sum_{j=1}^m
\Pr_{t_{ij}}\!\bigl[
V_{ij}(t_{ij})
\ge
\max\{\beta_{ij},Q_j+x\}
\bigr]
\le \frac12
\Bigr\},
\]
and let
\[
Y_i(t_i)
=
\bigl\{
j\in[m]:
V_{ij}(t_{ij})<Q_j+\eta_i
\bigr\}.
\]
Define the final truncated valuation
\[
\widehat v_i^C(t_i,S)
=
v_i^0\bigl(t_i,S\cap Y_i(t_i)\bigr).
\]

For every available set $S\subseteq[m]$, define the auxiliary residual surplus
\[
\widetilde u_i(t_i,S)
=
\max_{T\subseteq S}
\left(
\widehat v_i^C(t_i,T)-Q(T)
\right),
\qquad
Q(T)=\sum_{j\in T}Q_j,
\]
and set the entry fee to its median:
\[
\delta_i(S)
=
\operatorname{Median}_{t_i}
\!\left[
\widetilde u_i(t_i,S)
\right].
\]
The final truncation ensures that $\widetilde u_i(t_i,S)$ has sufficiently small sensitivity to any one item, allowing its expectation to be compared with its median.

Although $\widetilde u_i$ is defined through the auxiliary truncated valuation, it is attainable in the original market in the following pointwise sense:
\[
\widetilde u_i(t_i,S)
\le
\max_{T\subseteq S}
\left(
v_i(t_i,T)-c(T)-Q(T)
\right).
\]
Thus, whenever the auxiliary residual surplus covers the entry fee, the buyer can also profitably accept the actual cost-plus-markup menu.
In particular, because $\delta_i(S)$ is the median of $\widetilde u_i(t_i,S)$, the actual buyer accepts the entry fee with probability at least one half.

\paragraph{The mechanism.}
Buyers are visited sequentially in the order $1,2,\dots,n$.
When buyer $i$ is reached, let $S$ denote the set of items that remain available.

\begin{definition}[Fixed-cost ASPE]\label{def:fixed_cost_aspe}
Fix a seller cost vector $c$.
The mechanism $\ASPE(c)$ proceeds as follows.
\begin{enumerate}
    \item When buyer $i$ is reached, the mechanism reveals the available set $S$, the item prices
$
    (c_j+Q_j)_{j\in S},
    $
    and the entry fee $\delta_i(S)$.

    \item If buyer $i$ rejects the entry fee, she receives nothing and the mechanism continues to the next buyer.

    \item If buyer $i$ accepts, she may choose any bundle $T\subseteq S$ and pays
    \[
    \delta_i(S)+\sum_{j\in T}(c_j+Q_j).
    \]
    The selected bundle $T$ is removed from the inventory, and the mechanism continues to the next buyer.
\end{enumerate}
\end{definition}

Conditional on the remaining set $S$, the buyer faces a fixed menu: the entry fee and all item prices are determined before she makes her decision and do not depend on her report.
She may always reject the menu, and, if she accepts, she chooses her utility-maximizing bundle.

\begin{proposition}[Fixed-cost core approximation]\label{prop:multi-core-aspe}
Fix a seller cost vector $c$, and let $\Pi(c)$ denote the seller's expected profit under $\ASPE(c)$.
Then:
\begin{enumerate}
    \item $\ASPE(c)$ is buyer-DSIC and buyer ex-post IR;
    \item its expected profit satisfies
    \[
    \Pi(c)
    \ge
    \frac18\,\mu(c)
    -
    \frac{57}{8}\,
    \gft^{\mathrm{ud}}_{\mathcal F^m}(c).
    \]
\end{enumerate}
\end{proposition}

\paragraph{Proof Overview.}
The supporting-price construction makes $\sum_jQ_j$ equal to one half of the welfare retained after the first truncation.
The mechanism then accounts for this quantity according to whether each item sells.
If item $j$ is sold, the seller directly earns the markup $Q_j$.
If it remains unsold, its supporting-price contribution is reflected in the buyers' residual surplus and can be partially extracted through the inventory-dependent entry fees.
The final truncation ensures that these residual surpluses concentrate around their medians, while the pointwise comparison above guarantees that the entry fees defined in the auxiliary instance are feasible for the buyers in the original market.
Finally, the welfare removed by the truncations, together with the remaining losses in the ASPE analysis, is bounded by the fixed-cost unit-demand matching benchmark.
Combining these bounds gives \cref{prop:multi-core-aspe}; the complete proof is deferred to \cref{app:buyercore}.

Because the seller observes $c$ and can always deploy $\ASPE(c)$, the same guarantee lower-bounds the seller's optimal conditional profit.
Averaging over $c$ gives the seller-side core guarantee used in the final assembly.

\section{Approximating the Tail via the Unit-Demand benchmark}\label{subsec:global_tail}

We now turn to the tail component, which captures the surplus above the
truncation threshold $\tau$.
The goal of this section is to reduce this part of the multidimensional GFT
to the unit-demand matching market.
We first show that
\[
\E[\GFT^T(t,c)]
\le
\frac{1}{1-\ln 2}\,
\gft^{*,\mathrm{ud}}_{\mathcal F^m}.
\]
We then use the copies reduction to relate
$\gft^{*,\mathrm{ud}}_{\mathcal F^m}$
to the buyer-side and seller-side profit benchmarks and construct truthful
mechanisms that approximate each of them.

The first step relies on the observation that the tail is both
\emph{itemwise} and \emph{sparse}.
For each item $j$, let $Y_j$ be the largest singleton GFT that can be
generated by trading that item, and let $X_j$ be its excess above $\tau$:
\[
Y_j
=
\max_{i\in[n]}
\bigl(v_i(t_i,\{j\})-c_j\bigr)^+,
\qquad
X_j
=
(Y_j-\tau)^+.
\]
Also, let
$
p_j=\Pr[Y_j>\tau]=\Pr[X_j>0].
$
Because item $j$ can be allocated at most once, its contribution to the
tail is at most $X_j$, and hence the total tail GFT is at most
$\sum_jX_j$.

The choice
$
\tau
=
2\gft^{*,\mathrm{ud}}_{\mathcal F^m}
$
ensures that items with a nonzero tail contribution are rare.
Moreover, the variables $Y_1,\dots,Y_m$ are mutually independent, since
$Y_j$ depends only on the item-$j$ coordinates of the buyers' types and on
$c_j$.
Thus, simultaneous tail events are sufficiently uncommon that the expected
sum $\E[\sum_jX_j]$ is within a constant factor of the expected largest
excess $\E[\max_jX_j]$.
Finally, the largest such excess is bounded by the best singleton trade,
which is feasible under the unit-demand matching constraint.
The next two lemmas formalize this argument.

\begin{lemma}\label{lem:global_tail_pointwise}
For every realization $(t,c)$,
\[
\GFT^T(t,c)\le \sum_{j=1}^m X_j.
\]
\end{lemma}

\begin{proof}
Fix any realization $(t,c)$.
For every buyer $i$, item $j$, and clause $k$, recall that
\[
w_{ij}^k=(v_{ij}^{(k)}(t_{ij})-c_j)^+,
\qquad
(w_{ij}^k)^T=(w_{ij}^k-\tau)^+.
\]
Since
\[
w_{ij}^k\le \max_{i'\in[n]}\max_{k'} w_{i'j}^{k'} = Y_j,
\]
we have
\[
(w_{ij}^k)^T\le (Y_j-\tau)^+=X_j.
\]
Now consider any feasible allocation and any choice of active clauses.
Because each item can be allocated at most once, the total tail contribution of that allocation is at most $\sum_{j=1}^m X_j$.
Maximizing over feasible allocations and clause choices proves the claim.
\end{proof}

\begin{lemma}\label{lem:global_tail_bound}
\[
\E[\GFT^T(t,c)]
\le
\frac{1}{1-\ln 2}\,\gft^{*,\mathrm{ud}}_{\mathcal F^{m}}.
\]
\end{lemma}

\begin{proof}
By \cref{lem:global_tail_pointwise},
\[
\E[\GFT^T(t,c)]
\le
\E\!\left[\sum_{j=1}^m X_j\right].
\]

Next, we bound $\sum_{j=1}^m p_j$.
Since $\tau = 2\gft^{*,\mathrm{ud}}_{\mathcal F^{m}}$, applying Markov's inequality yields
\[
\Pr\!\left[\max_{j\in [m]} Y_j > \tau\right]
\le
\frac{\E[\max_j Y_j]}{\tau}
\le
\frac{\gft^{*,\mathrm{ud}}_{\mathcal F^{m}}}{2\gft^{*,\mathrm{ud}}_{\mathcal F^{m}}}
=
\frac12.
\]

On the other hand, because the variables $Y_j$ are independent, we have
\[
\Pr\!\left[\max_{j\in [m]} Y_j > \tau\right]
=
1-\prod_{j=1}^m (1-p_j)
\ge
1-e^{-\sum_{j=1}^m p_j}.
\]
Combining these bounds gives $1-e^{-\sum_j p_j}\le \frac12$, which simplifies to $\sum_{j=1}^m p_j \le \ln 2$.

We now apply~\cref{lem:tail_expectation_bound} directly to the independent nonnegative random variables $X_1,\dots,X_m$.
Because $\Pr[X_j>0]=p_j$, and we established $\sum_{j=1}^m p_j\le \ln 2$, we have
\[
\E\!\left[\sum_{j=1}^m X_j\right]
\le
\E[\max_j X_j]
+
(\ln 2)\,\E\!\left[\sum_{j=1}^m X_j\right].
\]
Rearranging terms, we obtain
\[
\E\!\left[\sum_{j=1}^m X_j\right]
\le
\frac{1}{1-\ln 2}\,\E[\max_j X_j].
\]

Finally, because $X_j\le Y_j$ pointwise for every $j$, we have
\[
\E[\max_j X_j]\le \E[\max_j Y_j]\le \gft^{*,\mathrm{ud}}_{\mathcal F^{m}}.
\]
Combining this with the previous bounds proves the lemma.
\end{proof}
\newcommand{\SIDE}{\Pi^*_B}
\newcommand{\BO}{\mathrm{BO}_{1/2}}
\newcommand{\wtR}{\widehat R}

\subsection{Approximating the Buyer-Side Copy Benchmark}
\label{sec:buyertail}

We now construct a truthful and budget-balanced mechanism that approximates
the buyer-side profit benchmark
$\Pi_B^*(\mathcal I^{\mathrm{copies}})$
of the unit-demand copies instance.
Together with the seller-side benchmark, this quantity controls the
unit-demand matching GFT through~\cref{cor:copies_profit_approx}.

An approach in single-dimensional two-sided markets is to use a
``buyers' lawyer'': among all mechanisms offered to the seller, choose one
that maximizes the buyers' total expected utility and then charge the buyers
VCG payments.
In our setting, however, this unrestricted approach need not be budget
balanced.
As shown in~\cref{obs:vcg-deficit}, the presence of one buyer can
increase the utility available to the other buyers by changing the mechanism
offered to the seller.
The resulting VCG payment may therefore be negative, forcing the mechanism
to subsidize the buyers.

We retain the VCG approach but restrict the alternatives over which it
optimizes.
We define randomized \emph{procurement menus}.
A realization of such a menu designates each item to at most one buyer and
post a price for that item.\footnote{At an ironed point of the
seller's procurement-payment curve, the desired acceptance probability and
expected payment may require randomizing between two posted prices.}
After observing the realized menu and her costs, the seller chooses, for
each buyer, at most one designated item that maximizes her nonnegative
profit.
The accepted offers form a matching, and the mechanism is truthful and ex-post IR for the seller. 

The restricted procurement menu family is designed to satisfy two additional properties.
First, it is rich enough to implement a constant fraction of the buyer-side
benchmark.
More importantly, the menus separate the buyers: removing buyer $i$ simply removes
the offers designated to her and leaves the menus faced by all other buyers
unchanged.
Thus, a buyer cannot create a positive externality for the others.
When the mechanism chooses the procurement menu maximizing total reported
buyer utility and applies VCG payments, every buyer's VCG charge is therefore
nonnegative, and the resulting mechanism remains weakly budget balanced.

\begin{assumption}\label{ass:standing}
For each buyer $i$, the type space $T_i$ is a compact metric space, and for each item $j$ the map $t_i\mapsto v_{ij}(t_i)$ is continuous, bounded, and nonnegative.

For each item $j$, the seller's cost $c_j$ is drawn independently from a continuous strictly increasing distribution $F_j$ supported on a bounded interval $[0,\bar c_j]$. The buyer profile $t=(t_1,\dots,t_n)$ is independent of the seller cost vector $c=(c_1,\dots,c_m)$. \footnote{The assumptions in this section are chosen for 
clarity rather than maximal generality. On the buyer side, the proof only
requires standard Borel type spaces and measurable, nonnegative singleton
values. The continuity and bounded-support assumptions simplify the
seller-side ironing and existence arguments. The construction can also be
formulated for arbitrary independent seller-cost distributions using
threshold splitting and an appropriate bound on the prices considered in
the procurement-menu optimization.}
\end{assumption}

\begin{proposition}[Buyer-side copy benchmark]
\label{thm:main_tailbuyer}
Under~\cref{ass:standing}, tThere exists a randomized mechanism that is DSIC and ex-post IR for the
seller, BIC and interim IR for the buyers, and ex-post WBB, such that:
    \[
        \E[\GFT]
        \ge
        \frac14\,\Pi_B^*(\mathcal I^{\mathrm{copies}}).
    \]
\end{proposition}

The approximation loses a factor of two in each of two steps.
First, we scale the ex-ante allocation probabilities of the buyer-side
benchmark so that the total probability incident to each buyer is at most
$1/2$.
Convexity of the seller's procurement-payment curve implies that this
half-capacity relaxation retains at least one half of the benchmark.

Second, we implement a solution to this relaxation through a randomized
procurement menu.
Each item is independently designated to a buyer and offered to the seller
at an appropriate posted price.
The half-capacity condition guarantees that, conditional on a particular
designated item being profitable, it is the only profitable item designated
to that buyer with probability at least $1/2$.
On this event, the seller necessarily selects it, so the resulting menu
captures at least one half of the relaxed objective.

Finally, the VCG-selected menu performs at least as well as this candidate
menu.
Removing one buyer leaves every other buyer's menu unchanged, so all VCG
charges are nonnegative.
The VCG step therefore makes the buyers truthful without causing any
additional approximation loss or violating budget balance.
All the missing proofs can be found in \cref{app:buyertail}.

\paragraph{Seller-side procurement curves.}
For each item $j$, define the quantile-payment curve
\[
r_j(q)
=
qF_j^{-1}(q),
\qquad q\in[0,1].
\]
Posting the deterministic procurement price $F_j^{-1}(q)$ makes the price
cover the seller's cost with probability $q$ and produces expected payment
$r_j(q)$.
Let $\widehat R_j$ be the lower convex envelope of $r_j$.
Equivalently, $\widehat R_j(q)$ is the smallest expected payment obtainable
from a randomized posted price whose probability of covering the seller's
cost is $q$.

Let
\[
g_j(q)
=
\widehat R_{j,+}'(q),
\qquad q\in[0,1),
\]
denote the right derivative of $\widehat R_j$.
The seller's ironed virtual cost can be written as
$
\widetilde\psi_j(c)
=
g_j(F_j(c)).
$
Since $\widehat R_j$ is convex, $g_j$ and hence
$\widetilde\psi_j$ are nondecreasing.

We use two standard consequences of this one-dimensional ironing
construction.
First, if an item's probability of being used is nonincreasing in its cost
and its ex-ante use probability is $q$, then its expected ironed virtual
cost is at least $\widehat R_j(q)$; this is the ironing inequality in
\cref{lem:1d-ironing}.
Second, every point
$(q,\widehat R_j(q))$ can be implemented by a lottery over at most two
posted prices, and these prices can be chosen below an appropriate
subgradient of $\widehat R_j$ at $q$; this is
\cref{lem:two-price}.
The first fact connects the copies benchmark to our relaxation, while the
second implements a relaxation solution as an actual procurement menu.

\subsubsection{The half-capacity relaxation}

Fix a buyer profile $t=(t_1,\dots,t_n)$.
For the analysis, define the buyer-side benchmark conditional on $t$ by
\[
\Pi_B^*(t)
=
\E_c\!\left[
\max_{\substack{M\subseteq[n]\times[m]\\M\text{ is a matching}}}
\sum_{(i,j)\in M}
\bigl(v_{ij}(t_i)-\widetilde\psi_j(c_j)\bigr)^+
\right].
\]
Thus,
\[
\Pi_B^*(\mathcal I^{\mathrm{copies}})
=
\E_t[\Pi_B^*(t)].
\]

For a nonnegative matrix
$y=(y_{ij})_{i\in[n],j\in[m]}$, write
$
q_j(y)=\sum_{i=1}^n y_{ij}.
$
We define the half-capacity relaxation
\begin{equation}
\label{eq:bo-half}
\mathrm{REL}_{1/2}(t)
=
\max_{y\ge0}
\left\{
\sum_{i,j}y_{ij}v_{ij}(t_i)
-
\sum_{j=1}^m
\widehat R_j\!\left(q_j(y)\right)
\ \middle|\
\begin{array}{ll}
\displaystyle\sum_jy_{ij}\le\frac12
    &\text{for every buyer }i,\\[0.5em]
\displaystyle\sum_iy_{ij}\le1
    &\text{for every item }j
\end{array}
\right\}.
\end{equation}

The variable $y_{ij}$ should be viewed as the target probability that the
pair $(i,j)$ becomes active in the procurement menu.
The total mass
$q_j(y)=\sum_i y_{ij}$
is therefore the target probability that item $j$ is active, and
$\widehat R_j(q_j(y))$ is the expected procurement payment required to
generate that probability.
The item constraints impose supply feasibility.
The buyer constraints deliberately reduce each buyer's total active
probability from $1$ to $1/2$.
This reduction is what later ensures that an active item faces no competing
active item for the same buyer with constant probability.

The next lemma compares this relaxation with the original buyer-side
benchmark.

\begin{lemma}[The half-capacity relaxation preserves half the benchmark]
\label{thm:bo-lower-bounds-side}
For every buyer profile $t$,
\[
\mathrm{REL}_{1/2}(t)
\ge
\frac12\,\Pi_B^*(t).
\]
\end{lemma}

\begin{proof}
Fix $t$, and for every seller-cost realization $c$, let $M^*(c)$ be the
monotone optimal matching selector from~\cref{lem:selector-monotone}, applied
to the weights
$
\bigl(v_{ij}(t_i)-\widetilde\psi_j(c_j)\bigr)^+.
$
Thus, $M^*(c)$ contains only strictly positive-weight edges, and, holding
$c_{-j}$ fixed, whether item $j$ is used by $M^*(c)$ is nonincreasing in
$c_j$.

Define the half-scaled ex-ante matching marginals
\[
y_{ij}
=
\frac12\Pr_c\!\bigl[(i,j)\in M^*(c)\bigr],
\qquad
q_j
=
\sum_i y_{ij}
=
\frac12\Pr_c\!\bigl[j\text{ is used by }M^*(c)\bigr].
\]
Because $M^*(c)$ is a matching for every $c$,
$
\sum_j y_{ij}\le\frac12$ for every buyer $i$
and 
$\sum_i y_{ij}=q_j\le\frac12\le1$
for every item $j$.
Hence $y$ is feasible for the relaxation~\eqref{eq:bo-half}.
Since every selected edge has strictly positive weight, the positive part
can be removed on the selected matching. Therefore,
\begin{align*}
\Pi_B^*(t)
&=
\E_c\!\left[
\sum_{(i,j)\in M^*(c)}
\bigl(v_{ij}(t_i)-\widetilde\psi_j(c_j)\bigr)
\right]\\
&=
2\sum_{i,j}y_{ij}v_{ij}(t_i)
-
\sum_{j=1}^m
\E_c\!\left[
\widetilde\psi_j(c_j)
\mathbf 1\!\left\{j\text{ is used by }M^*(c)\right\}
\right].
\end{align*}

For each item $j$, the conditional use probability
$
\Pr_{c_{-j}}\!\left[
j\text{ is used by }M^*(c_j,c_{-j})
\right]
$
is nonincreasing by~\cref{lem:selector-monotone}, and its expectation over
$c_j$ is
$
\Pr_c\!\left[j\text{ is used by }M^*(c)\right]
=
2q_j.
$
Applying the one-dimensional ironing inequality
\cref{lem:1d-ironing} therefore gives
\[
\E_c\!\left[
\widetilde\psi_j(c_j)
\mathbf 1\!\left\{j\text{ is used by }M^*(c)\right\}
\right]
\ge
\widehat R_j(2q_j).
\]
Summing over items yields
\[
\Pi_B^*(t)
\le
2\sum_{i,j}y_{ij}v_{ij}(t_i)
-
\sum_j\widehat R_j(2q_j).
\]

Finally, since $q_j\le1/2$, convexity of $\widehat R_j$ and
$\widehat R_j(0)=0$ imply
$
\widehat R_j(q_j)
\le
\frac12\,\widehat R_j(2q_j).
$
Evaluating the relaxation at the feasible point $y$, we obtain
\begin{align*}
\mathrm{REL}_{1/2}(t)
&\ge
\sum_{i,j}y_{ij}v_{ij}(t_i)
-
\sum_j\widehat R_j(q_j)\\
&\ge
\sum_{i,j}y_{ij}v_{ij}(t_i)
-
\frac12\sum_j\widehat R_j(2q_j) \ge
\frac12\,\Pi_B^*(t). \qedhere
\end{align*}
\end{proof}

\subsubsection{Implementing the relaxation with procurement menus}

We next formalize the restricted family of procurement menus over which the VCG mechanism will optimize. A procurement menu specifies a distribution over buyer designations and posted procurement prices. A realization of these random choices produces a deterministic collection of offers to the seller, whose subsequent choices determine the realized matching. 

\begin{definition}[Procurement menu] \label{def:procurement-menu} A procurement menu is a tuple $\ \rho=(x,a,b,\lambda) $ with the following components:
\begin{enumerate}[label=(\arabic*)] \item Let $x_{ij}$ be the probability that item $j$ is designated to buyer $i$.For every buyer-item pair $(i,j)$, $x_{ij}\ge0,$ and $\sum_{i=1}^n x_{ij}\le1 $ for every item $j$. 
\item For every item $j$, $ 0\le a_j\le b_j\le1$ and,$\lambda_j\in[0,1]$. The quantiles $a_j$ and $b_j$, together with $\lambda_j$, determine the randomized posted price for item $j$. 
\item Define $ q_j(\rho) = \lambda_j a_j+(1-\lambda_j)b_j. $ The buyer-side half-capacity constraints require \[ \sum_{j=1}^m x_{ij}q_j(\rho) \le \frac12 \qquad \text{for every buyer }i. \] \end{enumerate} Let $\mathcal P$ denote the family of all procurement menus satisfying these conditions. 
\end{definition}

Given a procurement menu $\rho$, \cref{alg:procurement-menu} executes it and output a matching and record the payments for each agents. We say that item $j$ is \emph{designated} to buyer $i$ when
$E_j=i$.
Designation does not itself allocate the item: it only makes buyer $i$
the unique possible recipient of $j$ in this realization of the menu.
The item trades only if the seller subsequently selects it.

\begin{algorithm}[t]
\caption{Execution of a procurement menu $\rho$}
\label{alg:procurement-menu}

\KwIn{A procurement menu
$\rho=(x,a,b,\lambda)\in\mathcal P$, with
$p_j^-=F_j^{-1}(a_j)$ and $p_j^+=F_j^{-1}(b_j)$ for every item $j$.}

\KwOut{A realized matching $M$, buyer trade payments
$(p_i^{\mathrm{tr}})_{i\in[n]}$, and seller trade payment
$p_S^{\mathrm{tr}}$.}

$M\gets\varnothing$\;
$p_i^{\mathrm{tr}}\gets 0$ for every $i\in[n]$\;
$p_S^{\mathrm{tr}}\gets 0$\;

\tcp*[l]{All label and price draws are mutually independent and
independent of the seller's costs.}

\ForEach{$j\in[m]$}{
    Draw a buyer label $E_j\in\{0,1,\dots,n\}$ according to
$\Pr[E_j=i]=x_{ij}$ for every $i \in [n]$ and $\Pr[E_j=0]=1-\sum_{i=1}^n x_{ij}$.
    \tcp*[r]{$E_j=0$ means that item $j$ is not offered for trade}

    Draw $P_j=p_j^-$ with probability $\lambda_j$, and
    $P_j=p_j^+$ with probability $1-\lambda_j$\;
}

Reveal the realized labels and prices $(E_j,P_j)_{j\in[m]}$ to the seller\;

\BlankLine

\ForEach{$i\in[n]$}{
    $\mathcal C_i
    \gets
    \{j\in[m]:E_j=i\}$
    \tcp*[r]{items designated to buyer $i$}

    $\mathcal A_i
    \gets
    \{j\in\mathcal C_i:P_j-c_j\ge0\}$
    \tcp*[r]{profitable designated items for buyer $i$}

    \If{$\mathcal A_i\neq\varnothing$}{
        Choose $j_i^*
            \in
            \argmax_{j\in\mathcal A_i}(P_j-\hat{c}_j),
        $
        using a fixed tie-breaking rule

        Allocate item $j_i^*$ to buyer $i$ and set
        $M\gets M\cup\{(i,j_i^*)\}$\;

        $p_i^{\mathrm{tr}}\gets P_{j_i^*}$\;

        $p_S^{\mathrm{tr}}
        \gets
        p_S^{\mathrm{tr}}+P_{j_i^*}$\;
    }
}

\Return{$\bigl(M,(p_i^{\mathrm{tr}})_{i\in[n]},
p_S^{\mathrm{tr}}\bigr)$}\;
\end{algorithm}

For buyer $i$, define her expected trade utility under menu $\rho$, before any VCG transfer, by \[ U_i(t_i,\rho) = \E\!\left[ \sum_{(i,j)\in M}v_{ij}(t_i) - p_i^{\mathrm{tr}} \right], \] where  $ \bigl(M,(p_k^{\mathrm{tr}})_{k\in[n]},p_S^{\mathrm{tr}}\bigr) $ is the output of \cref{alg:procurement-menu}, and the expectation is over the seller's costs and the menu's internal randomization. Because $M$ is a matching, the sum contains at most one term for buyer $i$. The next lemma has two roles. First, compactness and continuity ensure that the VCG optimization over $\mathcal P$ admits a measurable maximizing menu. Second, it records that every fixed procurement menu is already truthful and individually rational for the seller and strongly budget balanced with respect to its trade payments.

\begin{lemma}
\label{lem:menu-basic}
\label{lem:continuation-basic}
The family $\mathcal P$ is compact.
For every buyer $i$ and type $t_i$, the map
$
\rho\longmapsto U_i(t_i,\rho)
$
is continuous on $\mathcal P$.
For every fixed $\rho\in\mathcal P$, the map
$
t_i\longmapsto U_i(t_i,\rho)
$
is continuous on $T_i$.

Moreover, every fixed procurement menu has a direct implementation that is
DSIC and ex-post individually rational for the seller, and the trade
payments within the menu are ex-post strongly budget balanced.
\end{lemma}

An arbitrary procurement menu could post a price exceeding the value of the
buyer to whom the item is designated.
For the menu constructed from an optimal relaxation solution, the next
lemma shows that every edge receiving positive fractional mass has buyer value
at least the marginal procurement cost of the corresponding item.
The two-price implementation lemma will then ensure that every posted price
used for that item is no larger than this marginal cost.

This property is important for the approximation analysis:
it allows us to lower-bound buyer utility using only realizations in which
an item is the unique profitable item designated to its buyer, while safely
discarding all other realizations.

\begin{lemma}[Marginal-value condition]
\label{lem:left-slope}
Fix $t$, and let $y^*(t)$ be an optimizer of~\eqref{eq:bo-half}.
For every item $j$, define
$
q_j
=
\sum_i y^*_{ij}(t).
$
If $q_j>0$, let
$
\gamma_j
=
\widehat R_{j,-}'(q_j)
$
be the left derivative of $\widehat R_j$ at $q_j$.
Then
\[
\gamma_j\in\partial\widehat R_j(q_j),
\]
and whenever $y^*_{ij}(t)>0$,
\[
v_{ij}(t_i)\ge\gamma_j.
\]
\end{lemma}

\begin{proof}
For a finite convex function on a closed interval, the left derivative at
every point in $(0,1]$ exists and belongs to the subdifferential at that
point.

Fix $(i,j)$ with $y^*_{ij}(t)>0$.
For every
$0<\varepsilon\le y^*_{ij}(t)$,
the point
$
y^*(t)-\varepsilon e_{ij}
$
is feasible for~\eqref{eq:bo-half}.
Optimality of $y^*(t)$ therefore implies
\[
0
\ge
-\varepsilon v_{ij}(t_i)
+
\widehat R_j(q_j)
-
\widehat R_j(q_j-\varepsilon) \quad \Rightarrow \qquad v_{ij}(t_i)
\ge
\frac{
\widehat R_j(q_j)-\widehat R_j(q_j-\varepsilon)
}{
\varepsilon
}.
\]
Letting $\varepsilon\downarrow0$ gives
$
v_{ij}(t_i)
\ge
\widehat R_{j,-}'(q_j)
=
\gamma_j.
$
\end{proof}

We now convert the fractional solution $y^*(t)$ into a procurement menu.
For each item $j$, the quantity
$
q_j^*=\sum_i y^*_{ij}(t)
$
is implemented as the probability that the posted price for item $j$
covers the seller's cost.
Conditional on $q_j^*>0$, we set
\[
x_{ij}=\frac{y^*_{ij}(t)}{q_j^*},
\]
so that $x_{ij}$ is the probability that item $j$ is designated to buyer
$i$.
The resulting identity
$
y^*_{ij}(t)=x_{ij}q_j^*
$
is exact: the target probability for the pair $(i,j)$ is decomposed into
the probability of designation and the probability that the item is
profitable for the seller.

The two-price construction implements both the target probability $q_j^*$
and its associated expected procurement payment exactly, so this
factorization introduces no approximation loss.
Within this implementation step, the only loss comes from the possibility
that several items are simultaneously profitable for the same buyer.
The half-capacity constraint guarantees that, conditional on a particular
item being profitable and designated to buyer $i$, it is the unique
profitable item in buyer $i$'s cluster with probability at least $1/2$.
On this event, \cref{alg:procurement-menu} necessarily allocates the item.

\begin{lemma}[Implementing a half-capacity solution]
\label{lem:menu-implementation}
\label{lem:candidate-outcome}
Fix a buyer profile $t$, let $y^*(t)$ optimize~\eqref{eq:bo-half}, and define $q_j= \sum_i y^*_{ij}(t).$
For each item with $q_j>0$, choose $\gamma_j$ as in
\cref{lem:left-slope}, and apply~\cref{lem:two-price} to
$(q_j,\gamma_j)$.
Let $(a_j,b_j,\lambda_j)$ be the resulting two-price parameters.
For every item with $q_j=0$, set
$
(a_j,b_j,\lambda_j)=(0,0,1).
$
Finally, define
\[
x_{ij}
=
\begin{cases}
\dfrac{y^*_{ij}(t)}{q_j},&q_j>0,\\[0.8em]
0,&q_j=0.
\end{cases}
\]

Let
$
\rho_t=(x,a,b,\lambda).
$
Then $\rho_t\in\mathcal P$, and
\[
\sum_{i=1}^nU_i(t_i,\rho_t)
\ge
\frac12\,\mathrm{REL}_{1/2}(t).
\]
\end{lemma}

\begin{proof}
For every buyer $i$,
\[
\sum_jx_{ij}q_j
=
\sum_jy^*_{ij}(t)
\le
\frac12.
\]
For every item $j$ with $q_j>0$,
$
\sum_i x_{ij}
=
1.
$
Thus $\rho_t\in\mathcal P$.

Fix a buyer $i$.
For every item $j$, define the active-edge event
$
E_{ij}
=
\{E_j=i\}\cap\{c_j<P_j\}.
$
Because designations, price draws, and costs are independent across items,
the events $\{E_{ij}\}_{j=1}^m$ are independent.
Since the seller-cost distributions are atomless, strict and weak
profitability have the same probability.
The two-price construction therefore gives
\[
\Pr[E_{ij}]
=
x_{ij}\Pr[c_j<P_j]
=
x_{ij}q_j
=
y^*_{ij}(t).
\]

Suppose $y^*_{ij}(t)>0$.
By~\cref{lem:left-slope},
$
v_{ij}(t_i)\ge\gamma_j.
$
By~\cref{lem:two-price}, both possible prices for item $j$ are at most
$\gamma_j$.
Hence
$
P_j\le v_{ij}(t_i)
$
whenever item $j$ is designated to buyer $i$ under the constructed menu.
In particular, every such trade gives the buyer nonnegative utility.

If $E_{ij}$ occurs and all other events $E_{ik}$ fail, then item $j$ is the
unique strictly profitable item designated to buyer $i$.
The seller must therefore select it.
Consequently,
\begin{align*}
U_i(t_i,\rho_t)
&\ge
\sum_j
\E\!\left[
\bigl(v_{ij}(t_i)-P_j\bigr)
\mathbf 1\{E_{ij}\}
\prod_{k\ne j}\mathbf 1\{E_{ik}^c\}
\right]\\
&=
\sum_j
\E\!\left[
\bigl(v_{ij}(t_i)-P_j\bigr)
\mathbf 1\{E_{ij}\}
\right]
\prod_{k\ne j}
\bigl(1-y^*_{ik}(t)\bigr),
\end{align*}
where the equality uses independence across items.
For $q_j>0$,
\begin{align*}
\E\!\left[
\bigl(v_{ij}(t_i)-P_j\bigr)
\mathbf 1\{E_{ij}\}
\right]
&=
x_{ij}
\E\!\left[
\bigl(v_{ij}(t_i)-P_j\bigr)
\mathbf 1\{c_j<P_j\}
\right]\\
&=
x_{ij}
\left(
q_jv_{ij}(t_i)
-
\E\!\left[
P_j\mathbf 1\{c_j<P_j\}
\right]
\right).
\end{align*}
The two-price construction implements the ironed payment curve, so
\[
\E\!\left[
P_j\mathbf 1\{c_j<P_j\}
\right]
=
\widehat R_j(q_j).
\]
When $q_j=0$, both sides below are zero.
Thus, for every $j$,
\[
\E\!\left[
\bigl(v_{ij}(t_i)-P_j\bigr)
\mathbf 1\{E_{ij}\}
\right]
=
x_{ij}
\bigl(
q_jv_{ij}(t_i)-\widehat R_j(q_j)
\bigr).
\]

The buyer-side capacity constraint gives
$
\sum_k y^*_{ik}(t)\le\frac12.
$
Therefore,
\[
\prod_{k\ne j}
\bigl(1-y^*_{ik}(t)\bigr)
\ge
1-\sum_{k\ne j}y^*_{ik}(t)
\ge
\frac12.
\]
Since the edge-surplus terms are nonnegative, we obtain
\[
U_i(t_i,\rho_t)
\ge
\frac12
\sum_j
x_{ij}
\bigl(
q_jv_{ij}(t_i)-\widehat R_j(q_j)
\bigr).
\]

Summing over buyers and using
$\sum_i x_{ij}=1$
whenever $q_j>0$ gives
\begin{align*}
\sum_iU_i(t_i,\rho_t)
&\ge
\frac12
\sum_{i,j}
x_{ij}
\bigl(
q_jv_{ij}(t_i)-\widehat R_j(q_j)
\bigr)\\
&=
\frac12
\left(
\sum_{i,j}
y^*_{ij}(t)v_{ij}(t_i)
-
\sum_j\widehat R_j(q_j)
\right)\\
&=
\frac12\,\mathrm{REL}_{1/2}(t). \qedhere
\end{align*}
\end{proof}

\subsubsection{VCG over procurement menus}

The procurement menu $\rho_t$ constructed in
\cref{lem:menu-implementation} depends on the buyers' true types and
therefore does not by itself provide truthful buyer incentives.
We address this by selecting the procurement menu that maximizes the
buyers' reported total expected utility and applying a VCG payment rule
\cite{vickrey1961,clarke1971,groves1973}.

Recall that $U_i(t_i,\rho)$ is buyer $i$'s expected utility from the
trade generated by procurement menu $\rho$, before any VCG charge.
The expectation is over the seller's costs and the internal randomization
of the procurement menu, assuming that the seller follows her truthful,
profit-maximizing action.

By~\cref{lem:menu-basic} and the measurable maximum theorem, we may fix a
measurable maximizer for every buyer-report profile.

\begin{definition}[VCG over procurement menus]
\label{def:vcg-procurement-menu}
For a buyer-report profile $r=(r_1,\dots,r_n)$, let
$$
\rho^*(r)
\in
\arg\max_{\rho\in\mathcal P}
\sum_{k=1}^n U_k(r_k,\rho).
$$
For each buyer $i$, define
$$
\mathcal P_{-i}
=
\left\{
\rho=(x,a,b,\lambda)\in\mathcal P:
x_{ij}=0\text{ for every }j
\right\}
\quad \text{and} \quad
H_i(r_{-i})
=
\max_{\rho\in\mathcal P_{-i}}
\sum_{k\ne i}U_k(r_k,\rho).
$$
The VCG charge to buyer $i$ is
$
T_i(r)
=
H_i(r_{-i})
-
\sum_{k\ne i}U_k(r_k,\rho^*(r)).
$ The mechanism executes
\cref{alg:procurement-menu} with menu $\rho^*(r)$ and the seller's reported
cost vector $\widehat c$.
Let
$
\left(
M,\,
(p_i^{\mathrm{tr}})_{i\in[n]},\,
p_S^{\mathrm{tr}}
\right)
$
be the output of the algorithm.
The payments are
$
p_i^{b}
=
p_i^{\mathrm{tr}}+T_i(r)$ for buyer $i$,
and
$
p^{s}=p_S^{\mathrm{tr}}
$
for the seller.
\end{definition}

The key property of the procurement-menu family is that removing one buyer
does not change the trade opportunities faced by any other buyer.

\begin{lemma}[Buyer-exclusion invariance]
\label{lem:buyer-exclusion}
Fix a procurement menu
$
\rho=(x,a,b,\lambda)\in\mathcal P
$
and a buyer $i$.
There exists a menu $\rho^{-i}\in\mathcal P_{-i}$ such that, for every
buyer $k\ne i$ and every type $t_k$,
\[
U_k(t_k,\rho^{-i})
=
U_k(t_k,\rho).
\]
\end{lemma}

\begin{proof}
Define
$
\rho^{-i}
=
(x^{-i},a,b,\lambda),
$
where $x^{-i}_{kj} = 0$ if $k =i$ and $x^{-i}_{kj} = x_{kj}$ otherwise.
Thus, the designation probability previously assigned to buyer $i$ is
transferred to the null label $E_j=0$, while all other designation
probabilities and all price distributions remain unchanged.
Deleting a row can only relax the half-capacity constraints, so
$\rho^{-i}\in\mathcal P_{-i}$.

Couple the two executions using the same seller costs and price draws.
Whenever $E_j=i$ under $\rho$, set $E_j=0$ under $\rho^{-i}$; leave every
other buyer label unchanged.
For each buyer $k\ne i$, the designated set
\[
\mathcal C_k=\{j:E_j=k\},
\]
the associated posted prices, and the seller's choices from this set are
therefore identical in the two executions.
Hence buyer $k$ receives the same item and makes the same trade payment,
which proves
\[
U_k(t_k,\rho^{-i})
=
U_k(t_k,\rho). \qedhere
\]
\end{proof}

Applying \cref{lem:buyer-exclusion} to the selected menu $\rho^*(r)$ gives
\[
H_i(r_{-i})
\ge
\sum_{k\ne i}U_k(r_k,\rho^*(r)).
\]
Consequently,
\begin{equation}
\label{eq:vcg-charge-nonnegative}
T_i(r)\ge0
\qquad
\text{for every buyer $i$ and report profile $r$}.
\end{equation}

\begin{lemma}
\label{lem:vcg-menu-properties}
The mechanism in \cref{def:vcg-procurement-menu} is BIC, interim IR for buyers, DSIC, ex-post IR for the seller and ex-post WBB. In addition, for every true buyer profile $t$
    \[
    \E[\GFT\mid t]
    \ge
    \sum_{i=1}^n U_i(t_i,\rho^*(t)).
    \]
\end{lemma}

\begin{proof}
Fix a buyer $i$, her true type $t_i$, and reports $r_{-i}$ of the other
buyers.
If buyer $i$ reports $r_i$, her expected utility after the VCG charge is
\begin{align*}
&U_i\bigl(t_i,\rho^*(r_i,r_{-i})\bigr)
-
T_i(r_i,r_{-i})\\
&\qquad=
U_i\bigl(t_i,\rho^*(r_i,r_{-i})\bigr)
+
\sum_{k\ne i}
U_k\bigl(r_k,\rho^*(r_i,r_{-i})\bigr)
-
H_i(r_{-i}).
\end{align*}
The final term is independent of $r_i$.
When $r_i=t_i$, the selected menu $\rho^*(t_i,r_{-i})$ maximizes the sum
of the first two terms over all $\rho\in\mathcal P$.
Truthful reporting therefore maximizes buyer $i$'s expected utility.
% This proves part~(i).

At the truthful report, buyer $i$'s expected utility equals
\[
\max_{\rho\in\mathcal P}
\left\{
U_i(t_i,\rho)
+
\sum_{k\ne i}U_k(r_k,\rho)
\right\}
-
H_i(r_{-i}).
\]
Restricting the maximum to $\mathcal P_{-i}$ and observing that
$
U_i(t_i,\rho)=0
$ for every $\rho\in\mathcal P_{-i}$
we have
\[
U_i(t_i,\rho^*(t_i,r_{-i}))
-
T_i(t_i,r_{-i})
\ge
H_i(r_{-i})-H_i(r_{-i})
=
0.
\]
This proves interim individual rationality.

The selected procurement menu and all VCG charges depend only on the buyer
reports.
The seller therefore faces exactly the direct procurement-menu mechanism
analyzed in~\cref{lem:menu-basic}.
It is DSIC and ex-post individually rational for the seller.

For every realization of the procurement-menu algorithm,
\[
p_S^{\mathrm{tr}}
=
\sum_{i=1}^n p_i^{\mathrm{tr}}.
\]
Together with~\eqref{eq:vcg-charge-nonnegative}, this gives
\begin{align*}
\sum_{i=1}^n p_i^b-p^s
&=
\sum_{i=1}^n
\bigl(p_i^{\mathrm{tr}}+T_i(r)\bigr)
-
p_S^{\mathrm{tr}} =
\sum_{i=1}^nT_i(r)
\ge0.
\end{align*}
Thus the complete mechanism is ex-post weakly budget balanced.

Finally, fix a true buyer profile $t$ and consider one realized execution
of $\rho^*(t)$ under truthful seller reporting.
Let
$
u_i^{\mathrm{tr}}
$
denote buyer $i$'s utility from her item trade before the VCG charge, and
let
\[
u_S^{\mathrm{tr}}
=
p_S^{\mathrm{tr}}
-
\sum_{(i,j)\in M}c_j
\]
denote the seller's utility from the selected trades.
Because the trade payments are strongly budget balanced,
\[
\GFT(M;t,c)
=
\sum_{i=1}^n u_i^{\mathrm{tr}}
+
u_S^{\mathrm{tr}}.
\]
The seller can always reject every item, so
$
u_S^{\mathrm{tr}}\ge0
$
pointwise.
Taking expectation over the seller's costs and the procurement menu's
randomization, and using
$
U_i(t_i,\rho^*(t))
=
\E[u_i^{\mathrm{tr}}\mid t],
$
gives
\[
\E[\GFT\mid t]
\ge
\sum_{i=1}^nU_i(t_i,\rho^*(t)). \qedhere
\]
\end{proof}

\begin{proof}[Proof of~\cref{thm:main_tailbuyer}]
Fix a true buyer profile $t$, and let $\rho_t$ be the procurement menu
constructed in~\cref{lem:menu-implementation}.
Because $\rho^*(t)$ maximizes total buyer utility over $\mathcal P$,
\begin{align*}
\E[\GFT\mid t]
&\ge
\sum_iU_i(t_i,\rho^*(t)) \tag{by \cref{lem:vcg-menu-properties} }\\
&\ge
\sum_iU_i(t_i,\rho_t)\\
&\ge
\frac12\,\mathrm{REL}_{1/2}(t) \tag{by \cref{lem:menu-implementation}}\\
&\ge
\frac14\,\Pi_B^*(t). \tag{by \cref{thm:bo-lower-bounds-side}}
\end{align*}

Taking expectation over $t$ gives
\[
\E[\GFT]
\ge
\frac14\,\E_t[\Pi_B^*(t)]
=
\frac14\,\Pi_B^*(\mathcal I^{\mathrm{copies}}).
\]
The incentive, individual-rationality, and budget-balance properties follow
from~\cref{lem:vcg-menu-properties}.
\end{proof}

\subsection{Approximating the Seller-Side Copy Benchmark} \label{sec:sellertail} 
Having approximated the buyer-side copy benchmark, we now turn to $\Pi_S^*(\mathcal I^{\mathrm{copies}})$. This direction is considerably simpler since there is only one seller.  We propose the following posted-price mechanism. This mechanism is adapted from \cite{CHMS10}, where it was originally designed for revenue maximization. 

For the analysis, let $\Pi_S^*(c)$ denote the seller-side copy benchmark conditional on $c$, so that \[ \Pi_S^*(\mathcal I^{\mathrm{copies}}) = \E_c[\Pi_S^*(c)]. \] Fix a virtual-surplus-maximizing matching rule for this conditional benchmark, and let \[ q_{ij}^*(c) = \Pr_t\!\left[(i,j)\text{ is selected by this matching rule}\right] \] be the resulting ex-ante allocation probability of copy $(i,j)$. Because every realized allocation is a matching, \[ \sum_j q_{ij}^*(c)\le 1 \qquad\text{and}\qquad \sum_i q_{ij}^*(c)\le 1. \] The posted-price mechanism scales each of these probabilities by a factor of three and offers the corresponding quantile price. Buyers are visited in an arbitrary order, and each buyer chooses her utility-maximizing item among those that remain available. The scaling controls both possible sources of contention: several items may be attractive to the same buyer, and several buyers may be interested in the same item. Let $G_{ij}$ denote the CDF of the singleton value $v_{ij}(t_i)$.

\begin{definition}[Sequential Virtual Quantile Mechanism] Let $q_{ij}^*$ denote the optimal allocation probability of item $j$ to agent $i$ in the optimal matching that induces $\Pi^*_S(\mathcal{I}^{\mathrm{copies}})$.
The mechanism, $\cM^{\text{svq}}_S$ is defined by a price for each pair $(i, j)$ and an arbitrary sequence $\sigma$ over the buyers.
\begin{enumerate}
\item  For all $i, j$, set $\theta_{ij}(c) = G_{ij}^{-1}\left(1 - \frac{q_{ij}^*}{3}\right)$\footnote{This displays the construction for regular, atomless value distributions. For a general irregular distribution, we use the standard possibly randomized posted price from~\cite{CHMS10}}.
\item  Buyers arrive according to $\sigma$.
\item  When approached, buyer $i$ is presented with a menu of currently available items $J'_i$. She selects the item $j \in J'_i$ that maximizes her utility $v_{ij} - \theta_{ij}(c)$, provided $v_{ij} \ge \theta_{ij}(c)$.
\end{enumerate}
\end{definition}

We now prove that this simple, order-oblivious mechanism achieves a constant-factor approximation to the seller's optimal expected profit benchmark in the copies instance, $\Pi^*_S(\mathcal{I}^{\text{copies}})$.

\begin{proposition}\label{thm:sellertailmain}
The Sequential Virtual Quantile Mechanism $\cM^{\text{svq}}_S$ is Dominant Strategy Incentive Compatible for Buyers (DSIC-B), Strongly Budget Balanced (SBB), and achieves a $6.75$-approximation to the seller's optimal expected profit in the single-dimensional copies instance; that is:$$\Pi_S(\cM^{\text{svq}}_S) \ge \frac{1}{6.75} \Pi^*_S(\mathcal{I}^{\text{copies}})$$
\end{proposition}

\begin{proof}First, we establish the incentive and budget balance properties. Because $\cM^{\text{svq}}_S$ is a sequential posted-price mechanism, each arriving buyer $i$ faces a menu of available items with fixed prices $\theta_{ij}(c)$ that are independent of her own reported valuation. Since the buyer simply selects the utility-maximizing item from this static menu, truthful reporting (or straightforward utility maximization) is a dominant strategy, making the mechanism DSIC-B. Additionally, every transaction consists of a direct payment from the buyer to the seller equal to the posted price. Thus, the mechanism is trivially SBB. Since buyers has the options to opt-out (which result in $0$ utility), they are ex-post IR.

Next, we bound the approximation ratio. Let $p_{ij}^* = G_{ij}^{-1}(1 - q_{ij}^*)$ be the price corresponding to the optimal allocation probability $q_{ij}^*$ in the single-dimensional benchmark. We first note that the expected profit of the optimal mechanism for any agent is bounded above by selling to that agent at probability $q_{ij}^*$ and charging the corresponding price $p_{ij}^*$. Therefore, the optimal profit benchmark is bounded by:$$\Pi^*_S(\mathcal{I}^{\text{copies}}) \le \sum_{i,j} q_{ij}^* (p_{ij}^* - c_j)$$Because the cumulative distribution function $G_{ij}$ is non-decreasing, demanding a strictly lower acceptance probability $\frac{q_{ij}^*}{3} \le q_{ij}^*$ ensures that our posted price is at least the benchmark price: $\theta_{ij}(c) \ge p_{ij}^*$. Consequently, the profit margin weakly improves:$$\theta_{ij}(c) - c_j \ge p_{ij}^* - c_j$$A transaction between $i$ and $j$ will successfully execute if two independent events hold:\begin{itemize}\item $\mathcal{E}_i$: Buyer $i$ does not desire any other item $k \neq j$.\item $\mathcal{E}_j$: Item $j$ is not desired by any other buyer $\ell \neq i$.\end{itemize}If both occur, item $j$ is available when buyer $i$ arrives, and it is her unique utility-maximizing choice.Feasibility in the partition matroids requires $\sum_k q_{ik}^* \le 1$ and $\sum_\ell q_{\ell j}^* \le 1$.The expected number of other items desired by $i$ is $\sum_{k \neq j} \frac{q_{ik}^*}{3} \le \frac{1}{3}$. By Markov's inequality, the probability $i$ desires another item is $\le \frac{1}{3}$, thus $\Pr(\mathcal{E}_i | D_{ij}=1) \ge \frac{2}{3}$.Similarly, the expected number of other buyers desiring $j$ is $\sum_{\ell \neq i} \frac{q_{\ell j}^*}{3} \le \frac{1}{3}$, thus $\Pr(\mathcal{E}_j) \ge \frac{2}{3}$.Since valuations are independent in $\mathcal{I}^{\text{copies}}$, the probability $(i, j)$ is available and selected is at least:$$\Pr[\mathcal{E}_i \cap \mathcal{E}_j] = \Pr[\mathcal{E}_i] \cdot \Pr[\mathcal{E}_j] \ge \left(\frac{2}{3}\right) \left(\frac{2}{3}\right) = \frac{4}{9}$$The total expected profit is the sum of contributions from all successful trades:$$\Pi_S(M^{\text{svq}}_S) \ge \sum_{i,j} \frac{4}{27} q_{ij}^* (\theta_{ij}(c) - c_j) \ge \frac{4}{27} \sum_{i,j} q_{ij}^* (p_{ij}^* - c_j) \ge \frac{4}{27} \Pi^*_S(\mathcal{I}^{\text{copies}}) \qedhere$$\end{proof}
\section{Final Assembly}\label{sec:final}
We are now ready to combine the core and tail analyses to prove the main result for the multiple-buyer setting. Let $\mathcal I$ denote the original multi-buyer trade instance. We define our final mechanism, $\mathcal M^{\mathrm{final}}$, as a randomized delegation protocol: with probability $p$, we delegate pricing power to the seller; with probability $1-p$, we deploy the randomized buyer-side mechanism, denoted as $\mathcal M_B^{\mathrm{copy}}$, from \cref{thm:main_tailbuyer}. As in the bilateral trade setting, we first state a bound assuming the seller computes their profit-maximizing BIC and IIR mechanism, $\mathcal{M}_S^\star$.

\begingroup
\renewcommand{\thetheorem}{\ref{thm:main}}
\begin{theorem}
\mainresult
\end{theorem}
\addtocounter{theorem}{-1}
\endgroup

% \begin{theorem}\label{thm:main}Consider any Bayesian multidimensional instance with $n$ XOS buyers and one additive
% seller over independent items. the mechanism $\mathcal M^{\mathrm{final}}$ is BIC, interim IR, and ex-ante weakly budget balanced. Moreover, it satisfies:$$\gft(\mathcal M^{\mathrm{final}}) \ge \alpha\,\gft^*,$$where $\alpha \approx 1/1033$.\end{theorem}

\begin{proof}The incentive and budget properties are immediate.
By definition, $\mathcal M_S^\star$ is BIC, interim IR, and ex-ante WBB.
By \cref{thm:main_tailbuyer}, $\mathcal M_B^{\mathrm{copy}}$ is DSIC (hence BIC), every buyer is interim IR, the seller is ex-post IR, and the mechanism is ex-post WBB (hence ex-ante WBB).
Since the coin flip defining $\mathcal M^{\mathrm{final}}$ is independent of the agents' reports, $\mathcal M^{\mathrm{final}}$ inherits BIC, interim IR, and ex-ante WBB.

Because every weakly budget-balanced mechanism has expected GFT at least its expected profit, we lower bound the performance of $\mathcal M^{\mathrm{final}}$ by its expected profit. Let $\eta_S = 1/6.75$ and $\eta_B = 1/4$. By randomizing with $p = \frac{\eta_B}{\eta_S+\eta_B} = 27/43$, we balance the tail contributions of both sides. Let $U = \gft^{*,\mathrm{ud}}_{\mathcal F^m}$ and $\gamma = (1-\ln 2)^{-1}$. By \cref{lem:global_split,lem:global_tail_bound}, the first-best gains from trade is upper bounded by:\begin{equation}\label{eq:master_upper_bound}\gft^* \le \mu + \gamma U.\end{equation}
This inequality dictates our case analysis: either the core term $\mu$ is large, or the matching benchmark $U$ dominates the instance. Let $\theta = 57 + \frac{16\eta_S}{3.15} \approx 57.75$.

\smallskip
\noindent\emph{Case 1: Core-dominated ($\mu \ge \theta U$).}In this case, the core surplus is large enough for the entry-fee mechanism to extract significant profit. By \cref{prop:multi-core-aspe}, the seller's optimal profit is at least $\Pi_S^*(\mathcal I) \ge \frac{1}{8}\mu - \frac{57}{8}U$. The expected profit of our protocol is at least:$$\gft(\mathcal M^{\mathrm{final}}) \ge p \left( \frac{1}{8}\mu - \frac{57}{8}U \right).$$Substituting the case assumption $U \le \mu/\theta$ and the upper bound \cref{eq:master_upper_bound}, we obtain:$$\gft(\mathcal M^{\mathrm{final}}) \ge \frac{p}{8} \left( 1 - \frac{57}{\theta} \right) \mu \ge \frac{p(\theta-57)}{8(\theta+\gamma)} \gft^* = \alpha \gft^*.$$
\smallskip
\noindent\emph{Case 2: Tail-dominated ($\mu < \theta U$).}When the core is small, the matching benchmark $U$ controls the GFT. We lower bound the profit using the tail benchmarks $S = \Pi_S^*(\mathcal I^{\mathrm{copies}})$ and $B = \Pi_B^*(\mathcal I^{\mathrm{copies}})$. By randomizing between the optimal benchmarks, we have:$$\gft(\mathcal M^{\mathrm{final}}) \ge p \eta_S S + (1-p) \eta_B B.$$By our choice of $p$, we have $p\eta_S = (1-p)\eta_B = 4/43$. Applying the benchmark comparison $\frac{1}{2}(S+B) \ge \frac{1}{3.15}U$, we get:$$\gft(\mathcal M^{\mathrm{final}}) \ge \frac{4}{43}(S+B) \ge \frac{8}{43 \cdot 3.15} U.$$Substituting the GFT bound $\gft^* \le (\theta + \gamma)U$ yields:$$\gft(\mathcal M^{\mathrm{final}}) \ge \frac{8}{43 \cdot 3.15 (\theta + \gamma)} \gft^* = \alpha \gft^*.$$In both cases, the protocol guarantees at least an $\alpha$ fraction of the first-best GFT.\end{proof}

\paragraph{AI Disclosure.}
We used ChatGPT and Gemini during the preparation of this manuscript to assist with writing and presentation, and to help brainstorm and formulate some proof approaches and intermediate arguments. These tools materially influenced parts of the technical sections. The authors verified the correctness and originality of all content including references.

\bibliographystyle{alpha}
\bibliography{biblio}

\appendix
\crefalias{section}{appendix}
\crefalias{subsection}{appendix}
\section{Virtual Value Function}\label{app:virtualvalue}
\begin{definition}[Buyer Virtual Value]
    The \emph{Myerson virtual value} $\varphi_i(b_i)$ for a buyer $i$ with value $b_i$ drawn from distribution $D_i$ (with CDF $D_i$ and PDF $d_i$) is $\varphi_i(b_i) = b_i - \frac{1-D_i(b_i)}{d_i(b_i)}$. Let $\tilde{\varphi}_i(b_i)$ be the \emph{ironed virtual value}, which is the monotonized version of $\varphi_i(b_i)$.
\end{definition}

\begin{definition}[Seller Virtual Value]
    The \emph{Myerson virtual value} (or \emph{virtual cost}) $\psi_j(s_j)$ for a seller $j$ with cost $s_j$ drawn from distribution $F_j$ (with CDF $F_j$ and PDF $f_j$) is $\psi_j(s_j) = s_j + \frac{F_j(s_j)}{f_j(s_j)}$. Let $\tilde{\psi}_j(s_j)$ be the \emph{ironed virtual value}, which is the monotonized version of $\psi_j(s_j)$.
\end{definition}

\begin{definition}[Virtual Surplus]
    Given buyer values $\mathbf{v}$ (with virtual values $\boldsymbol{\varphi}(\mathbf{v})$) and seller costs $\mathbf{c}$ (with virtual costs $\boldsymbol{\psi}(\mathbf{c})$), the \emph{virtual surplus} of an allocation $\mathbf{x}$ (represented by a matching $A$) is defined from the perspective of the side whose profit is being maximized:
    \begin{itemize}
        \item for maximizing sellers' total profit (extracting from buyers): $\sum_{(i,j) \in M} (\varphi_i(v_i) - c_j)$;
        \item for maximizing buyers' total profit (extracting from sellers): $\sum_{(i,j) \in M} (v_i - \psi_j(c_j))$.
    \end{itemize}
    We use the ironed virtual values ($\tilde{\varphi}_i, \tilde{\psi}_j$) when the original virtual values are not monotone.
\end{definition}

Myerson's key insight relates the expected profit of any BIC and IR auction to the expected virtual surplus it generates.

\begin{lemma}[Myerson's Profit Equivalence \cite{Myerson81}]
    For any BIC 
    % \Aviad{DO we need DSIC? or is BIC enough?} 
    and IR auction:
    \begin{itemize}
        \item The expected total seller profit is equal to the expected virtual surplus $\E[\sum_{(i,j) \in M} (\tilde{\varphi}_i(b_i) - s_j)]$.
        \item The expected total buyer profit is equal to the expected virtual surplus $\E[\sum_{(i,j) \in M} (b_i - \tilde{\psi}_j(s_j))]$.
    \end{itemize}
    Therefore, an auction maximizes the expected profit for one side if and only if it maximizes the corresponding expected virtual surplus.
\end{lemma}

\section{Missing Proof from \cref{sec:technical}}\label{app:overview}

\begin{proof}[Proof of \cref{obs:vcg-deficit}]
Consider a market with two unit-demand buyers ($n=2$) and an additive seller owning two items ($m=2$). The seller's costs $c_1$ and $c_2$ are drawn independently from a uniform distribution $U[0,1]$. Buyer 1 desires only item 1 with a deterministic value $v_{11} = 1$, and Buyer 2 desires only item 2 with a deterministic value $v_{22} = 1$.

Let $\Pi^*(S)$ denote the maximum expected profit the mechanism can extract from the seller on behalf of a coalition of buyers $S \subseteq \{1, 2\}$.

First, we establish the single-buyer baselines. If the mechanism represents only Buyer 1, the optimal procurement auction posts a single price $p_1$ for item 1. The seller accepts if $c_1 \le p_1$, yielding an expected profit of $p_1(1-p_1)$. This is maximized at $p_1 = 0.5$, giving $\Pi^*(\{1\}) = 0.25$. Symmetrically, the optimal expected profit for Buyer 2 alone is $\Pi^*(\{2\}) = 0.25$.

Next, we lower bound the joint optimal expected profit $\Pi^*(\{1, 2\})$. When representing both buyers, the total value for procuring both items is $v_{11} + v_{22} = 2$. The mechanism can choose to offer a single bundle price $P \in [1, 2]$ for both items. The seller accepts the bundle if $c_1 + c_2 \le P$. Because $c_1, c_2 \sim U[0,1]$, the probability of acceptance for $P \in [1, 2]$ is $1 - \frac{(2-P)^2}{2}$. The expected profit from this pure bundle is:$$\Pi_{\mathrm{bundle}}(P) = (2 - P) \left( 1 - \frac{(2 - P)^2}{2} \right)$$By substituting $x = 2 - P$ and optimizing $x - x^3/2$, the maximum occurs at $x = \sqrt{2/3}$ (or $P \approx 1.184$). Plugging this back in yields an expected profit of $\frac{2}{3}\sqrt{\frac{2}{3}} \approx 0.544$. Because the true optimal joint mechanism must perform at least as well as a simple bundle, we have:$$\Pi^*(\{1, 2\}) \ge \frac{2}{3}\sqrt{\frac{2}{3}} \approx 0.544 > 0.50 = \Pi^*(\{1\}) + \Pi^*(\{2\})$$

Finally, we analyze the mechanism's budget under VCG payments. Let $U_i$ denote the realized gross utility of buyer $i$, and $P_{\mathrm{seller}}$ denote the realized payment made to the seller. By definition, $\Pi^*(\{1, 2\}) = \mathbb{E}[U_1 + U_2 - P_{\mathrm{seller}}]$. To guarantee truthfulness, the VCG payment $T_i$ charged to buyer $i$ ensures their net utility equals their marginal contribution to the system:$$\mathbb{E}[U_i - T_i] = \Pi^*(\{1, 2\}) - \Pi^*(\{ -i \})$$Rearranging this gives $\mathbb{E}[T_i] = \mathbb{E}[U_i] - \Pi^*(\{1, 2\}) + \Pi^*(\{ -i \})$. The mechanism is ex-ante WBB if the expected transfers collected from the buyers exceed the expected payment to the seller. The ex-ante net budget is:\begin{align*}\mathbb{E}[\text{Net Budget}] &= \mathbb{E}[T_1] + \mathbb{E}[T_2] - \mathbb{E}[P_{\mathrm{seller}}] \\
&= \Big( \mathbb{E}[U_1] - \Pi^*({1,2}) + \Pi^*({2}) \Big) + \Big( \mathbb{E}[U_2] - \Pi^*({1,2}) + \Pi^*({1}) \Big) - \mathbb{E}[P_{\mathrm{seller}}] \\
&= \Big( \mathbb{E}[U_1 + U_2 - P_{\mathrm{seller}}] \Big) - 2\Pi^*({1,2}) + \Pi^*({1}) + \Pi^*({2}) \\
&= \Pi^*({1, 2}) - 2\Pi^*({1,2}) + \Pi^*({1}) + \Pi^*({2}) \\
&= \Pi^*({1}) + \Pi^*({2}) - \Pi^*({1, 2})\end{align*}
Applying our bounds to this exact budget equation yields:$$\mathbb{E}[\text{Net Budget}] \le 0.25 + 0.25 - 0.544 = -0.044$$Because the expected net budget is strictly negative, the mechanism runs an ex-ante deficit, proving it fails to be Weakly Budget Balanced.\end{proof}
\section{Missing Proofs from \cref{sec:warmup_xos_bilateral}}\label{app:warmup}

\paragraph{Proof of \cref{lem:tail_expectation_bound}}
\begin{proof}
We begin by establishing a pointwise bound for any realization of the variables. Consider the sum $S = \sum_{i=1}^n X_i$. We can separate the maximum variable from the rest of the sum using an indicator function:
$$ \sum_{i=1}^n X_i \le \max_i X_i + \sum_{i=1}^n X_i \cdot \mathbf{1}_{\{\exists j \neq i \text{ s.t. } X_j > 0\}} $$
This inequality is true for all realizations:
\begin{itemize}
    \item If $X_i = 0$ for all $i$, both sides are $0$.
    \item If exactly one variable $X_k > 0$, then for that variable, the indicator $\mathbf{1}_{\{\exists j \neq k \text{ s.t. } X_j > 0\}} = 0$. The sum evaluates to $0$, and the right-hand side is simply $X_k$, making the two sides equal.
    \item If two or more variables are strictly positive, the summation term on the right evaluates to the full sum $S$. Since $\max_i X_i \ge 0$, the right-hand side is $S + \max_i X_i$, which is trivially greater than or equal to $S$.
\end{itemize}

Taking the expectation of both sides yields:
$$ \mathbb{E}\left[\sum_{i=1}^n X_i\right] \le \mathbb{E}\left[\max_i X_i\right] + \sum_{i=1}^n \mathbb{E}\left[X_i \cdot \mathbf{1}_{\{\exists j \neq i \text{ s.t. } X_j > 0\}}\right] $$

Because the variables are drawn independently, $X_i$ is independent of the event that some other variable $X_j > 0$. Therefore, the expectation of the product splits into the product of expectations:
$$ \mathbb{E}\left[\sum_{i=1}^n X_i\right] \le \mathbb{E}\left[\max_i X_i\right] + \sum_{i=1}^n \Big( \mathbb{E}[X_i] \cdot \Pr[\exists j \neq i \text{ s.t. } X_j > 0] \Big) $$

By the union bound, the probability that at least one other variable is strictly positive is upper-bounded by the sum of their individual probabilities:
$$ \Pr[\exists j \neq i \text{ s.t. } X_j > 0] \le \sum_{j \neq i} \Pr[X_j > 0] $$

Substituting this into our inequality and rearranging the summations, we get:
\begin{align*}
    \mathbb{E}\left[\sum_{i=1}^n X_i\right] &\le \mathbb{E}\left[\max_i X_i\right] + \sum_{i=1}^n \left( \mathbb{E}[X_i] \sum_{j \neq i} \Pr[X_j > 0] \right) \\
    &= \mathbb{E}\left[\max_i X_i\right] + \sum_{j=1}^n \left( \Pr[X_j > 0] \sum_{i \neq j} \mathbb{E}[X_i] \right) \\
    &= \mathbb{E}\left[\max_i X_i\right] + \sum_{j=1}^n \left( \Pr[X_j > 0] \cdot \mathbb{E}\left[\sum_{i \neq j} X_i\right] \right)
\end{align*}

Finally, since all variables are non-negative, $\mathbb{E}[\sum_{i \neq j} X_i] \le \mathbb{E}[\sum_{i=1}^n X_i]$. Factoring this complete expectation out of the summation provides the final bound:
\[ \mathbb{E}\left[\sum_{i=1}^n X_i\right] \le \mathbb{E}\left[\max_i X_i\right] + \left( \sum_{j=1}^n \Pr[X_j > 0] \right) \cdot \mathbb{E}\left[\sum_{i=1}^n X_i\right] \qedhere \]
\end{proof}

\paragraph{Proof of \cref{thm:warmup_tail_buyer_menu}}

\begin{lemma}\label{lem:warmup_tail_buyer_threshold}
For every realized singleton-value vector $v=(v_1,\dots,v_m)$, there exists a common threshold $\gamma(v)\ge 0$ and a possibly randomized procurement-price vector
$\mathbf q(v)=(q_1(v),\dots,q_m(v))$
such that the induced order-oblivious posted-price mechanism in $\mathcal I^{\mathrm{copies}}_{1}$ earns conditional expected buyer profit at least
\[
\frac12 \,\E_c\!\left[\max_{j\in [m]} \bigl(v_j-\tilde\psi_j(c_j)\bigr)^+ \,\middle|\, v\right].
\]
In the regular case one may take
\[
q_j(v)=\sup\{x : v_j-\tilde\psi_j(x)\ge \gamma(v)\}.
\]
In the general ironed case, one again randomizes on the boundary plateau of the ironed virtual-cost function.
\end{lemma}

\begin{proof}
Fix the buyer's singleton-value vector $v$, and for each item $j$ define $Y_j = (v_j-\tilde\psi_j(c_j))^+$.
Because the seller's costs are independent across items, the random variables $Y_1,\dots,Y_m$ are independent and nonnegative.
Samuel-Cahn's prophet inequality \cite{Cahn84} gives a threshold $\gamma(v)$ such that stopping at the first $Y_j$ that exceeds $\gamma(v)$ yields expected reward at least
\[
\frac12 \,\E_c\!\left[\max_{j\in [m]} Y_j \,\middle|\, v\right].
\]

We translate this threshold back into procurement prices.
For regular distributions, the price
\[
q_j(v)=\sup\{x : v_j-\tilde\psi_j(x)\ge \gamma(v)\}
\]
makes seller $j$ accept exactly when $Y_j\ge \gamma(v)$.
For ironed distributions, we randomize on the threshold plateau so that the acceptance probability matches the prophet-threshold rule.

As before, each seller in the copies instance is single-parameter.
Therefore Myerson's procurement identity implies that the expected buyer profit of the resulting posted-price mechanism equals its expected ironed virtual-cost surplus.
Consequently, the expected buyer profit of this order-oblivious posted-price mechanism is at least
\[
\frac12 \,\E_c\!\left[\max_{j\in [m]} \bigl(v_j-\tilde\psi_j(c_j)\bigr)^+ \,\middle|\, v\right].
\qedhere
\]
\end{proof}

\begin{proof}[Proof for \cref{thm:warmup_tail_buyer_menu}]
The proof is the exact analogue of \cref{thm:warmup_tail_seller_menu}, with values and costs interchanged.
For each realized singleton-value vector $v$, \cref{lem:warmup_tail_buyer_threshold} gives a procurement-price vector $\mathbf q(v)$ that is a $2$-approximate order-oblivious posted-price mechanism in the seller copies instance.
The same reduction argument applies verbatim after viewing the original seller as the unit-supply agent who chooses which item, if any, to trade.
In the present warm-up instance, the resulting mechanism is simply a simultaneous procurement menu, and it is DSIC and IR because the seller either rejects all offers or chooses the item that maximizes her utility $q_j(v)-c_j$.
Taking expectation over $v$ yields
\[
\Pi_B(\cM_B^{\mathrm{tail}})
\ge
\frac12\,\Pi_B^*(\mathcal I^{\mathrm{copies}}_{1}).
\qedhere
\]
\end{proof}
\section{Missing Proofs from \cref{sec:buyercore}}\label{app:buyercore}

\paragraph{Relationship to Cai--Zhao.}
After fixing $c$, the valuations $v_i^C$ form a standard one-sided XOS
instance over independent item coordinates. We apply the ASPE analysis in
Section~7.2 of~\cite{CZ17} to this induced instance, with their benchmark
allocation identified with $\sigma^C$. Under this identification, their first
buyer-specific cutoff is our $\chi_i$, their second cutoff is our $\eta_i$,
their final truncated valuation is our $\widehat v_i^C$, and their
residual-surplus function is our $\widetilde u_i$.

We use the parts of their analysis that depend only on this
one-sided structure. Specifically, Lemmas~14--18 of~\cite{CZ17} control the
contribution removed by the first truncation, and Lemmas~19--28 give the
supporting-price, concentration, and entry-fee analysis. The additional work
needed here is to: (i) identify fixed-cost core GFT with welfare in the induced
one-sided instance; (ii) compare the posted-price losses to the
fixed-cost unit-demand matching benchmark; and (iii) show that the residual
surplus used to set the entry fees is attainable in the original market at
prices $c_j+Q_j$. We prove these points below.

\subsection{Reduction to a One-Sided XOS Instance}

\begin{lemma}[Fixed-cost welfare identity]
\label{lem:core-aspe-reduction}
For every realization $t$,
\[
\gft^C(t,c)
=
\max_{\substack{(S_1,\dots,S_n)\\ \mathrm{feasible}}}
\sum_{i=1}^n v_i^C(t_i,S_i).
\]
Consequently,
\[
\mu(c)
=
\E_t\!\left[
\sum_{i=1}^n v_i^C\bigl(t_i,S_i^C(t)\bigr)
\right].
\]
\end{lemma}

\begin{proof}
For every feasible allocation $(S_1,\dots,S_n)$, the capped core contribution
of buyer $i$ is, by definition, $v_i^C(t_i,S_i)$. Maximizing over feasible
allocations gives the first identity. The second follows from the definition
of the welfare-maximizing allocation rule $\sigma^C$.
\end{proof}

The next elementary observation verifies the only properties of the baseline
thresholds used in the Cai--Zhao analysis.

\begin{lemma}[Baseline-threshold properties]
\label{lem:core-aspe-beta}
For every buyer $i$ and item $j$, under the threshold-splitting convention,
\[
\Pr_{t_{ij}}\!\left[V_{ij}(t_{ij})\ge \beta_{ij}\right]
=
\frac14 q_{ij},
\]
and
\[
\sum_{k\ne i}
\Pr_{t_{kj}}\!\left[V_{kj}(t_{kj})\ge \beta_{kj}\right]
\le
\frac14.
\]
\end{lemma}

\begin{proof}
The first identity follows from the definition of $\beta_{ij}$. Summing it
over buyers and using $\sum_k q_{kj}\le1$ gives
\[
\sum_k
\Pr\!\left[V_{kj}\ge\beta_{kj}\right]
=
\frac14\sum_kq_{kj}
\le
\frac14.
\]
Dropping the $i$th term proves the second claim.
\end{proof}

\subsection{Charging the First-Truncation Loss}

For each buyer $i$, let $f_i(t_i)=0$ if
$
\max_{j\in[m]}\{V_{ij}(t_{ij})-\beta_{ij}\}<0.
$
Otherwise, let $f_i(t_i)$ be a maximizer of
$V_{ij}(t_{ij})-\beta_{ij}$, using a fixed tie-breaking rule. Define
\[
F_i(t)
=
\begin{cases}
V_{i f_i(t_i)}\bigl(t_{i f_i(t_i)}\bigr),
& f_i(t_i)>0\text{ and }f_i(t_i)\in S_i^C(t),\\
0,&\text{otherwise},
\end{cases}
\]
and
\[
N_i(t)
=
\begin{cases}
v_i^C\bigl(t_i,S_i^C(t)\setminus\{f_i(t_i)\}\bigr),
& f_i(t_i)>0\text{ and }f_i(t_i)\in S_i^C(t),\\
v_i^C\bigl(t_i,S_i^C(t)\bigr),&\text{otherwise}.
\end{cases}
\]
Let
\[
\operatorname{Fav}(c)=\E_t\!\left[\sum_iF_i(t)\right],
\qquad
\operatorname{NonFav}(c)=\E_t\!\left[\sum_iN_i(t)\right].
\]

\begin{lemma}[The favorite contribution is unit-demand]
\label{lem:core-aspe-favorite}
We have
$
\mu(c)
\le
\operatorname{Fav}(c)+\operatorname{NonFav}(c)
$
and
$
\operatorname{Fav}(c)
\le
\gft^{\mathrm{ud}}_{\mathcal F^m}(c).
$
\end{lemma}

\begin{proof}
If buyer $i$ does not receive her favorite item under $\sigma^C(t)$, then
$F_i(t)+N_i(t)=v_i^C(t_i,S_i^C(t))$. If she does receive it, subadditivity of
$v_i^C$ gives
\[
v_i^C\bigl(t_i,S_i^C(t)\bigr)
\le
F_i(t)+N_i(t).
\]
Summing over buyers and taking expectations proves the first inequality.

For the second, the pairs
\[
M^F(t)
=
\left\{
(i,f_i(t_i)):
 f_i(t_i)>0,\ f_i(t_i)\in S_i^C(t)
\right\}
\]
form a matching. Moreover, for every $(i,j)\in M^F(t)$,
\[
F_i(t)=V_{ij}(t_{ij})
\le
\bigl(v_i(t_i,\{j\})-c_j\bigr)^+.
\]
The pointwise contribution of $M^F(t)$ is therefore at most the optimal
unit-demand matching GFT. Taking expectations proves the claim.
\end{proof}

Define the contribution retained by the first buyer-specific cutoff as
\[
\operatorname{Bounded}(c)
=
\E_t\!\left[
\sum_{i=1}^n
v_i^0\bigl(t_i,S_i^C(t)\bigr)
\right].
\]

We also use the following auxiliary one-sided benchmark. Consider the induced
unit-demand instance in which buyer $i$ has singleton value
$V_{ij}(t_{ij})$ for item $j$. Let $\operatorname{PostRev}(c)$ be the supremum expected revenue of a
rationed sequential posted-price mechanism in this instance: buyers are
visited sequentially, nonnegative prices may depend on the buyer--item pair,
and each buyer may purchase at most one available item.

\begin{lemma}[Cai--Zhao first-truncation bound]
\label{lem:core-aspe-cz-first}
\[
\operatorname{NonFav}(c)
\le
\operatorname{Bounded}(c)
+
\frac83\operatorname{PostRev}(c).
\]
\end{lemma}

\begin{proof}
Apply Lemmas~14--18 of~\cite{CZ17} with $b=1/4$. Their allocation
$\sigma^{(\beta)}$ is our feasible allocation rule $\sigma^C$, their
singleton values are $V_{ij}$, their first buyer-specific cutoff is $\chi_i$,
and their truncated valuation is $v_i^0$. The two threshold hypotheses used
in those lemmas are exactly \cref{lem:core-aspe-beta}. Lemma~14 decomposes the
nonfavorite contribution into the retained contribution above and a
high-singleton overlap term, while Lemmas~15--18 bound that term by
$2/(1-b)=8/3$ times the RSPM benchmark. No property of the dual construction
used elsewhere in~\cite{CZ17} enters these arguments.
\end{proof}

The comparison of this auxiliary revenue benchmark with GFT is specific to
our fixed-cost reduction.

\begin{lemma}[The auxiliary posted-price loss is unit-demand]
\label{lem:core-aspe-postrev-ud}
\[
\operatorname{PostRev}(c)
\le
\gft^{\mathrm{ud}}_{\mathcal F^m}(c).
\]
\end{lemma}

\begin{proof}
Fix any rationed sequential posted-price mechanism in the auxiliary instance
and any realization $t$. Its sold buyer--item pairs form a matching. If buyer
$i$ purchases item $j$ at auxiliary price $p_{ij}$, then
\[
p_{ij}
\le
V_{ij}(t_{ij})
\le
\bigl(v_i(t_i,\{j\})-c_j\bigr)^+.
\]
Thus, pointwise, the mechanism's revenue is at most the singleton GFT of the
same matching, and hence at most the optimal unit-demand matching GFT. Taking
expectations and then the supremum over auxiliary posted-price mechanisms
proves the claim.
\end{proof}

\subsection{Supporting Prices and the Fixed-Cost ASPE}

Because the supporting prices $\gamma_j^S(t_i)$ are exact, the bounded
contribution is represented exactly by the anonymous markups.

\begin{lemma}[Supporting-price identity]
\label{lem:core-aspe-support}
$
\operatorname{Bounded}(c)
=
2\sum_{j=1}^m Q_j.
$
\end{lemma}

\begin{proof}
By the definition of $v_i^0$ and the exact supporting-price property,
\begin{align*}
\operatorname{Bounded}(c)
&=
\sum_{i=1}^n
\E_t\!\left[
 v_i^0\bigl(t_i,S_i^C(t)\bigr)
\right]\\
&=
\sum_{i=1}^n
\E_t\!\left[
 \sum_{j\in S_i^C(t)}
 \gamma_j^{S_i^C(t)}(t_i)
\right]
=
2\sum_{j=1}^mQ_j. \qedhere
\end{align*}
\end{proof}

The next lemma is the main new bridge from the auxiliary one-sided instance
to the original two-sided market.

\begin{lemma}[Feasibility of the residual-surplus proxy]
\label{lem:core-aspe-proxy}
For every buyer $i$, type $t_i$, and available set $S$,
\[
\widetilde u_i(t_i,S)
\le
\max_{T\subseteq S}
\left(v_i(t_i,T)-c(T)-Q(T)
\right).
\]
\end{lemma}

\begin{proof}
Fix $T\subseteq S$ and choose an XOS clause $k$ attaining
$\widehat v_i^C(t_i,T)$. Since $\widehat v_i^C$ is obtained from $v_i^C$ by
deleting item coefficients, we may write
\[
\widehat v_i^C(t_i,T)
=
\sum_{j\in T}\widehat w_{ij}^k(t,c),
\qquad
0\le \widehat w_{ij}^k(t,c)
\le
\bigl(v_{ij}^{(k)}(t_{ij})-c_j\bigr)^+.
\]
Let
$
R
=
\left\{
 j\in T:\widehat w_{ij}^k(t,c)>Q_j
\right\}.
$
Then
\begin{align*}
\widehat v_i^C(t_i,T)-Q(T)
&=
\sum_{j\in T}
\bigl(\widehat w_{ij}^k(t,c)-Q_j\bigr)\\
&\le
\sum_{j\in R}
\bigl(\widehat w_{ij}^k(t,c)-Q_j\bigr)\\
&\le
\sum_{j\in R}
\bigl(v_{ij}^{(k)}(t_{ij})-c_j-Q_j\bigr)\\
&\le
v_i(t_i,R)-c(R)-Q(R)\\
&\le
\max_{T'\subseteq S}
\left(v_i(t_i,T')-c(T')-Q(T')
\right).
\end{align*}
The third line is valid because $j\in R$ implies
$\widehat w_{ij}^k(t,c)>0$ and hence
$v_{ij}^{(k)}(t_{ij})>c_j$. Maximizing the left-hand side over
$T\subseteq S$ proves the lemma.
\end{proof}

\begin{lemma}[Fixed-cost specialization of the Cai--Zhao ASPE bound]
\label{lem:core-aspe-extraction}
The expected seller profit of $\ASPE(c)$ satisfies
\[
\Pi(c)
\ge
\frac14\sum_{j=1}^mQ_j
-
\frac{20}{3}\operatorname{PostRev}(c).
\]
\end{lemma}

\begin{proof}
Apply the supporting-price and entry-fee analysis in Lemmas~19--28 of
\cite{CZ17}, and in particular their Lemma~27, to the induced valuation
$v_i^C$, the benchmark allocation $\sigma^C$, and $b=1/4$. Their second
buyer-specific cutoff $\tau_i$, final truncated valuation $\widehat v_i$,
and residual-surplus function $\mu_i$ are exactly our $\eta_i$,
$\widehat v_i^C$, and $\widetilde u_i$, respectively. The hypotheses of
that analysis hold because $v_i^C$ is XOS over independent items,
\cref{lem:core-aspe-beta} supplies the required baseline-threshold
properties, and the supporting prices are exact (so their supporting-price
parameter is $\alpha=1$).

It remains only to connect their one-sided revenue accounting to our market.
By \cref{lem:core-aspe-proxy}, a buyer's true utility from the
cost-plus-markup menu pointwise dominates the residual-surplus proxy used to
set her entry fee. Hence, with ties at the entry decision resolved in favor of entry, the median
fee is accepted with probability at least $1/2$, exactly as in the proof of
Lemma~27 of~\cite{CZ17}. Moreover, the set
available when buyer $i$ is visited depends only on earlier buyers' types and
is independent of $t_i$. Finally, whenever item $j$ sells at price $c_j+Q_j$,
the seller earns margin exactly $Q_j$. Thus every step of the sold-versus-
unsold accounting in that lemma applies unchanged. Substituting $b=1/4$ into
its bound gives the stated coefficient $20/3$.
\end{proof}

\subsection{Completing the Fixed-Cost Bound}

\begin{proof}[Proof of~\cref{prop:multi-core-aspe}]
Conditional on the remaining set of items, each buyer faces a fixed entry fee
and fixed item prices, all independent of her report. Truthful utility
maximization is therefore a dominant strategy, and rejecting the menu gives
utility zero. Thus $\ASPE(c)$ is buyer-DSIC and buyer ex-post IR.

For the profit guarantee, combine
\cref{lem:core-aspe-favorite,lem:core-aspe-cz-first,lem:core-aspe-postrev-ud}:
\begin{align*}
\mu(c)
\le
\operatorname{Fav}(c)+\operatorname{NonFav}(c)
\le
\gft^{\mathrm{ud}}_{\mathcal F^m}(c)
+
\operatorname{Bounded}(c)
+
\frac83\operatorname{PostRev}(c)
\le
\operatorname{Bounded}(c)
+
\frac{11}{3}\gft^{\mathrm{ud}}_{\mathcal F^m}(c).
\end{align*}
By \cref{lem:core-aspe-support},
\[
\sum_{j=1}^mQ_j
\ge
\frac12\mu(c)
-
\frac{11}{6}\gft^{\mathrm{ud}}_{\mathcal F^m}(c).
\]
Substituting this inequality into \cref{lem:core-aspe-extraction} and using
\cref{lem:core-aspe-postrev-ud} once more gives
\begin{align*}
\Pi(c)
&\ge
\frac14\sum_jQ_j
-
\frac{20}{3}\operatorname{PostRev}(c)\\
&\ge
\frac18\mu(c)
-
\frac{11}{24}\gft^{\mathrm{ud}}_{\mathcal F^m}(c)
-
\frac{20}{3}\gft^{\mathrm{ud}}_{\mathcal F^m}(c)\\
&=
\frac18\mu(c)
-
\frac{57}{8}\gft^{\mathrm{ud}}_{\mathcal F^m}(c). \qedhere
\end{align*}
\end{proof}

\section{Missing Proofs for~\cref{sec:buyertail}}
\label{app:buyertail}

\subsection{Seller-Side Ironing Facts}
\begin{lemma}
\label{lem:convex-envelope}
The function
$
\widehat R:[0,1]\to\R_+
$
is continuous and convex, satisfies $\widehat R(0)=0$, and is absolutely
continuous. Its right derivative
$
g=\widehat R'_+
$
is nondecreasing and Borel measurable on $[0,1)$, and
\[
\widehat R(q)
=
\int_0^q g(u)\,du
\qquad
\text{for every }q\in[0,1].
\]
Moreover, if $q\in(0,1]$ and $\gamma\in\partial\widehat R(q)$, then
\[
\widehat R(q)\le q\gamma.
\]
\end{lemma}

\begin{proof}
Since $r$ is continuous on the compact interval $[0,1]$, its lower convex
envelope $\widehat R$ is finite and convex.
The zero function is convex and lies below $r$, while $r(0)=0$.
Therefore,
\[
\widehat R\ge0
\qquad\text{and}\qquad
\widehat R(0)=0.
\]
The lower convex envelope of a continuous function on a compact interval is
continuous. A finite convex function that is continuous at the endpoints is
absolutely continuous on the interval. Hence its right derivative exists
except possibly at the endpoint, is nondecreasing and integrable, and
\[
\widehat R(q)
=
\widehat R(0)+\int_0^q\widehat R'_+(u)\,du
=
\int_0^q g(u)\,du.
\]

If $\gamma\in\partial\widehat R(q)$, the supporting-line inequality,
evaluated at $0$, gives
\[
0
=
\widehat R(0)
\ge
\widehat R(q)+\gamma(0-q)
=
\widehat R(q)-q\gamma.
\]
Therefore we have 
$
\widehat R(q)\le q\gamma.
$
\end{proof}

\begin{lemma}[Two-price implementation of an ironed point]
\label{lem:two-price}
Fix $q\in(0,1]$ and $\gamma\in\partial\widehat R(q)$.
There exist
$
0\le a\le q\le b\le1$
 and 
$\lambda\in[0,1]$
such that
\[
q=\lambda a+(1-\lambda)b,
\qquad 
\widehat R(q)
=
\lambda r(a)+(1-\lambda)r(b),
\]
and
\[
F^{-1}(a)\le\gamma,
\qquad
F^{-1}(b)\le\gamma.
\]
For $q=0$, the trivial choice
$
a=b=0,$ and $ \lambda=1
$
implements the point $(0,\widehat R(0))$.
\end{lemma}

\begin{proof}
Suppose first that
$
\widehat R(q)=r(q).
$
Set $a=b=q$ and $\lambda=1$.
By~\cref{lem:convex-envelope},
$
qF^{-1}(q)
=
r(q)
=
\widehat R(q)
\le
q\gamma,
$
and therefore $F^{-1}(q)\le\gamma$.

Now suppose that
$
\widehat R(q)<r(q).
$
Because $r$ is continuous and $\widehat R$ is its lower convex envelope,
$q$ lies in the interior of an interval $[a,b]$ on which $\widehat R$ is
affine and whose endpoints are contact points:
\[
\widehat R(a)=r(a),
\qquad
\widehat R(b)=r(b).
\]
Since $q\in(a,b)$, every subgradient of $\widehat R$ at $q$ equals the
common slope of this affine segment. In particular, $\gamma$ is the slope
of the line
\[
L(u)
=
\widehat R(q)+\gamma(u-q).
\]

Choose $\lambda\in[0,1]$ such that
$
q=\lambda a+(1-\lambda)b.
$
Since $\widehat R$ is affine on $[a,b]$,
\[
\widehat R(q)
=
\lambda\widehat R(a)+(1-\lambda)\widehat R(b)
=
\lambda r(a)+(1-\lambda)r(b).
\]

By~\cref{lem:convex-envelope},
$\widehat R(q)-q\gamma\le0,$
so the intercept of $L$ is nonpositive. Hence
\[
r(a)=L(a)\le\gamma a,
\qquad
r(b)=L(b)\le\gamma b.
\]
If $a>0$, this gives
\[
F^{-1}(a)=\frac{r(a)}a\le\gamma.
\]
If $a=0$, then $F^{-1}(a)=0\le\gamma$.
The same argument gives
\[
F^{-1}(b)\le\gamma.
\]
\end{proof}

\begin{lemma}[One-dimensional ironing inequality]
\label{lem:1d-ironing}
Let $c\sim F$, let
$
x:[0,\bar c]\to[0,1]
$
be nonincreasing, and let
$
q=\E[x(c)].
$
Then
\[
\E\!\left[\widetilde\psi(c)x(c)\right]
\ge
\widehat R(q).
\]
\end{lemma}

\begin{proof}
Let
$
U=F(c)\sim\mathrm{Unif}[0,1]
$
and define
$
\widetilde x(u)=x(F^{-1}(u)).
$
The function $\widetilde x:[0,1]\to[0,1]$ is nonincreasing and satisfies
\[
\int_0^1\widetilde x(u)\,du
=
\E[x(c)]
=
q.
\]
Moreover,
\[
\E\!\left[\widetilde\psi(c)x(c)\right]
=
\int_0^1g(u)\widetilde x(u)\,du.
\]

Because $g$ is nondecreasing, among all functions
$\phi:[0,1]\to[0,1]$ with integral $q$, the integral
\[
\int_0^1g(u)\phi(u)\,du
\]
is minimized by placing all mass on the smallest quantiles, namely by
\[
\phi(u)=\mathbf 1\{u\le q\}.
\]
Indeed, moving any amount of mass from a larger quantile to a smaller one
cannot increase the integral against the nondecreasing function $g$.
Therefore,
\[
\E\!\left[\widetilde\psi(c)x(c)\right]
\ge
\int_0^qg(u)\,du
=
\widehat R(q),
\]
where the last equality follows from~\cref{lem:convex-envelope}.
\end{proof}

\subsection{A Monotone Optimal-Matching Selector}

The proof that the half-capacity relaxation preserves a constant fraction
of the buyer-side benchmark uses an optimal matching whose use of each item
is monotone in that item's cost. We next specify a deterministic selector
with this property.

\begin{remark}[A concrete monotone selector]
\label{rem:selector}
Fix a buyer profile $t$ and define the edge weights
\[
w_{ij}(c)
=
\bigl(v_{ij}(t_i)-\widetilde\psi_j(c_j)\bigr)^+.
\]
Among all maximum-weight matchings for these weights, select a matching
according to the following deterministic tie-breaking rule:
\begin{enumerate}[label=(\roman*)]
    \item minimize the cardinality of the matching;
    \item subject to that, maximize
    \[
    B(M)
    =
    \sum_{j=1}^m
    2^{-j}\,
    \mathbf 1\{\text{item $j$ is used by }M\};
    \]
    \item subject to both preceding criteria, use an arbitrary fixed
    lexicographic order over matchings.
\end{enumerate}
Let $M^*(c)$ denote the selected matching.

The final lexicographic rule only makes the selector single-valued.
The first two rules are the ones used in the monotonicity argument below.
\end{remark}

\begin{lemma}[Properties of the selector]
\label{lem:selector-monotone}
Fix $t$ and let $M^*(c)$ be the selector from~\cref{rem:selector}.
Then:
\begin{enumerate}[label=(\alph*)]
    \item every edge in $M^*(c)$ has strictly positive weight;
    \item for every item $j$ and every fixed $c_{-j}$, the indicator
    \[
    x_j(c_j,c_{-j})
    =
    \mathbf 1\{\text{item $j$ is used by }M^*(c_j,c_{-j})\}
    \]
    is nonincreasing in $c_j$.
\end{enumerate}
\end{lemma}

\begin{proof}
For part~(a), suppose that a selected edge had weight zero.
Removing that edge would preserve the total weight and strictly reduce the
cardinality of the matching, contradicting the first tie-breaking rule.

For part~(b), fix $j$ and $c_{-j}$.
Because $\widetilde\psi_j$ is nondecreasing, every edge weight incident to
item $j$ is nonincreasing in $c_j$.
Let
\[
W_{\mathrm{use}}(c_j)
\]
be the maximum total weight among matchings that use item $j$, and let
\[
W_{\mathrm{omit}}
\]
be the maximum total weight among matchings that omit item $j$.
Then $W_{\mathrm{use}}(c_j)$ is nonincreasing in $c_j$, whereas
$W_{\mathrm{omit}}$ does not depend on $c_j$.

Suppose that the selected matching uses item $j$ at a higher cost
$c_j^{\mathrm{hi}}$. Then
\[
W_{\mathrm{use}}(c_j^{\mathrm{hi}})
\ge
W_{\mathrm{omit}}.
\]
For every lower cost $c_j^{\mathrm{lo}}<c_j^{\mathrm{hi}}$,
\[
W_{\mathrm{use}}(c_j^{\mathrm{lo}})
\ge
W_{\mathrm{use}}(c_j^{\mathrm{hi}})
\ge
W_{\mathrm{omit}}.
\]
If either inequality is strict, every maximum-weight matching at
$c_j^{\mathrm{lo}}$ uses item $j$, and the claim follows.

It remains to consider the equality case
\[
W_{\mathrm{use}}(c_j^{\mathrm{lo}})
=
W_{\mathrm{use}}(c_j^{\mathrm{hi}})
=
W_{\mathrm{omit}}.
\]
Let $M^{\mathrm{hi}}$ be the selected matching at
$c_j^{\mathrm{hi}}$.
By part~(a), its edge incident to item $j$ has strictly positive weight.
Lowering $c_j$ can only increase that edge's weight.
If the weight of $M^{\mathrm{hi}}$ increased strictly, then
\[
W_{\mathrm{use}}(c_j^{\mathrm{lo}})
>
W_{\mathrm{use}}(c_j^{\mathrm{hi}}),
\]
contradicting the equality above.
Thus, $M^{\mathrm{hi}}$ remains a maximum-weight matching at
$c_j^{\mathrm{lo}}$.

The family of matchings omitting $j$, including all of their weights and
tie-breaking priorities, is unchanged when $c_j$ decreases.
Meanwhile, $M^{\mathrm{hi}}$ remains available on the use-$j$ side with
the same tie-breaking priority that caused the use-$j$ side to be selected
at $c_j^{\mathrm{hi}}$.
Any additional maximum-weight matchings that appear after lowering $c_j$
can only improve the best tie-breaking priority on the use-$j$ side.
Therefore, the selected matching at $c_j^{\mathrm{lo}}$ also uses item
$j$.
\end{proof}

\subsection{Proof of~\cref{lem:menu-basic}}
\label{proof:menu-basic}

\begin{lemma}[Restatement of~\cref{lem:menu-basic}]
The procurement-menu family $\mathcal P$ is compact.
For every buyer $i$ and type $t_i$, the map
$
\rho\longmapsto U_i(t_i,\rho)
$
is continuous on $\mathcal P$.
For every fixed procurement menu $\rho\in\mathcal P$, the map
$
t_i\longmapsto U_i(t_i,\rho)
$
is continuous on $T_i$.

Moreover, every fixed procurement menu admits an equivalent direct
implementation that is DSIC and ex-post individually rational for the
seller, and the trade payments within the menu are ex-post strongly budget
balanced.
\end{lemma}

\begin{proof}
A procurement menu is parametrized by
$
\rho=(x,a,b,\lambda).
$
All coordinates lie in a finite-dimensional bounded box.
The constraints
\[
x_{ij}\ge0,
\qquad
\sum_i x_{ij}\le1,
\qquad
0\le a_j\le b_j\le1,
\qquad
\lambda_j\in[0,1],
\]
and
\[
\sum_jx_{ij}
\bigl(\lambda_ja_j+(1-\lambda_j)b_j\bigr)
\le\frac12
\]
are closed.
Therefore, $\mathcal P$ is a closed subset of a compact box and is itself
compact.

Fix a buyer $i$ and a type $t_i$.
There are finitely many possible designation vectors and finitely many
possible choices of the lower or upper posted price for each item.
Fix one such designation vector and one such price-branch vector.
The probability of this realization is a continuous function of
$(x,\lambda)$.

Conditional on the fixed realization, every posted price is a continuous
function of $(a,b)$ because every $F_j^{-1}$ is continuous under
\cref{ass:standing}.
The seller's selected item in a buyer cluster is determined by comparisons
of the form
\[
P_j-c_j\ge P_k-c_k
\qquad\text{and}\qquad
P_j-c_j\ge0.
\]
For the limiting menu, the set of cost vectors on which any such comparison
is an equality has probability zero: the cost distributions are atomless
and independent.
Outside this null set, all relevant comparisons are strict, so the selected
items and payments are unchanged under sufficiently small perturbations of
the menu parameters.
The corresponding buyer payoff therefore converges pointwise almost surely.

Buyer values and posted prices are uniformly bounded under
\cref{ass:standing}.
Dominated convergence implies that the buyer's expected payoff,
conditional on the fixed designation and price-branch vectors, is
continuous in the menu parameters.
Since there are only finitely many such realizations and their probabilities
are themselves continuous, the map
\[
\rho\longmapsto U_i(t_i,\rho)
\]
is continuous.

Now fix a procurement menu $\rho$.
The random allocation and trade-payment rule induced by $\rho$ does not
depend on buyer $i$'s type.
Hence there exist trade probabilities
\[
\pi_{ij}(\rho)\in[0,1]
\]
and an expected trade payment $m_i(\rho)$, both independent of $t_i$, such
that
\[
U_i(t_i,\rho)
=
\sum_{j=1}^m
\pi_{ij}(\rho)v_{ij}(t_i)
-
m_i(\rho).
\]
Since each map $t_i\mapsto v_{ij}(t_i)$ is continuous, so is
$t_i\mapsto U_i(t_i,\rho)$.

It remains to verify the seller-side properties.
Fix a realization of the menu's designations and posted prices.
An equivalent direct implementation asks the seller to report a cost
vector $\widehat c$.
For each buyer cluster, the mechanism selects at most one item, choosing
an item that maximizes
\[
P_j-\widehat c_j
\]
whenever this maximum is nonnegative, and selecting no item otherwise.
Ties are resolved according to a fixed rule.

Under a truthful report $\widehat c=c$, this rule selects, in every cluster,
the option maximizing the seller's true profit.
Because the buyer clusters are disjoint and the seller's costs are
additive, it also maximizes her total true utility over all feasible
selections.
Any misreport can only induce another feasible selection and therefore
cannot increase her true utility.
Truthful reporting is consequently a dominant strategy.
The seller may always induce the empty selection and receives nonnegative
utility under truthful reporting, establishing ex-post individual
rationality.

Finally, whenever item $j$ trades, its designated buyer pays $P_j$ and the
seller receives exactly $P_j$.
Thus, for every realization of the menu, reports, and costs, the total
trade payment collected from the buyers equals the payment made to the
seller.
The trade component of every procurement menu is therefore ex-post
strongly budget balanced.
\end{proof}

\end{document}